\documentclass[12pt]{article}
\usepackage[utf8]{inputenc}
\usepackage{color}
\usepackage{fullpage}
\usepackage{amsmath}
\usepackage{amssymb, amsfonts}
\usepackage{amsthm}
\usepackage{siunitx}
\usepackage{booktabs}
\newtheorem{assumption}{Assumption}
\newcommand{\ind}[1]{\mathbb{I}(#1)}

\usepackage{multirow}
\usepackage{bm}
\usepackage{mathrsfs}
\usepackage{xr}
\usepackage{tablefootnote}
\usepackage{tabularx}
\usepackage{amsfonts}
\usepackage{graphicx}
\usepackage{caption}
\usepackage{subcaption}
\usepackage{float}
\usepackage{natbib}

\usepackage{geometry}
\newtheorem{pro}{Proposition}
\newtheorem{theorem}{Theorem}

\newtheorem{lemma}{Lemma}
\numberwithin{equation}{section}

\def\cp{\mathop{\rightarrow}\limits^{p}}
\def\cd{\mathop{\rightarrow}\limits^{d}}

\newtheorem{remark}{Remark}

\allowdisplaybreaks[4]
\def\pr{\mathbb{P}}
\def\E{\mathbb{E}}
\def\var{\mathrm{Var}}
\def\cov{\mathrm{Cov}}
\def\X{{\boldsymbol X}}

\def\S{{\boldsymbol S}}
\def\W{{\bm {W}}}

\def\tilde{\widetilde}
\def\U{\bm U}

\def\tr{{\rm tr}}

\def\Y{{\bm Y}}
\def\diag{\hbox{diag}}
\def\A{{\bf A}}
\def\D{{\bf D}}
\def\V{{\bm V}}
\def\T{{\bm T}}
\def\R{{\bf R}}
\def\bth{{\bm\theta}}
\def\C{{\mathbf{C}}}

\title{\bf\Large Spatial-Sign based High dimensional Change Point Inference}
\author{Jixuan Liu$^1$, Long Feng$^1$, Liuhua Peng$^2$ and Zhaojun Wang$^1$\\
School of Statistics and Data Science, KLMDASR, LEBPS, and LPMC,\\ Nankai University$^1$\\
University of Melbourne$^2$}
\date{\today}

\begin{document}

\maketitle
\begin{abstract}
High-dimensional changepoint inference, adaptable to diverse alternative scenarios, has attracted significant attention in recent years. In this paper, we propose an adaptive and robust approach to changepoint testing. Specifically, by generalizing the classical mean-based cumulative sum (CUSUM) statistic, we construct CUSUM statistics based on spatial medians and spatial signs. We introduce test statistics that consider the maximum and summation of the CUSUM statistics across different dimensions, respectively, and take the maximum across all potential changepoint locations. The asymptotic distributions of test statistics under the null hypothesis are derived. Furthermore, the test statistics 
exhibit asymptotic independence under mild conditions. Building on these results, we propose an adaptive testing procedure that combines the max-$L_\infty$-type and max-$L_2$-type statistics to achieve high power under both sparse and dense alternatives. Through numerical experiments and theoretical analysis, the proposed method demonstrates strong performance and exhibits robustness across a wide range of signal sparsity levels and heavy-tailed distributions. 

\end{abstract}

\section{Introduction}

High-dimensional data often exhibit complex heterogeneity, arising in genomics, finance, neuroscience, and environmental monitoring.  One key form of heterogeneity is the changepoint structure, where the data process suddenly changes at some unknown time point or location.  Detecting and localizing such changepoints is vital: in genomics it can indicate copy number alterations; in finance it reveals market regime shifts; and in network monitoring it signals emerging anomalies.
For an extensive review, see \cite{aue2013structural,niu2016multiple,casini2019structural,truong2020selective}.


In this paper, we consider a sequence of $p$-dimensional random vectors of size $n$, i.e., $\{\X_i:=(X_{i,1},\ldots,X_{i,p})^\top\in\mathbb{R}^p\}_{i=1}^n$, from the following mean-change model:
\begin{align}\label{cpm}
    \X_i = {\boldsymbol\theta}_0 + \boldsymbol\delta \ind{i>\tau} + {\boldsymbol\epsilon}_i,\ \ i=1,\ldots,n,
\end{align}
where ${\boldsymbol\theta}_0\in\mathbb{R}^p$ represents the baseline mean level, $\boldsymbol\delta=(\delta_1,\ldots,\delta_p)^\top\in\mathbb{R}^p$ is the change signal parameter measuring the magnitude of the change in mean, $\tau\in\{1,\ldots,n\}$ denotes a potential changepoint, and $\{{\boldsymbol\epsilon}_i=(\epsilon_{i,1},\ldots,\epsilon_{i,p})^\top\in\mathbb{R}^p\}_{i=1}^n$ are random noises with zero mean. The goal of interest is to test whether there exists a changepoint, that is,
\begin{align}\label{H0}
\begin{gathered}
    H_0: \tau=n\ \text{and}\ \boldsymbol\delta=\mathbf{0}\ \text{~~versus~~} 
    H_1: \text{there exists}\ \tau\in\{1,\ldots,n-1\}\ \text{and}\ \boldsymbol\delta\neq \mathbf{0},
\end{gathered}
\end{align}
under the scenario where both the sample size $n$ and dimension $p$ grow to infinity. 
A review of recent developments in various testing procedures for \eqref{H0} is provided by \cite{liu+zhang+liu-2022high}.

The most widely used methods for testing \eqref{H0} are to construct statistics that compare segments of the data. Among these, the mean-based cumulative sum (CUSUM) statistic is the most common approach.
Specifically, the CUSUM statistic $\{\breve\C_{\gamma}(k)\}_{k=1}^{n}$ is frequently used with $\gamma=0$ or $0.5$, where
\begin{equation*}
    \begin{aligned}
        \breve\C_{\gamma}(k)=&\left\{\frac{k}{n}\left(1-\frac{k}{n}\right)\right\}^{1-\gamma}\sqrt{n}\breve{\D}^{-1/2}\left(\breve{\boldsymbol\theta}_{1:k}-
        \breve{\boldsymbol\theta}_{k+1:n}\right)\,.
    \end{aligned}
\end{equation*}
Here, $\breve{\boldsymbol\theta}_{a:b}=(b-a+1)^{-1}\sum_{i=a}^b\X_i$ for $1\leq a\leq b \leq n$, and $\breve{\D}^{-1}$ is an estimator for the inverse of the (long-run) variance. A common choice is a diagonal matrix $\breve{\D}=\diag\{\breve\sigma_1^{2},\breve\sigma_2^2,\ldots,\breve\sigma_n^2\}$ with $\breve{\sigma}_j^2$ being the sample variance of $\{X_{1,j},X_{2,j},\ldots,X_{n,j}\}$ for $j=1,2,\ldots,p$.

For the mean-based CUSUM statistic, various methods for aggregating dimensions and locations have been explored. \cite{bai2010common, horvath2012change, 2016-jin+pan+yang+zhou-p2355} considered 
the max-$L_2$-type statistic $\max_{1 \leq k \leq n} \|\breve\C_0(k)\|^2$, and established its convergence, after normalization, to the supremum of a Gaussian process under $H_0$. 
\citet{wang2022inference} replaced each component of $\breve\C_0(k)$ with a self-normalized $U$-statistic.
\citet{2013-chan+Horvath+Huskova-p955} proposed $ \max_{\lambda_n \leq k \leq n - \lambda_n} \|\breve\C_{0.5}(k)\|^2 $ 
with $ \lambda_n \in [1, n/2] $ as a user-specified boundary removal parameter, and showed convergence to the extreme value distribution of the Gumbel type under $ H_0 $.
Alternatively, \citet{wang+zou+wang+yin-2019-Multiple} considered a sum-$L_2$-type statistic $\sum_{k=1}^{n-1} \|\breve\C_{0.5}(k)\|^2$.
Beyond $ L_2 $-aggregations, $ L_\infty $-aggregations in conjunction with the maximum operator have also attracted considerable attention. \cite{Jirak-2015} proposed the max-$L_\infty$-type statistic $ \max_{1 \leq k \leq n} \|\breve\C_0(k)\|_\infty $, and showed that it converges to the Gumbel distribution under $ H_0 $. \citet{yu2021finite} considered $ \max_{\lambda_n \leq k \leq n - \lambda_n} \|\breve\C_{0.5}(k)\|_\infty $ and employed a multiplier bootstrap to approximate its null distribution.
Furthermore, \cite{wang2023} also considered $ \max_{\lambda_n \leq k \leq n - \lambda_n} \|\breve\C_{0.5}(k)\|_\infty $, and established its convergence to the Gumbel distribution under $H_0$, thereby enabling simple implementation that avoids numerical approximations.

Many changepoint detection methods rely on sample means or assume Gaussian or other light-tailed distributions \citep{horvath2012change,2013-chan+Horvath+Huskova-p955,2016-jin+pan+yang+zhou-p2355}, leading to poor performance under heavy-tailed data.
In traditional multivariate analysis, \citet{matteson2014nonparametric} developed a homogeneity test based on energy distance combined with a maximum-type statistic. 
In addition, \citet{lung2015homogeneity} introduced a rank-based method that extends the Mann–Whitney–Wilcoxon two-sample test to changepoint detection using a maximum operator. 
However, both methods are limited to fixed dimensions, and they fail or lack theoretical guarantees as the dimension $p$ tends to infinity.
To fill in this gap, we propose changepoint tests based on spatial medians and spatial signs \citep{oja2010multivariate}, which are robust to heavy-tailed data and have been extensively applied to high-dimensional data analysis \citep{zou2014multivariate,wang+peng+li-2015high,feng2016multivariate,cheng2023statistical,liu+feng+wang+2024}. 
In this paper, we develop max-$L_\infty$-type tests based on spatial medians, which are powerful under sparse change signals, and max-$L_2$-type tests based on spatial signs, which are effective when the change is dense.

In practice, whether the alternatives are dense or sparse is often unknown. To address this, adaptive strategies have been developed that combine $L_2$-type and $L_\infty$-type tests, which are sensitive to dense weak signals and sparse strong signals, respectively.
These adaptive strategies are designed to be effective across a wide range of alternative change patterns.
Let $\breve\C_{\gamma}(k)=(\breve{C}_{\gamma,1}(k),\ldots,\breve{C}_{\gamma,p}(k))^{\top}$,
\cite{liu2020unified} introduced $\breve{T}_{q,s_0}=\max_{\lambda\leq k\leq n-\lambda}\{ \sum_{j=1}^{s_0}\vert \breve{C}_{0,(j)}(k)\vert^q \}^{1/q}$ with $1\leq q\leq \infty$ and $1\leq s_0\leq p$, where $\vert \breve{C}_{0,(1)}(k)\vert\geq \cdots\geq \vert \breve{C}_{0,(q)}(k)\vert$ are the order statistics of $\{ \vert \breve{C}_{0,j}(k)\vert \}_{j=1}^p$.
They then proposed an adaptive procedure by taking the minimum of $p$-values corresponding to $T_{q,s_0}$ over a series of $q$ values with a fixed $s_0$. 
Similarly, \cite{zhang2022adaptive} considered an adaptive test over a series of self-normalized $U$-statistic-based CUSUM statistics. 
\citet{wang2023} proposed double-max-sum methods that combine $p$-values from max-$L_\infty$-type and sum-$L_2$-type tests using their asymptotic independence.
However, all these methods are based on sample means and are not robust to heavy-tailed distributions. This motivates us to develop adaptive strategies that combine spatial-median- and spatial-sign-based $L_\infty$-type and $L_2$-type tests.

In this paper, we propose CUSUM statistics based on spatial medians and spatial signs, respectively.
These are used to construct max-$L_\infty$-type and max-$L_2$-type test statistics, each defined by taking the maximum over all possible changepoint locations.
The proposed tests apply to a general model that accommodates heavy-tailed distributions. In addition, we develop adaptive strategies that combine the $p$-values from the two types of tests using Fisher combination, thereby leveraging the strengths of both tests under different change signals.
The contributions of this paper are outlined as follows.
\begin{itemize}
    \item[(i)] Our proposed methods are based on spatial medians and spatial signs, which are well-recognized techniques for analyzing heavy-tailed data.
    These approaches not only exhibit advantageous performance for heavy-tailed data but also maintain results comparable to mean-based methods when applied to normally distributed data. Although spatial-sign based methods have been extensively studied in the literature, this is the first paper to apply them to changepoint inference. Our work pioneers the integration of spatial-sign techniques into this area, offering a robust and distribution-free approach for testing changepoints and detecting structural changes. This novel application not only broadens the scope of spatial-sign methods but also provides new insights and tools for high-dimensional change point analysis.
    
     \item[(ii)] The adaptive strategies proposed in this paper, which combine $p$-values from both max-$L_2$-type and max-$L_\infty$-type tests, effectively adjust to different levels of signal sparsity. Extensive simulation studies demonstrate that the combined test consistently outperforms existing methods, particularly under heavy-tailed distributions. Therefore, our proposed methods offer dual robustness--they are not only resilient to heavy-tailed data but also highly adaptive to varying sparsity levels of alternatives. This dual advantage marks a significant contribution to the literature on high-dimensional change point inference.
     
     \item[(iii)] 
     Theoretically, we derive the asymptotic null distributions of the max-$L_2$-type and max-$L_\infty$-type test statistics under a general model. Furthermore, we establish the asymptotic independence between the two statistics, which motivates the adaptive procedure that combines their $p$-values. Finally, we characterize the asymptotic behavior of the proposed tests under the local alternative. This paper is the first to study the asymptotic independence between two Gumbel-type limit distributions in high-dimensional settings. In contrast, most existing works focus on asymptotic independence between a Gumbel distribution and an asymptotically normal distribution. Establishing such a result is highly nontrivial and requires the development of several new technical tools. Our work thus fills an important gap in the literature and opens new avenues for studying extreme value theory under high-dimensional asymptotics.
 \end{itemize}

The paper is organized as follows. Section~\ref{sec:preliminary} reviews spatial medians and spatial signs with model assumptions. Sections~\ref{Sec:MAX} and~\ref{Sec:SUM} introduce max-$L_\infty$-type and max-$L_2$-type tests, respectively, and derive their asymptotic properties. Section~\ref{Sec:adaptive} presents the adaptive combination strategy and its theoretical justification. Simulation studies are reported in Section~\ref{Sec:Simus}, and real data applications are presented in Section~\ref{Sec:Real data}. Concluding remarks are in Section~\ref{Sec:Con Remarks}.

\textbf{Notations:}  For a $d$-dimensional vector $\boldsymbol x$, denote its Euclidean norm and maximum-norm as $\Vert \boldsymbol x\Vert$ and $\Vert \boldsymbol x\Vert_\infty$, respectively. Denote $a_n\lesssim b_n$ if there exists constant $C$, $a_n\leq C b_n$ and $a_n \asymp b_n$ if  both $a_n\lesssim b_n$ and $b_n\lesssim a_n$ hold. For $a, b \in \mathbb{R}$, we write $a \wedge b=\min \{a, b\}$. Let $\psi_{\alpha_0}(x)=\exp \left(x^{\alpha_0}\right)-1$ be a function defined on $[0, \infty)$ for $\alpha_0>0$. Then the Orlicz norm $\|\cdot\|_{\psi_{\alpha_0}}$ of a random variable $X$ is defined as $\|X\|_{\psi_{\alpha_0}}=\inf \left\{t>0, \E\left\{\psi_{\alpha_0}(|X| / t)\right\} \leqslant 1\right\}$. Let $\operatorname{tr}(\cdot)$ be a trace for matrix, $\lambda_{min}(\cdot)$ and $\lambda_{max}(\cdot )$ be the minimum and maximum eigenvalue for symmetric martix. For a symmetric matrix $\A=(a_{ij})_{p\times p}$, we denote $\Vert \A\Vert_1=\Vert \A\Vert_\infty=\max_{1\leq j\leq p}\sum_{i=1}^p\vert a_{ij}\vert$, $\Vert \A\Vert_F=\left\{\tr(\A^2)\right\}^{1/2}$. Denote $\mathbf I_p$ as the $p$-dimensional identity matrix, and $ \operatorname{diag}\{v_1,v_2,\ldots,v_p\}$ to be the diagonal matrix with entries $\boldsymbol v=(v_1,v_2,\ldots,v_p)^{\top}$. 

\section{Preliminary}\label{sec:preliminary}

In this paper, we consider the following model for random noises $\{\boldsymbol\epsilon_i\}_{i=1}^n$:
\begin{align}\label{modelx}
\boldsymbol\epsilon_i=\nu_i\mathbf\Gamma \boldsymbol W_i,
\end{align}
where $\boldsymbol{\Gamma}$ is a nonrandom and invertible $p\times p$ matrix, $\nu_i$ is a nonnegative univariate random variable that is independent with the spatial sign of $\boldsymbol W_i$, and $\boldsymbol W_i=\left(W_{i, 1}, \ldots, W_{i, p}\right)^{\top}$ is a
$p$-dimensional random vector satisfies the following assumption.
\begin{assumption}\label{ass:max1}
    $W_{i, 1}, \ldots, W_{i, p}$ are i.i.d.~symmetric random variables with $\E\left(W_{i, j}\right)=0$, $\E\left(W_{i, j}^2\right)= 1$, and $\left\|W_{i, j}\right\|_{\psi_{\alpha_0}} \leqslant c_0$ with some constant $c_0>0$ and $1 \leqslant \alpha_0 \leqslant 2$.
\end{assumption}

\begin{remark}
    Model (\ref{modelx}) has been widely adopted in high-dimensional spatial median and spatial sign-based approaches \citep{wang+peng+li-2015high,cheng2023statistical,liu+feng+wang+2024}. It encompasses a broad class of widely used multivariate models and distribution families, such as the independent components model \citep{nordhausen2009signed,ilmonen2011semiparametrically,yao2015sample} with $\nu_i$ as a nonnegative constant and the family of elliptical distributions \citep{hallin2006semiparametrically,oja2010multivariate,fang2018symmetric} with $\boldsymbol{W}_i\sim N(\mathbf{0}, \mathbf{I}_p)$. 
    Assumption \ref{ass:max1} is identical to Condition C1 in \cite{cheng2023statistical}, ensuring that $\boldsymbol{\theta}_0 + \boldsymbol{\delta}\mathbb{I}(i>\tau)$ coincides with the population spatial median of $\X_i$ and that $W_{i,j}$ follows a sub-exponential distribution. 
    For elliptical distributions where $\boldsymbol{W}_i\sim N(\mathbf{0},\mathbf{I}_p)$, Assumption \ref{ass:max1} holds automatically.
\end{remark}

The spatial sign is an extension of the univariate sign to vectors and the spatial sign function is defined as $U(\boldsymbol{x})=\|\boldsymbol{x}\|^{-1}\boldsymbol{x}\mathbb{I}(\boldsymbol{x}\neq \boldsymbol{0})$. Spatial sign-based techniques are widely employed for inference on location parameters in multivariate and high-dimensional settings \citep{oja2010multivariate,cheng2023statistical}. These methods offer improved efficiency compared to mean-based approaches in heavy-tailed distributions.
They typically require an estimator of the location parameter, for which we adopt the sample spatial median in this paper.

Based on $\X_a,\ldots,\X_b$ for $1\leq a\leq b \leq n$, the classical sample spatial median $\tilde{\boldsymbol\theta}_{a:b}$ is defined as  
\begin{equation*}
    \begin{aligned}
        \tilde{\boldsymbol \theta}_{a:b}=\arg\min_{\boldsymbol \beta}\sum_{i=a}^{b}\|\X_i-\boldsymbol\beta\|,
    \end{aligned}
\end{equation*}
serving as an estimator of the corresponding population spatial median. While $\tilde{\boldsymbol{\theta}}_{a:b}$ demonstrates robustness in multivariate settings \citep{oja2010multivariate,cheng2023statistical}, it discards scalar information for each variable and may perform poorly when substantial differences exist across dimensions. 
To address this limitation, \cite{feng2016multivariate} proposed a scalar-transformation-invariant method that jointly estimates the median and a diagonal matrix to standardize each variable to a common scale, accounting for variance heterogeneity.
In particular, we seek a pair of diagonal matrix ${\mathbf{D}}$ and vector ${\boldsymbol{\theta}}$ that jointly satisfy
\begin{align}\label{eqs}
\frac{1}{b-a+1} \sum_{i=a}^{b} U({\boldsymbol{\varepsilon}}_i)=\mathbf{0} \text { ~~and~~ } \frac{p}{b-a+1} \operatorname{diag}\left\{\sum_{i=a}^{b} U({\boldsymbol{\varepsilon}}_i) U({\boldsymbol{\varepsilon}}_i)^\top\right\}=\mathbf{I}_{p},
\end{align}
where ${\boldsymbol{\varepsilon}}_i={\mathbf{D}}^{-1 / 2}(\boldsymbol{X}_i-\boldsymbol{\theta})$. The pair $\left({\D},{\boldsymbol{\theta}}\right)$ can be viewed as a simplified version of the Hettmansperger-Randles (HR) estimator \citep{hettmansperger2002practical}, ignoring the off-diagonal elements of the scatter matrix. To solve \eqref{eqs}, we can adapt the recursive algorithm of \cite{feng2016multivariate}, iterating the following three steps until convergence:
\begin{itemize}
\item[(i)] $\bm\varepsilon_i \leftarrow \D^{-1/2}(\X_i-\bth)$,
~~$i=a,\ldots,b$;
\item[(ii)] $\bth \leftarrow
\bth+\frac{\D^{1/2}\sum_{j=a}^{b}U(\bm\varepsilon_i)}{\sum_{j=a}^{b}\|\bm\varepsilon_i\|^{-1}}$;
\item[(iii)] $\D \leftarrow p
\D^{1/2}\diag\{(b-a+1)^{-1}\sum_{i=a}^{b}U(\bm\varepsilon_i)U(\bm\varepsilon_i)^\top\}\D^{1/2}$.
\end{itemize}
The resulting estimators of location and diagonal matrix based on $\X_a,\ldots,\X_b$ are denoted as $\hat{\bm \theta}_{a:b}$ and $\hat{\D}_{a:b}$. The algorithm can be initialized using the sample mean and sample variances.

For $i=1,\ldots,n$, we denote
$\boldsymbol U_i=U(\mathbf D^{-1/2}\boldsymbol\epsilon_i)$ and $R_i=\Vert \mathbf D^{-1/2}\boldsymbol\epsilon_i\Vert$ as the scale-invariant spatial-sign and radius of the random noise is $\boldsymbol\epsilon_i$, respectively. 
Denote $\D=\diag\{d_1^2,\ldots,d_p^2\}$ and $\boldsymbol W_i=\left(W_{i, 1}, \ldots, W_{i, p}\right)^{\top}$, we impose the following assumptions.

\begin{assumption}\label{ass:max2}
    The moments $\zeta_k=\E\left(R_i^{-k}\right)$ for $k=1,2,3,4$ exist for large enough $p$. In addition, there exist two positive constants $\underline{b}$ and $\bar{B}$ such that $\underline{b} \leqslant \lim \sup _p \E\left(R_i / \sqrt{p}\right)^{-k} \leqslant \bar{B}$ for $k=1,2,3,4$.
\end{assumption}

\begin{assumption}\label{ass:max3}
    There exist some positive constant $\underline d$ such that $\lim\inf_{p\rightarrow \infty}\min_{j=1,2,\ldots,p}d_j>\underline{d}$.
    In addition,
    the shape matrix $\R=\mathbf D^{-1/2}\mathbf\Gamma\mathbf\Gamma^\top \mathbf D^{-1/2}=\left(\sigma_{j \ell}\right)_{p \times p}$ satisfies: (i) $\tr(\R)=p$; (ii) there exist positive constants $\underline{m}$ and $\overline{M}$ such that $\underline{m}\leq \sigma_{jj}\leq \overline{M}$ for $j=1,2,\ldots,p$; (iii) $\max _{j=1,\ldots,p}\sum_{\ell=1}^p\left|\sigma_{j \ell}\right| \leqslant a_0(p)$, where $a_0(p)\asymp p^{1-\eta_0}$ for some positive constant $\eta_0\leq 1/2$. (iv) $\tr\left(\R^2\right)-p=o\left(n^{-1} p^2\right)$.
\end{assumption}

\begin{remark}
     Assumption \ref{ass:max2} extend Assumption 1 in \cite{zou2014multivariate}, which indicates that $\zeta_k\asymp p^{-k/2}$ for $k=1,2,3,4$. This is a mild condition introduced to prevent $\boldsymbol X_i$ from concentrating too much near its population spatial median. 
     It has been verified in \cite{zou2014multivariate} that Assumption \ref{ass:max2} holds for multivariate normal, Student-$t$, and mixtures of multivariate normal distributions.
     For further discussions on similar assumptions, see \cite{cardot2013efficient, zou2014multivariate, cheng2023statistical}.
\end{remark}

\begin{remark}
    Conditions (i)--(iii) on $\R$ in Assumption \ref{ass:max3} are commonly adopted and are similar to Condition C3 in \cite{cheng2023statistical}, where a similar condition is imposed on $\boldsymbol{\Gamma}\boldsymbol{\Gamma}^{\top}$ instead of on $\R$.
    The introduction of $\D$ enhances the efficiency of our methods compared to those based on $\widetilde{\boldsymbol\theta}_{a:b}$, particularly when there are significant variance differences across dimensions. Conditions (iv) on $\R$ in Assumption \ref{ass:max3} is crucial for establishing the consistency of the diagonal matrix estimators \citep{liu+feng+wang+2024}.
\end{remark}

\begin{remark}
    Assumptions \ref{ass:max1}--\ref{ass:max3} ensure that under $H_0$, when $b-a\to\infty$ satisfies $\log p = o((b-a)^{1/3})$ and $\log (b-a) = o(p^{1/3 \wedge \eta_0})$, $\hat{\boldsymbol\theta}_{a:b}$ admits a Bahadur representation with a maximum-norm bound on the remainder term \citep{liu+feng+wang+2024}. Specifically, we have
    \begin{align*}
        \hat{\D}_{a:b}^{-1/2}\left(\hat{\bth}_{a:b}-\bth_0\right)=\frac{1}{b-a+1}\zeta_{1}^{-1}\sum_{i=a}^b \U_i+\mathbf{C}_{a:b},
    \end{align*}
    where $\Vert \mathbf{C}_{a:b}\Vert_{\infty}=(b-a)^{-1/2}O_p[(b-a)^{-1/4}\log^{1/2}\{(b-a)p\}+p^{-(1/6\wedge \eta_0/2)}\log^{1/2}\{(b-a)p\}] = o_{p}((b-a)^{-1/2})$.
\end{remark}

\section{Max-$L_\infty$-type tests}\label{Sec:MAX}

It is well known that $L_\infty$-type statistics are particularly effective in detecting sparse alternatives. In this section, we introduce two max-$L_\infty$-type test statistics based on spatial median for testing \eqref{H0}. 


We account for the potential changepoint in Model \eqref{modelx} under the alternative hypothesis when estimating the diagonal matrix $\D$. Assume that the changepoint $\tau$ does not occur within the first or last $\varrho$-proportion of the samples, where $\varrho\in(0,1/2)$ is a fixed constant. This assumption is commonly adopted in the changepoint detection literature; see, for example, \cite{zhao2022optimal}. Denote $(\hat{\bth}_{1}^{(\varrho)},\hat{\D}_{1}^{(\varrho)}):=(\hat{\bth}_{1:[n\varrho]},\hat{\D}_{1:[n\varrho]})$ and $(\hat{\bth}_{2}^{(\varrho)},\hat{\D}_{2}^{(\varrho)}):=(\hat{\bth}_{(n-[n\varrho]+1):n}, \hat{\D}_{(n-[n\varrho]+1):n})$ as the estimators of $(\bth,\D)$ based on the first $[n\varrho]$ and the last $[n\varrho]$ samples, respectively. 
Denote by $\hat{d}_{1,1}^{(\varrho)2}$ and $\hat{d}_{2,1}^{(\varrho)2}$ the first diagonal element of $\hat{\D}_{1}^{(\varrho)}$ and $\hat{\D}_{2}^{(\varrho)}$, respectively. These quantities serve as estimators of $d_1^2$ in $\D$. 
Define $$\hat{\D} =\left(\hat{\D}_{1}^{(\varrho)}/\hat{d}_{1,1}^{(\varrho)2} + \hat{\D}_{2}^{(\varrho)}/\hat{d}_{2,1}^{(\varrho)2}\right)/2,$$ which serves as a consistent estimator of $\D/d_1^2$ under both the null and alternative hypotheses.  The consistency of $\hat{\D}$ can be established similarly to the proof of Lemma 2 in the Supplementary Materials of \cite{liu+feng+wang+2024} under suitable conditions.

For $k=1,\ldots,n$, we define the spatial-median-based CUSUM statistic as
\begin{equation}
    \notag 
    \C_{\gamma}(k)=\left\{\frac{k}{n}\left(1-\frac{k}{n}\right)\right\}^{1-\gamma}\sqrt{n}\hat{\D}^{-1/2}\left(\hat{\bth}_{1:k}-\hat{\bth}_{(k+1):n}\right)\,.
\end{equation}
Given the relatively slow convergence rate of the maximum norm of $\C_{\gamma}(k)$, we proposed two versions of the adjusted max-$L_\infty$-type statistics, defined as
\begin{align*}
    M_{n,p}:=\max_{\lambda_n\leq k\leq n-\lambda_n}\|\C_{0}(k)\|_{\infty}\cdot (1-n^{-1/2})\ \text{~and~}\
    M^{\dagger}_{n,p}:=\max_{\lambda_n\leq k\leq n-\lambda_n}\|\C_{0.5}(k)\|_{\infty}\cdot (1-n^{-1/2})\,,
\end{align*}    
where $\lambda_n\in[1,n/2]$ is a pre-specified boundary removal parameter.

Many researchers have studied the mean-based max-$L_\infty$-type statistics \citep{Jirak-2015,yu2021finite,wang2023}, defined as $\breve{M}_{n,p}=\max_{1 \leq k \leq n} \|\breve\C_{0}(k)\|_{\infty} $ and $ \breve{M}_{n,p}^{\dagger}=\max_{\lambda_n \leq k \leq n - \lambda_n} \|\breve\C_{0.5}(k)\|_{\infty}$.
When $n\wedge p\rightarrow\infty$, \cite{Jirak-2015} showed that $\breve{M}_{n,p}$ weakly converges to the Gumbel distribution under certain decay conditions on componentwise correlations, provided that $H_0$ holds. \cite{yu2021finite} proposed a multiplier bootstrap method to approximate the distribution of $\breve{M}^\dagger_{n,p}$ under $H_0$. \cite{wang2023} further derived the asymptotic null distribution for both $\breve{M}_{np}$ and $\breve{M}^\dagger_{n,p}$ under weaker conditions on componentwise correlations among $p$ variables compared to \cite{Jirak-2015}.

We now derive the asymptotic distribution of $M_{n,p}$ and $M_{n,p}^{\dagger}$ under $H_0$.
To accommodate dependence across dimensions, we introduce the following assumption, which is less restrictive than the logarithmic decay condition imposed in \cite{Jirak-2015}. For a more detailed discussion of this assumption, we refer to \citet{liu+feng+wang+2024}.


\begin{assumption}[Componentwise correlations]\label{ass:cor}
        Assume that $\max_{1\leq j<\ell\leq p}|\sigma_{j\ell}|\leq \varrho_0$ for all $p\geq 2$ for some $\varrho_0 \in(0,1)$. Let $\left\{\varpi_p\right\}_{p \geq 1}$ and $\left\{\kappa_p\right\}_{p \geq 1}$ be sequences of positive constants satisfying $\varpi_p=o(1 / \log p)$ and $\kappa_p \rightarrow 0$ as $p \rightarrow \infty$. For $1 \leq j \leq p$, define $B_{p, j}=\left\{1 \leq \ell \leq p: \left|\sigma_{j\ell}\right| \geq \varpi_p\right\}$ and $C_p=\left\{1 \leq j \leq p: \left|B_{p, j}\right| \geq p^{\kappa_p}\right\}$. We assume that $\left|C_p\right| / p \rightarrow 0$ as $p \rightarrow \infty$.
\end{assumption}

\begin{theorem}\label{thm:Max-Max}
Suppose Assumptions \ref{ass:max1}--\ref{ass:cor} hold, if $\log^7 n=o(p^{1/6\wedge\eta_0/2})$ and $\log^2 p=o(n^{1/5\wedge \lambda/3})$ for some positive constant $\lambda\in(0,1)$, then under $H_0$, 
\begin{itemize}
\item[(i)] If $\lambda_n\sim n^{\lambda}$, as $(n,p)\to\infty$,
\[
    \pr\left(p^{1/2}\zeta_1 M_{n,p}\leq u_p\{\exp(-x)\}\right) \to \exp\{-\exp(-x)\},
\]
where $u_p\{\exp(-x)\}=\sqrt{\{x+\log(2p)\}/2}$.

\item[(ii)] If $\lambda_n\sim n^{\lambda}$, as $(n,p)\to\infty$,
\[
    \pr\left(p^{1/2}\zeta_1 M^{\dagger}_{n,p}\leq \frac{x+D(p\log h_n)}{A(p\log h_n)}\right) \to \exp\{-\exp(-x)\},
\]
where $A(x)=\sqrt{2\log x}$, $D(x)=2\log x+2^{-1}\log\log x-2^{-1}\log\pi$ and $h_n=\left\{(\lambda_n/n)^{-1}-1\right\}^2$.
\end{itemize}
\end{theorem}

\begin{remark}
    { A key contribution of Theorem \ref{thm:Max-Max} is showing that, under mild conditions, $M_{n,p}$ and $M_{n,p}^\dagger$ share the same asymptotic Gumbel distribution, extending the mean-based results of \citet{wang2023} to spatial medians. While prior work focused on a single spatial median $\hat{\bth}_{1:n}$ \citep{liu+feng+wang+2024}, our analyses of $M_{n,p}$ and $M_{n,p}^\dagger$ involve a sequence of dependent spatial medians $\hat{\bth}_{1:k}$ for $k\in\{\lambda_n,\ldots,n-\lambda_n\}$, which is more theoretically challenging.
    }
\end{remark}






To implement the max-$L_\infty$-type tests based on $M_{n,p}$ and $M^{\dagger}_{n,p}$, we need to estimate the unknown quantity $\zeta_1$. 
To eliminate the effect of potential changepoints, we estimate $\zeta_1$ by
\begin{align*}
    \hat{\zeta}_{1}=\frac{1}{2[n\varrho]}\sum_{i=1}^{[n\varrho]}\|\hat{\D}^{-1/2}(\X_i-\hat{\bm \theta}_{1}^{(\varrho)})\|^{-1}+\frac{1}{2[n\varrho]}\sum_{i=n-[n\varrho]+1}^{n}\|\hat{\D}^{-1/2}(\X_i-\hat{\bm \theta}_{2}^{(\varrho)})\|^{-1},
\end{align*}
which ensures that the estimation $\hat{\zeta_1}$ is derived by the stable and homogeneous segments of data. Similar to the Proof of Lemma 3 in the Supplementary Materials of \cite{liu+feng+wang+2024}, it can be shown that $\hat{\zeta}_1$ is a consistent estimator of $\zeta_1/d_1$, i.e., $\hat{\zeta}_1\cp \zeta_1/d_1$ as $(n,p)\rightarrow \infty$, under both the null and alternative hypotheses.

Based on Theorem \ref{thm:Max-Max}, we obtain the $p$-values associated with $M_{n,p}$ and $M^{\dagger}_{n,p}$ as
\begin{align*}
    {\rm p}_{M_{n,p}} &:= 1 - G\big(2p\hat \zeta_{1}^2M_{n,p}^2-\log(2p)\big)\ \text{~~and}\\
    {\rm p}_{M^{\dagger}_{n,p}} &:= 1 - G\big( p^{1/2}\hat \zeta_{1}A(p\log h_n)M^{\dagger}_{n,p}-D(p\log h_n)\big)\,,
\end{align*}
where $G(x)=\exp\{-\exp(-x)\}$ denotes the standard Gumbel distribution.
If the $p$-value falls below a pre-specified significant level $\alpha\in(0,1)$, we reject the null hypothesis that there is no changepoint in the data sequence.
It can be expected that either max-$L_\infty$-type testing procedure would be effective in detecting sparse and strong change signals.

\begin{pro}\label{prop:max-Linf}
    Suppose Assumptions \ref{ass:max1}--\ref{ass:cor} hold and $\tau=[cn]$ for some $c\in (0,1)$. Then, if $\log^7 n=o(p^{1/6\wedge\eta_0/2})$, $\lambda_n\sim n^\lambda$ and $\log^2 p=o(n^{1/5\wedge \lambda/3})$ for some positive constant $\lambda\in(0,1)$, we have, (i) the test based on $M_{np}$ is consistent if $\|\boldsymbol\delta\|_\infty\geq C\sqrt{\log p/n}$ for large enough constant $C$; (ii) the test based on $M_{np}^\dagger$ is consistent if $\|\boldsymbol\delta\|_\infty\geq C\sqrt{\log \{p\log(h_n)\}/n}$ for large enough constant $C$.
\end{pro}
Proposition \ref{prop:max-Linf} establishes the consistency of the max-$L_\infty$-type tests based on $M_{n,p}$ and $M_{n,p}^\dagger$ under $H_1$, subject to certain conditions on the magnitude of the changes. This result aligns with the optimal rate (up to a logarithmic factor) for sparse changepoint alternatives in the literature \citep{liu+zhang+liu-2022high}.

\section{Max-$L_2$-type tests}\label{Sec:SUM}

For the max-$L_2$-type approach, we introduce two types of scalar-transformation-invariant spatial-sign-based CUSUM test statistics, motivated by \citet{wang+peng+li-2015high}, \citet{feng2016multivariate}, and \citet{feng2016}. Specifically, for $k=1,\ldots,n$, we define
\begin{equation}
    \begin{aligned}
    \tilde{\C}_{\gamma}(k)=\left\{\frac{k}{n}\left(1-\frac{k}{n}\right)\right\}^{-\gamma}\sqrt{\frac{p}{n}}\left(\hat{\S}_{k}-\frac{k}{n}\hat{\S}_{n}\right),
    \end{aligned}
\end{equation}
where $\hat{\S}_{k}=\sum_{i=1}^k \hat{\U}_{i}$ for $k=1,\ldots,n$ with $\hat{\U}_{i}=U\big(\hat{\D}^{-1/2}(\X_i-\hat{\boldsymbol\theta}_{1:n})\big)$. 

For $\gamma=0$, we define the max-$L_2$-type test statistic $S_{n,p}$ as
\begin{align}
S_{n,p}=\max_{\lambda_n\le k\le n-\lambda_n}\left\{\tilde{\C}_{0}(k)^\top \tilde{\C}_{0}(k)-\frac{k(n-k)p}{n^2}\right\}\cdot(1-n^{-1/2}).
 \end{align}
This statistic serves as a spatial-sign-based analogue to the mean-based max-$L_2$-type statistic $\max_{1\leq k\leq n}\|\breve\C_0(k)\|^2$ \citep{bai2010common,horvath2012change,2016-jin+pan+yang+zhou-p2355}.
Following \citet{feng2016multivariate}, we impose the following assumption.
\begin{assumption}\label{ass:sum_R4_0}
(i) $\operatorname{tr}\left(\R^4\right)/\operatorname{tr}^2\left(\R^2\right)=o\left(1\right)$, (ii) $n^{-2} p^2 / \operatorname{tr}\left(\R^2\right)=O(n^{-\omega_0})$ for some $\omega_0\in (0,2)$.
\end{assumption}

\begin{remark}
    Assumption \ref{ass:sum_R4_0} (i) is a common condition for $L_2$-type test statistic in high dimension \citep{chen2010two,feng2016multivariate,wang+zou+wang+yin-2019-Multiple}, requiring that the eigenvalues of $\R$ do not diverge excessively. If all the eigenvalues of $\R$ are bounded, then $\tr(\R^2)=O(p)$ and $\tr(\R^4)=O(p)$. Consequently, Assumption \ref{ass:sum_R4_0}(i) holds trivially, while Assumption \ref{ass:sum_R4_0}(ii) simplifies to $p = O(n^{2-\omega_0})$ in this case.  
\end{remark}

\begin{theorem}\label{thm:Max-Sum}
    Suppose Assumptions \ref{ass:max1}--\ref{ass:max3} and \ref{ass:sum_R4_0} hold, and that $\log p=o(n)$. Then, under $H_0$, if $\lambda_n\rightarrow \infty$ and $\lambda_n/n\rightarrow 0$ as $n\rightarrow\infty$, it holds that
    \begin{align*}
        \frac{S_{n,p}}{\sqrt{2\tr(\R^2)}} \cd \max_{0\le t\le 1}V(t)\,,
    \end{align*}
    where $V(t)$ is a continuous Gaussian process with $\E\{V(t)\}=0$ and $\E\{V(t)V(s)\}=(1-t)^2s^2$ for $ 0\le s\le t \le 1$.
\end{theorem}

In practice, it is essential to construct a ratio-consistent estimator of $\tr(\R^2)$ under both the null and alternative hypotheses. 
To this end, we estimate $\tr(\R^2)$ using the first and last $[n\varrho]$ samples, as follows:
\begin{equation*}
    \begin{aligned}
        \widehat{\tr(\R^2)}=&\frac{p^2}{2[n\varrho]([n\varrho]-1)}\sum_{1\leq i\not=j\leq [n\varrho]}\left\{ U(\hat{\D}^{-1/2}(\X_{i}-\hat{\boldsymbol\theta}_{1}^{(\varrho)}))^\top U(\hat{\D}^{-1/2}(\X_{j}-\hat{\boldsymbol\theta}_{1}^{(\varrho)})) \right\}^2\\
        &+\frac{p^2}{2[n\varrho]([n\varrho]-1)}\sum_{n-[n\varrho]+1\leq i\not=j\leq n}\left\{ U(\hat{\D}^{-1/2}(\X_{i}-\hat{\boldsymbol\theta}_{2}^{(\varrho)}))^\top U(\hat{\D}^{-1/2}(\X_{j}-\hat{\boldsymbol\theta}_{2}^{(\varrho)})) \right\}^2.
    \end{aligned}
\end{equation*}
By Proposition 1 in \cite{li2016simpler}, it follows directly that $\widehat{\tr(\R^2)}/\tr(\R^2)\cp 1$ as $(n,p)\rightarrow \infty$.

According to Theorem \ref{thm:Max-Sum}, the $p$-value of the test based on $S_{n,p}$ is given by
\begin{align}
    {\rm p}_{S_{n,p}}=1-F_{V}\left(\frac{S_{n,p}}{\sqrt{2\widehat{\tr(\R^2)}}}\right),
\end{align}
where $F_{V}(\cdot)$ is the cumulative distribution funcion (cdf) of $\max_{0\leq t\leq 1} V(t)$.

\begin{remark}
    The quantiles of $\max_{0\leq t\leq 1} V(t)$ can be accurately approximated via Monte Carlo simulation. Consider a uniform discretization $T=\{t_i=i / N_d: i= 1, \ldots, N_d\}$ and the number of simulations $B$.
    For $b=1,\ldots,B$, let $v_b=\max _{t \in T} V_b(t)$, where $\left(V_b\left(t_1\right), \ldots, V_b\left(t_{N_d}\right)\right)^{\top}$ is sampled from the $N_d$-dimensional multivariate normal distribution with mean zero and covariance matrix with the $(j,\ell)$-th element given by $(1-t_j)^2 t_\ell^2$ for $1\leq \ell \leq j\leq N_d$.
    Then, the sample quantile of $\{v_b\}_{b=1}^{B}$ is used to approximate the theoretical quantile of $\max_{0\leq t\leq 1} V(t)$.
\end{remark}

For $\gamma=0.5$, we define the corresponding max-$L_2$-type test statistic as
\begin{align}
S^\dagger_{n,p}=\max_{\lambda_n\le k\le n-\lambda_n}\left\{\tilde \C_{0.5}(k)^\top \tilde \C_{0.5}(k)-p\right\}\cdot (1-n^{-1/2})\,,
\end{align}
which is the spatial-sign-based analogue to the mean-based statistic $\max_{\lambda_n\leq k\leq n-\lambda_n}\|\breve\C_{0.5}(k)\|^2$ \citep{2013-chan+Horvath+Huskova-p955}.

\begin{assumption}\label{ass:sum_R4}
    There exists a constant $\omega_1\in(0,1/4)$ such that: (i) $\operatorname{tr}\left(\R^4\right)/\operatorname{tr}^2\left(\R^2\right)=O\left(n^{-1+2\omega_1}\right)$; (ii) $\operatorname{tr}\left(\R^4\right)/\operatorname{tr}^2\left(\R^2\right)\exp\{-p/128\lambda^2_{\max}(\R) \}=O\left(n^{-1+2\omega_1}\right)$; (iii) $n=O(p^{1/(1-2\omega_1)})$; and (iv) $p^2n^{-2}/\tr(\R^2)=O(n^{-\omega_1})$.
\end{assumption}
\begin{remark}
      Assumption \ref{ass:sum_R4} is a stronger condition than Assumption \ref{ass:sum_R4_0}, ensuring that the remainder term in $S_{n,p}^\dagger$ remains $o_p(1)$. If all eigenvalues of $\R$ are bounded, this assumption reduces to $p = O(n^{2-\omega_1})$ and $n=O(p^{1/(1-2\omega_1)})$ for some $0<\omega_1<1/4$.
\end{remark}

\begin{theorem}\label{Thm:S_np5}
Suppose Assumptions \ref{ass:max1}--\ref{ass:max3} and \ref{ass:sum_R4} hold. Then, under $H_0$, if $\log p=o(n)$ and $\lambda_n\sim n^\lambda$ for some $\lambda\in (0,1)$, it holds that
\begin{align*}
\pr\left(\frac{A(\log(n^2/\lambda_n^2))}{\sqrt{2\tr(\R^2)}}|S_{n,p}^\dagger|\le x+D(\log(n^2/\lambda_n^2))\right)\to \exp\{-2\exp(-x)\}\,,
\end{align*}
where $A(x)$ and $D(x)$ are defined in Theorem \ref{thm:Max-Max}.
\end{theorem}

Theorem \ref{Thm:S_np5} implies that the $p$-value of the test based on $S^{\dagger}_{n,p}$ is
\begin{align*}
    {\rm p}_{S^{\dagger}_{n,p}} &:= 1 - \tilde{G}\big(A(\log(n^2/\lambda_n^2))|S_{n,p}^\dagger|/\sqrt{2\tr(\R^2)}-D(\log(n^2/\lambda_n^2)\big),
\end{align*}
where $\tilde{G}(x)=\exp\{-2\exp(-x)\}$ denotes the Gumbel distribution with a factor of 2 in the exponent. Both sum-$L_\infty$-type testing procedures based on $S_{np}$ and $S_{np}^\dagger$ are expected to be effective in detecting dense change signals. The following proposition establishes the consistency of these two tests under $H_1$.
\begin{pro}\label{prop:max-L2}
     Suppose Assumptions \ref{ass:max1}--\ref{ass:max3} and \ref{ass:sum_R4} hold. Under $H_1$ with $\tau=[cn]$ for some $c\in (0,1)$, if $\log p=o(n)$ and $\lambda_n\sim n^{\lambda}$ for some positive constant $\lambda\in (0,1)$, $\Vert\boldsymbol\delta\Vert_\infty=o((n\wedge p)^{1/2})$ and $\Vert\boldsymbol\delta\Vert^{-1}\Vert\boldsymbol\delta\Vert_{\infty}=o(p^{1/2}n^{-1/2})$, then the tests based on $S_{np}$ or $S_{np}^\dagger$ are consistent as $\Vert\boldsymbol\delta\Vert\rightarrow\infty$.
\end{pro}

\begin{remark}
    Proposition \ref{prop:max-L2} shows the consistency of the proposed max-$L_2$-type tests under a sequence of local alternatives. The condition $\Vert\boldsymbol\delta\Vert_\infty = o((n\wedge p)^{1/2})$ prevents excessively large signal components, preserving the model structure in high-dimensional setting. Similar constraints on the signal magnitude are also imposed in~\cite{wang+peng+li-2015high,feng2016multivariate} to ensure that the properties of the test statistic under the alternative hypothesis can be properly characterized. The condition $\Vert\boldsymbol\delta\Vert^{-1}\Vert\boldsymbol\delta\Vert_\infty=o(p^{1/2}n^{-1/2})$ restricts the signal from being overly sparse. Let $\delta_{\max}=\max\{|\delta_1|,\ldots,|\delta_p|\}$ and $\delta_{\min}=\min\{|\delta_1|,\ldots,|\delta_p|\}$, and $s_0$ the number of nonzero components in $\boldsymbol\delta$, this condition simplifies to $s_0/n \gg p^{-1}$ when $\delta_{\max} \asymp \delta_{\min}$.  
\end{remark}

\section{Adaptive Strategy}\label{Sec:adaptive}

In practice, whether the potential signal is sparse or dense across dimensions is often unknown. To capture different types of signals, we propose integrating max-$L_\infty$-type and max-$L_2$-type testing procedures, inspired by \citet{wang2023}, which focused on test statistics based on sample means. A key characteristic of this combined approach is that, under some mild conditions and $H_0$, the max-$L_\infty$-type and max-$L_2$-type statistics are asymptotically independent. To proceed, we introduce the following additional assumption:
\begin{assumption}\label{ass:ind1}
    There exist constants $\eta_1>0$ and $\varrho_0 \in(0,1)$ such that $\max _{1 \leq j<\ell \leq p}\left|\sigma_{j \ell}\right| \leq \varrho_0$ and $\max _{1 \leq j \leq p} \sum_{\ell=1}^p \sigma_{j\ell}^2 \leq(\log p)^{\eta_1}$ for all $p \geq 3$. In addition, there exist some constants  $0<\underline{c}<\bar{c}<\infty$, such that $\underline{c}\leq \lambda_{\min }({\R}) \leq \lambda_{\max }(\R)\leq \bar{c}$. 
\end{assumption}

\begin{remark}
    Assumption \ref{ass:ind1} is stronger than Assumption \ref{ass:cor} and \ref{ass:sum_R4_0} (i). Under Assumption \ref{ass:ind1}, $\tr(\R^4)/\tr^2(\R^2)=O(p^{-1})$. Therefore, Assumptions \ref{ass:sum_R4_0} (ii) and \ref{ass:sum_R4} are satisfied if $n=O(p^{2-4(-\omega_1+1/4)/(1-2\omega_1)})$ and $p^{3/(4-2\omega_1)}(\log p)^{-\eta_1/(2-\omega_1)}=O(n)$. Intuitively, if $\lim_{n\to\infty} p/n \in (0,\infty)$, all of Assumptions \ref{ass:cor}--\ref{ass:sum_R4} are satisfied.
\end{remark}

\begin{theorem}\label{thm:ind_H0}
Suppose $H_0$ and Assumptions \ref{ass:max1}--\ref{ass:max3} and \ref{ass:sum_R4}--\ref{ass:ind1} hold, if $\log^7 n=o(p^{1/6\wedge\eta_0/2})$ and $\log^2 p=o(n^{1/5\wedge \lambda/3})$, we have,
\begin{itemize}
\item[(i)] If $\lambda_n\sim n^{\lambda}$ for some $\lambda\in(0,1)$ , then, as $(n,p)\to\infty$, $M_{n,p}$ is asymptotically independent of $S_{n,p}$ in the sense that
\[
    \pr\left(p^{1/2} \zeta_1 M_{n,p}\leq u_p\{\exp(-x)\}, \frac{S_{n,p} }{\sqrt{2{\tr(\R^2)}}}\leq y\right) \to \exp\{-\exp(-x)\}\cdot F_V(y);
\]

\item[(ii)] If $\lambda_n\sim n^{\lambda}$ for some $\lambda\in(0,1)$, then, as $(n,p)\to\infty$, $M^{\dagger}_{n,p}$ is asymptotically independent of $S_{n,p}^{\dagger}$ in the sense that

\begin{equation*}
    \begin{aligned}
        \pr\left(p^{1/2} \zeta_1 M^{\dagger}_{n,p}\leq \frac{x+D(p\log h_n)}{A(p\log h_n)}, \frac{A(\log(n^2/\lambda_n^2))}{\sqrt{2\tr(\R^2)}}|S_{n,p}^\dagger|\le y+D(\log(n^2/\lambda_n^2))\right) \\
        \to \exp\{-\exp(-x)\}\cdot \exp\{-2\exp(-x)\}.
    \end{aligned}
\end{equation*}
\end{itemize}
\end{theorem}
\begin{remark}
    {\cite{wang2023} established the asymptotic independence between max-$L_{\infty}$-type and sum-$L_2$-type statistics, which converge marginally to the Gumbel and normal distributions, respectively. 
    To the best of our knowledge, this paper is the first to study the asymptotic independence between max-$L_{\infty}$-type and max-$L_2$-type statistics, both of which converge to Gumbel-type limits in high-dimensional settings.
    This advances the theoretical understanding of extreme-value behavior in high dimensions and represents an important contribution to the literature.
    }
\end{remark}
According to Theorem \ref{thm:ind_H0}, we propose combining the individual $p$-values from the max-$L_\infty$ and max-$L_2$-type test statistics using Fisher's method \citep{littell1971asymptotic,littell1973asymptotic}. Specifically, we define the combined $p$-values as
\begin{align*}
    { p}_{M,S} &:= 1 - F_{\chi^2_4}\Big(-2(\log { p}_{M_{n,p}} + \log { p}_{S_{n,p}})\Big)\ \text{~and}\\
    { p}_{M^\dagger,S^\dagger} &:= 1 - F_{\chi^2_4}\Big(-2(\log { p}_{M^{\dagger}_{n,p}} + \log { p}_{S^{\dagger}_{n,p}})\Big),
\end{align*}
where $F_{\chi^2_4}$ denotes the cdf of the chi-squared distribution with 4 degrees of freedom.
The justification for this approach lies in the asymptotic independence of the two types of test statistics under the null hypothesis, as established in Theorem~\ref{thm:ind_H0}. Consequently, both
$-2(\log { p}_{M_{n,p}} + \log { p}_{S_{n,p}})$ and $-2(\log { p}_{M^{\dagger}_{n,p}} + \log { p}_{S^{\dagger}_{n,p}})$ converges in distribution to $F_{\chi^2_4}$ under $H_0$.
Therefore, either ${\rm p}_{M,S}$ or ${\rm p}_{M^{\dagger},S}$ can be used as the final $p$-value for testing $H_0$. If the combined $p$-value is smaller than a pre-specified significance level $\alpha\in(0,1)$, then we reject $H_0$.
The size of the combined test is asymptotically controlled according to Theorem \ref{thm:ind_H0}.

We now turn to analyze the power of the combined test under the local alternative hypothesis:
\begin{equation*}
    H_{1;n,p}:\vert\mathcal{A}\vert= o\{p/(\log\log p)^2\wedge\sqrt{\tr(\R^2)}/\log n\} \text{ ~and~ } \Vert\boldsymbol\delta\Vert^2=o\{n^{-1}\sqrt{2\tr(\R^2)}
    \},
\end{equation*}
where $\mathcal{A}=\{1\leq 
 j\leq p :\delta_j\not= 0\}$ is the support of $\boldsymbol\delta$. 
 
The next theorem establishes that the max-$L_\infty$-type and max-$L_2$-type test statistics remain asymptotically independent under the local alternative.

\begin{theorem}\label{thm:ind_H1}
    Suppose Assumptions \ref{ass:max1}-\ref{ass:max3} and \ref{ass:sum_R4}-\ref{ass:ind1} hold. Under $H_{1;n,p}$, if $\log^7 n=o(p^{1/6\wedge\eta_0/2})$ and $\log^2 p=o(n^{1/5\wedge \lambda/3})$, we have,
    \begin{itemize}
\item[(i)] If $\lambda_n\sim n^{\lambda}$ for some $\lambda\in(0,1)$, then, as $(n,p)\to\infty$, $M_{n,p}$ is asymptotically independent of $S_{n,p}$ in the sense that
\[
    \pr\Big(p^{1/2} \zeta_1 M_{n,p}\leq u_p\{\exp(-x)\}, \frac{S_{n,p}}{\sqrt{2{\tr(\R^2)}}}\leq y\Big) \to \exp\{-\exp(-x)\}\cdot F_V(y);
\]

\item[(ii)] If $\lambda_n\sim n^{\lambda}$ for some $\lambda\in(0,1)$, then, as $(n,p)\to\infty$, $M^{\dagger}_{n,p}$ is asymptotically independent of $S_{n,p}^{\dagger}$ in the sense that

\begin{equation*}
    \begin{aligned}
        \pr\Big(p^{1/2} \zeta_1 M^{\dagger}_{n,p}\leq \frac{x+D(p\log h_n)}{A(p\log h_n)}, \frac{A(\log(n^2/\lambda_n^2))}{\sqrt{2\tr(\R^2)}}|S_{n,p}^\dagger|\le y+D(\log(n^2/\lambda_n^2))\Big) \\
        \to \exp\{-\exp(-x)\}\cdot  \exp\{-2\exp(-x)\}.
    \end{aligned}
\end{equation*}
\end{itemize}
\end{theorem}

{
\begin{remark}
    Theorem \ref{thm:ind_H1} shows the asymptotic independence of the max-$L_\infty$-type and max-$L_2$-type test statistics under the local alternative $H_{1;n,p}$. 
    Notably, the signal strength conditions required under $H_{1;n,p}$ in our setting are more restrictive than those in \citet{wang2023}. This is mainly because, unlike the sample mean with its explicit additive form, the spatial median and spatial sign require a Bahadur representation for asymptotic analysis. However, this expansion relies on the assumption of i.i.d.~symmetric data \citep{feng2016multivariate, cheng2023statistical}. Under strong signals, structural changes break this symmetry, causing the spatial median to diverge from the mean and invalidating the expansion. Therefore, we focus on the local alternative regime, where the signal is weak enough that the Bahadur representation remains approximately valid, ensuring analytical tractability.
\end{remark}
}
Based on Theorem \ref{thm:ind_H1}, we compare the power of the adaptive tests to their non-adaptive counterparts.
Let $M$ denote either $M_{n,p}$ or $M_{n,p}^\dagger$, and $S$ denote either $S_{n,p}$ or $S_{n,p}^\dagger$, with corresponding $p$-values $p_M$ and $p_S$. For a given significance level $\alpha \in (0,1)$, let $\beta_{M,\alpha}$ and $\beta_{S,\alpha}$ be the power functions of $M$ and $S$, respectively. According to \citet{littell1971asymptotic,littell1973asymptotic}, the power of Fisher’s combination test is comparable to that of the minimal $p$-value test, $\min\{p_M, p_S\}$, with power function $\beta_{M \wedge S, \alpha} = \pr(\min\{p_M, p_S\} \leq 1 - \sqrt{1 - \alpha})$. On one hand, we have the bound
\begin{equation}\label{eq:power_H1}
\begin{aligned}
\beta_{M\wedge S, \alpha} &\ge \pr(\min\{{p}_{M},{ p}_{S}\}\leq \alpha/2)\\
&= \beta_{M,\alpha/2}+\beta_{S,\alpha/2}-\pr({p}_{M}\leq \alpha/2, {p}_{S}\leq \alpha/2)\\
&\ge \max\{\beta_{M,\alpha/2},\beta_{S,\alpha/2}\}.
\end{aligned}
\end{equation}
On the other hand, under the local alternative $H_{1;n,p}$, the asymptotic independence of $M$ and $S$ in Theorem~\ref{thm:ind_H1} yields
\begin{equation}\label{eq:power_H1np}
    \beta_{M\wedge S, \alpha} \ge \beta_{M,\alpha/2}+\beta_{S,\alpha/2}-\beta_{M,\alpha/2}\beta_{S,\alpha/2}+o(1),
\end{equation}
For small $\alpha$, the difference between $\beta_{M,\alpha}$ and $\beta_{M,\alpha/2}$ (and similarly for $S$) is small. Therefore, \eqref{eq:power_H1} and \eqref{eq:power_H1np} suggest that the adaptive test achieves power at least comparable to, and often exceeding, that of the individual max-$L_\infty$-type or max-$L_2$-type tests. Similar discussions can be found in \citet{wang2023}.

\begin{remark}\label{rem:ada_esti}
    Similar to \cite{wang2023}, when the null hypothesis is rejected, we propose two adaptive changepoint estimation methods by combining $L_\infty$-type and $L_2$-type statistics:
 \begin{equation}
    \notag \hat{\tau} :=\begin{cases}\hat{\tau}_M:=\underset{\lambda_n \leq k \leq n-\lambda_n}{\arg \max } \|\C_{0}(k)\|_{\infty}, & \text { if } p_{M_{n, p}}<p_{S_{n, p}}, \\ \hat{\tau}_S:=\underset{\lambda_n \leq k \leq n-\lambda_n}{\arg \max } \|\tilde{\C}_{0}(k)\|^2, & \text { otherwise, }\end{cases}
\end{equation}
or
\begin{equation}
    \notag \hat{\tau}^\dagger :=\begin{cases}\hat{\tau}_{M^{\dagger}}:=\underset{\lambda_n \leq k \leq n-\lambda_n}{\operatorname{argmax}} \|\C_{0.5}(k)\|_{\infty}, & \text { if } p_{M_{n, p}^{\dagger}}<p_{S^\dagger_{n, p}}, \\ \hat{\tau}_{S^\dagger}:=\underset{\lambda_n \leq k \leq n-\lambda_n}{\arg \max } \|\tilde{\C}_{0.5}(k)\|^2, & \text { otherwise }\,.\end{cases}
\end{equation}
These estimators adaptively choose between the $L_\infty$-based and $L_2$-based statistics based on which corresponding $p$-value provides stronger evidence against the null. Notably, they replace the conventional mean-based CUSUM statistic \citep{wang2023} with a spatial-sign-based CUSUM statistic.
\end{remark}

\section{Simulation studies}\label{Sec:Simus}

To evaluate the performance of the proposed spatial-median and spatial-sign-based methods, we conduct a series of simulation studies to assess test size, power, and changepoint estimation accuracy, with respect to sample size $n$, dimension $p$, signal strength $\boldsymbol\delta$, sparsity level and noise distribution. 
We include a broad range of competing methods for comparison:
\begin{itemize}
    \item Our proposed tests, with $p$-values ${ p}_{M_{n,p}}$, ${ p}_{M_{n,p}^\dagger}$, ${ p}_{S_{n,p}}$, ${ p}_{S^{\dagger}_{n,p}}$, ${ p}_{M,S}$, and ${ p}_{M^\dagger,S^\dagger}$,  referred to as SMAX(0), SMAX(0.5), SSUM(0), SSUM(0.5), SCMS(0) and SCMS(0.5);
    \item The max-$L_2$-aggregation methods proposed by \cite{2013-chan+Horvath+Huskova-p955} and \cite{2016-jin+pan+yang+zhou-p2355}, referred to as CHH and JPYZ, respectively;
    \item The double-max-sum methods proposed by \cite{wang2023}, referred to as 
    DMS(0) and DMS(0.5). 
    \item The adaptive procedures in \cite{liu2020unified} over $q\in \{1,2,3,4,5,\infty\}$ with $s_0=p/2$ and \cite{zhang2022adaptive} over $q\in \{2,6\}$, referred to as LZZL and ZWS, respectively.
\end{itemize}
In particular, SMAX(0), SMAX(0.5), SSUM(0), SSUM(0.5), SCMS(0), SCMS(0.5), CHH, DMS(0.5), and LZZL require a boundary removal parameter. For a fairness comparison, we set this parameter to $\lambda_n:=\lfloor 0.2 n\rfloor$ for all methods. For our proposed spatial-sign-based tests, we set $\varrho=0.2$ when estimating $\zeta_1$ and $\D$. 

The following scenarios are considered for random noises:
\begin{itemize}
    \item \textbf{I}: Multivariate normal distribution with mean zero and covariance matrix $\mathbf\Sigma$.
    \item \textbf{II}: Multivariate $t$-distribution with degrees of freedom $6$ and covariance matrix $\mathbf\Sigma$. 
     \item \textbf{III}: Multivariate mixture normal distribution with pdf $\gamma f_{p}(\mathbf{0}, \mathbf\Sigma) + (1 - \gamma) f_{p}(\mathbf{0}, 9\mathbf\Sigma)$, where $f_{p}(\cdot; \cdot)$ is the density function of $p$-dimensional multivariate normal distribution, and $\gamma$ is set to $0.8$.
\end{itemize}
In all scenarios, the covariance matrix is specified as $\mathbf\Sigma=(0.5^{\vert j-\ell\vert})_{1\leq j,\ell\leq p}$. 
Each method's empirical size, power, and changepoint estimation accuracy are evaluated over 500 Monte Carlo replications, with a nominal significance level of $\alpha = 5\%$.

\subsection{Size performance}

To evaluate the size performance, we consider $n=200$ with $p\in \{100,200,300,400\}$ for illustration. Table \ref{tab:size1} presents the size of each test for different $ (n, p) $ under Scenarios I--III.  It is evident that our proposed tests--SMAX(0), SMAX(0.5), SSUM(0), SSUM(0.5), SCMS(0), and SCMS(0.5)--maintain good control over the Type I error rate as $ (n, p) $ increases.
Most of the other methods also demonstrate good Type I error control, with the exception of the CHH method, which exhibits inflation in the Type I error rate.
This inflation is due to the CHH method being primarily designed for normally distributed data with independent components, failing to adapt to other distributions and correlations between dimensions. In contrast, our proposed method allows for heavy-tailed distributions and takes into account the correlations between dimensions.

\begin{table}[!htp]
\centering
\footnotesize
\begin{tabular}{lcccccc}
\toprule
$(n,p)$ & SMAX(0)& SSUM(0)& SCMS(0)& SMAX(0.5)& SSUM(0.5)& SCMS(0.5)\\
\midrule
\multicolumn{7}{l}{\textbf{Scenario (I)}}\\ 
   (200,100) &5.0&6.4&7.4&4.0&1.4&3.6\\ 

   (200,200) &6.0&5.4&7.2&4.6&1.0&3.0\\ 
   (200,300) &5.4&4.2&5.8&3.2&1.2&2.8\\ 
   (200,400) &5.4&5.0&5.4&4.8&0.4&3.0\\ 
 \midrule
\multicolumn{7}{l}{\textbf{Scenario (II)}}\\ 
   (200,100) &4.6&6.8&8.4&4.0&0.8&3.0\\ 

   (200,200) &4.2&6.2&7.2&4.8&0.4&3.2\\ 
   (200,300) &4.2&4.2&6.0&3.6&0.8&2.2\\ 
   (200,400) &4.6&4.0&6.0&4.2&0.6&3.4\\  
  \midrule
\multicolumn{7}{l}{\textbf{Scenario (III)}}\\ 
   (200,100) &5.0&8.6&8.8&4.6&1.6&4.8\\ 

   (200,200) &4.6&8.2&7.8&5.6&1.4&5.0\\ 
   (200,300) &4.6&5.2&6.6&4.0&0.6&2.0\\ 
   (200,400) &4.2&2.6&4.0&3.2&0.2&1.6\\ 

 \bottomrule
 \toprule
 $(n,p)$ & JPYZ & CHH & DMS(0)& DMS(0.5)& LZZL& ZWS\\
\midrule
\multicolumn{7}{l}{\textbf{Scenario (I)}}\\ 
   (200,100) &10.2&10.6&8.6&7.6&7.0&5.0\\ 
   (200,200) &7.6&8.6&7.8&6.6&4.8&6.0\\ 
   (200,300) &5.8&8.6&8.3&8.0&6.2&7.2\\ 
   (200,400) &6.2&7.0&4.8&3.6&3.6&6.2\\  
 \midrule
\multicolumn{7}{l}{\textbf{Scenario (II)}}\\  
   (200,100) &6.6&10.6&6.0&6.4&3.6&5.6\\ 

   (200,200) &3.6&15.8&4.8&5.0&2.8&7.2\\ 
   (200,300) &2.8&19.2&3.8&3.4&3.4&6.0\\ 
   (200,400) &2.4&18.6&5.0&4.8&4.8&5.6\\ 

 \midrule
\multicolumn{7}{l}{\textbf{Scenario (III)}}\\ 
   (200,100) &4.4&14.8&3.8&4.0&3.4&8.0\\ 

   (200,200) &2.0&21.2&5.2&3.4&4.0&7.0\\ 
   (200,300) &1.2&26.0&3.0&3.0&2.2&6.8\\ 
   (200,400) &0.8&32.8&2.8&4.4&4.2&5.2\\ 
\bottomrule
\end{tabular}
\caption{Empirical size (in $\%$) performance under Scenarios I--III.}\label{tab:size1}
\end{table}

\subsection{Power performance}

To evaluate the power performance across different levels of sparsity under alternatives, we consider $\delta_{j}=\sqrt{\Delta/k}$ for $j=1,2,\ldots,k$ and $\delta_{j}=0$ otherwise, such that $\Vert \boldsymbol\delta\Vert^2=\Delta$. Figures \ref{fig:power1_tau05}--\ref{fig:power1_tau025} present the empirical power of different methods for varying signal strength $\Delta$, signal sparsity levels $k$, and changepoint locations $\tau$, with $(n,p)=(200,200)$ for illustration. 

In Scenario I, the ensemble methods LZZL and ZWS show a slight advantage when $ \tau/n = 0.5 $, while the DMS(0.5) method performs better when $\tau/n = 0.25$. However, in Scenarios II and III, these ensemble methods exhibit a faster power decay as $k$ increases, and their performance is significantly inferior to that of the SCMS(0) and SCMS(0.5) methods. As expected, the spatial-sign-based methods demonstrate significantly higher power compared to other approaches for heavy-tailed data. Notably, the two adaptive methods, SCMS(0) and SCMS(0.5), perform well across various sparsity levels. When $\tau/n = 0.5$, SCMS(0) achieves outstanding performance compared to all other methods. Moreover, even when $\tau/n = 0.25$, i.e., the changepoint is closer to the boundary, SCMS(0) still outperforms SCMS(0.5). This is primarily due to the slower convergence rate of the statistic in SSUM(0.5), which hinders its ability to take advantage of the statistic after scaling, thereby affecting the performance of the adaptive method. This warrants further investigation.

\begin{figure}[!ht]
\centering
\begin{subfigure}[b]{0.9\textwidth}
        \centering
        \includegraphics[width=\textwidth]{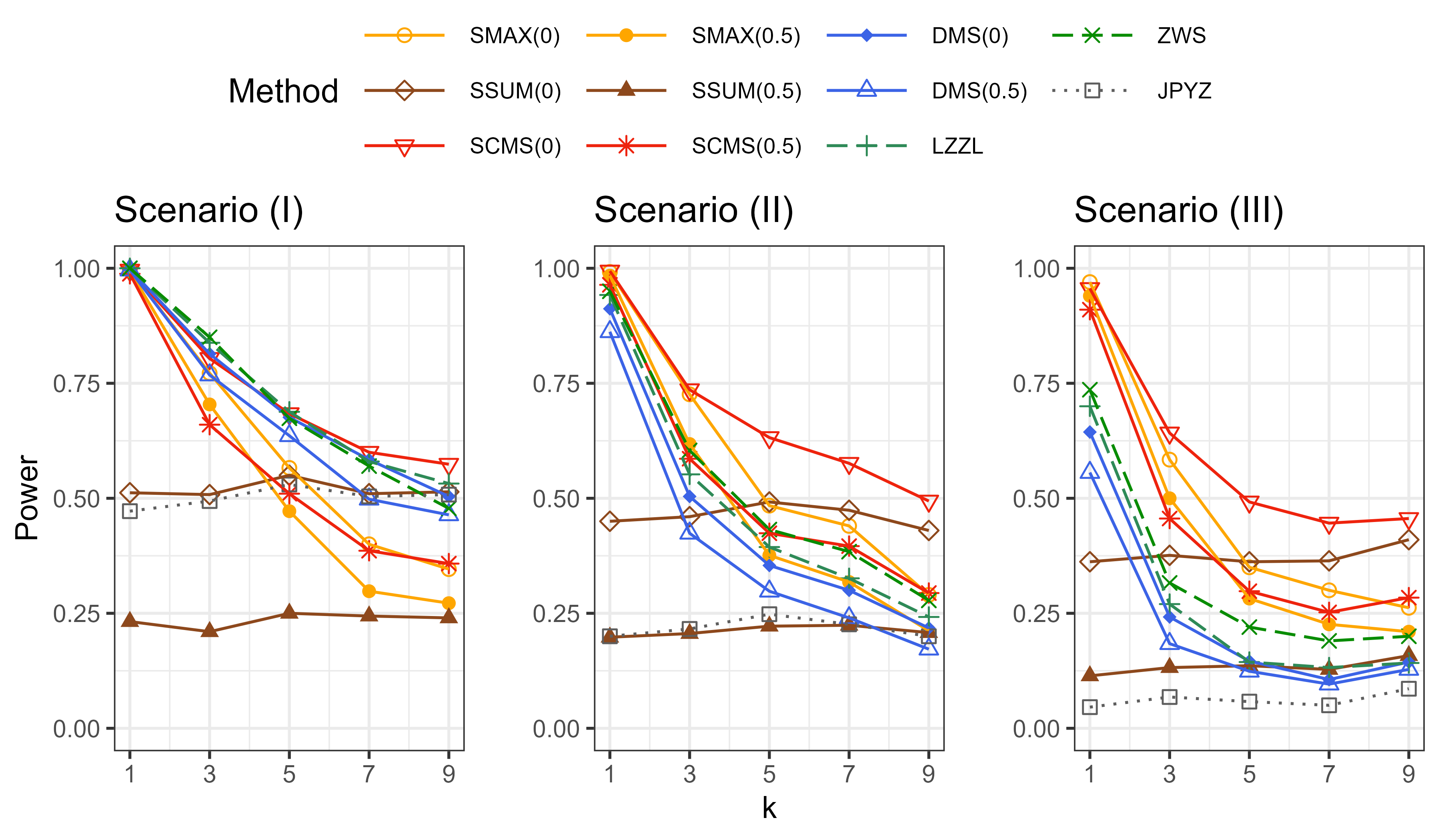}
        \caption{$\Delta=1$}
    \end{subfigure}
    \begin{subfigure}[b]{0.9\textwidth}
        \centering
        \includegraphics[width=\textwidth]{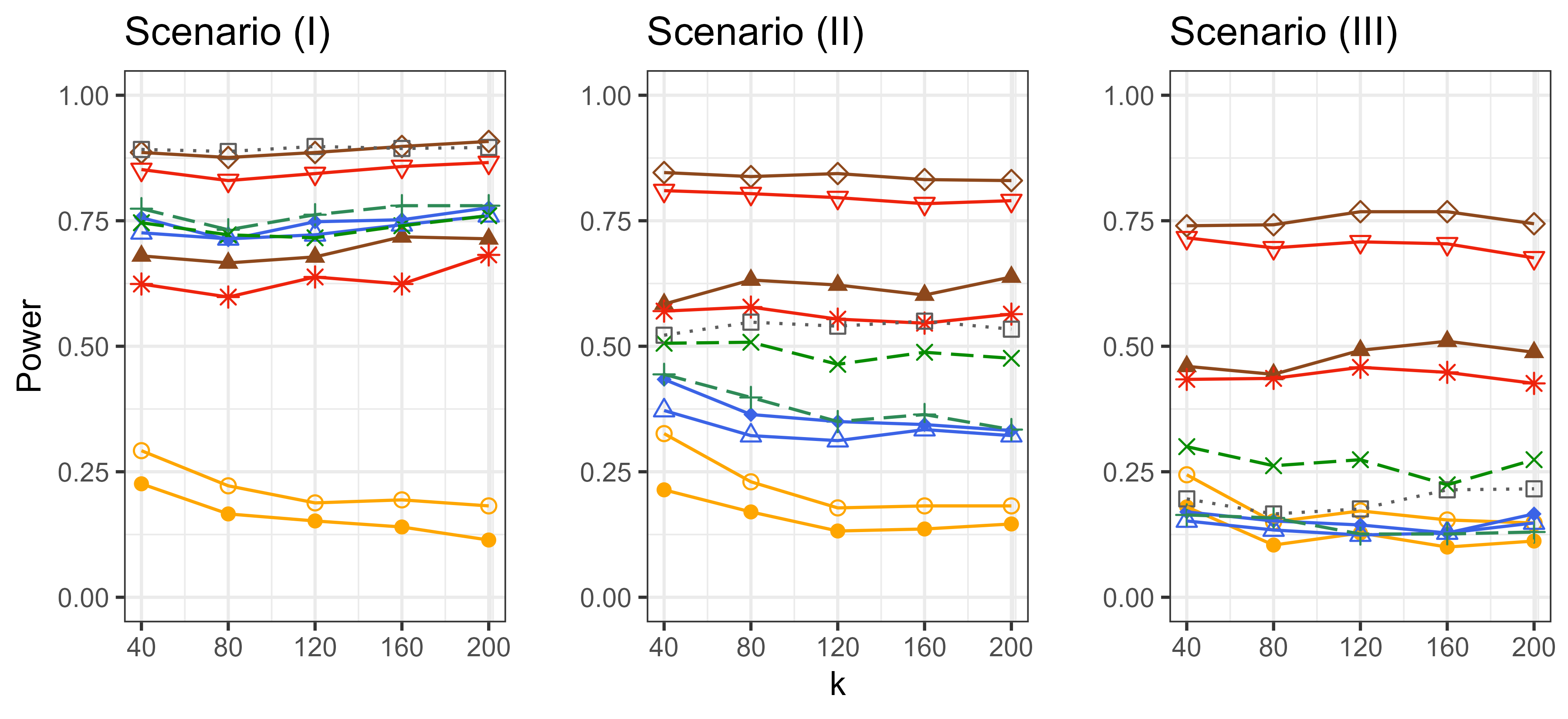}
        \caption{$\Delta=2$}
    \end{subfigure}
\caption{Power of tests with different signal strength $\Delta$, signal sparsity levels $k$, and changepoint locations $\tau$ for Scenarios I--III with $(n,p)=(200,200)$ and $\tau/n=0.5$. \label{fig:power1_tau05}}
\end{figure}
\begin{figure}[!ht]
\centering
\begin{subfigure}[b]{0.9\textwidth}
        \centering
        \includegraphics[width=\textwidth]{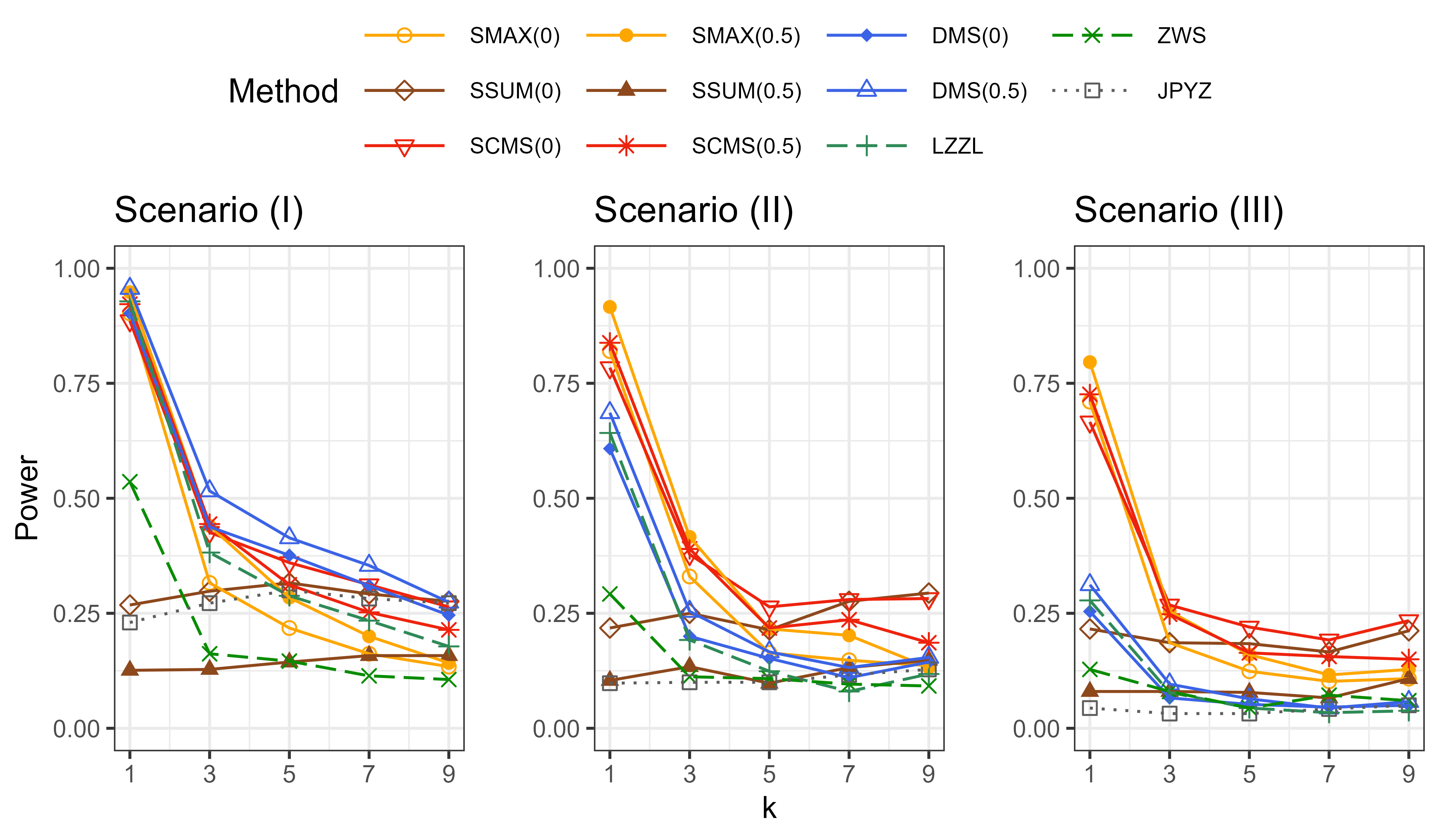}
        \caption{$\Delta=1$}
    \end{subfigure}
    \begin{subfigure}[b]{0.9\textwidth}
        \centering
        \includegraphics[width=\textwidth]{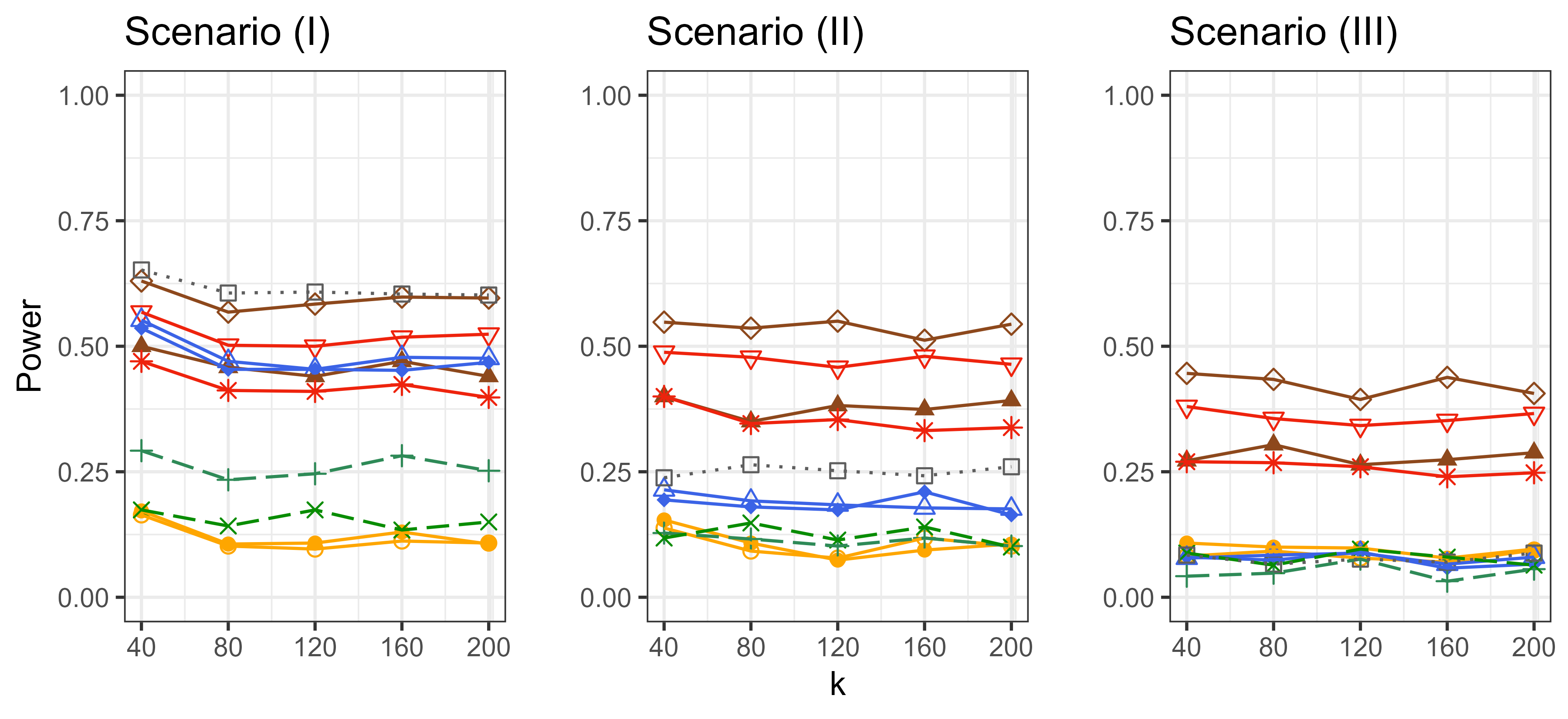}
        \caption{$\Delta=2$}
    \end{subfigure}
\caption{Power of tests with different signal strength $\Delta$, signal sparsity levels $k$, and changepoint locations $\tau$ for Scenarios I--III with $(n,p)=(200,200)$ and $\tau/n=0.25$. \label{fig:power1_tau025}}
\end{figure}

\subsection{Estimation accuracy}

We next evaluate the accuracy of single changepoint estimation. We consdier the spatial-sign based methods: SMAX(0) - $\hat{\tau}_M$, SSUM(0) - $\hat{\tau}_S$, SCMS(0) - $\hat{\tau}$,  SMAX(0.5) - $\hat{\tau}_{M^\dagger}$, SSUM(0.5) - $\hat{\tau}_{S^\dagger}$, SCMS(0.5) - $\hat{\tau}^\dagger$. 
For comparison, we also implement several procedures recommended in \cite{wang2023}: MAX(0), MAX(0.5), SUM(0.5), DMS(0), and DMS(0.5).

\begin{figure}[!ht]
\centering
\begin{subfigure}[b]{0.9\textwidth}
        \centering
        \includegraphics[width=\textwidth]{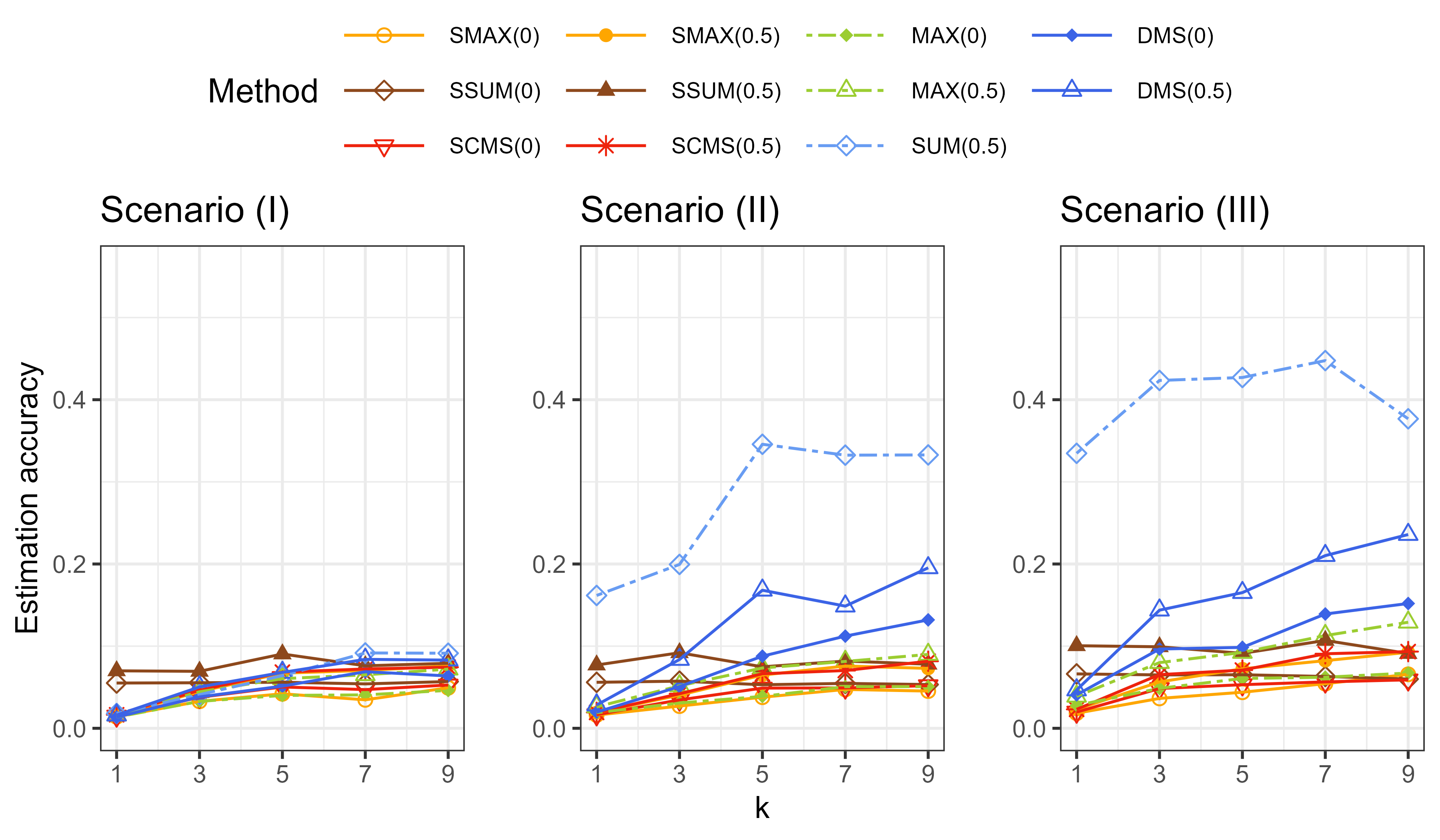}
        \caption{$\Delta=1$}
    \end{subfigure}
    \begin{subfigure}[b]{0.9\textwidth}
        \centering
        \includegraphics[width=\textwidth]{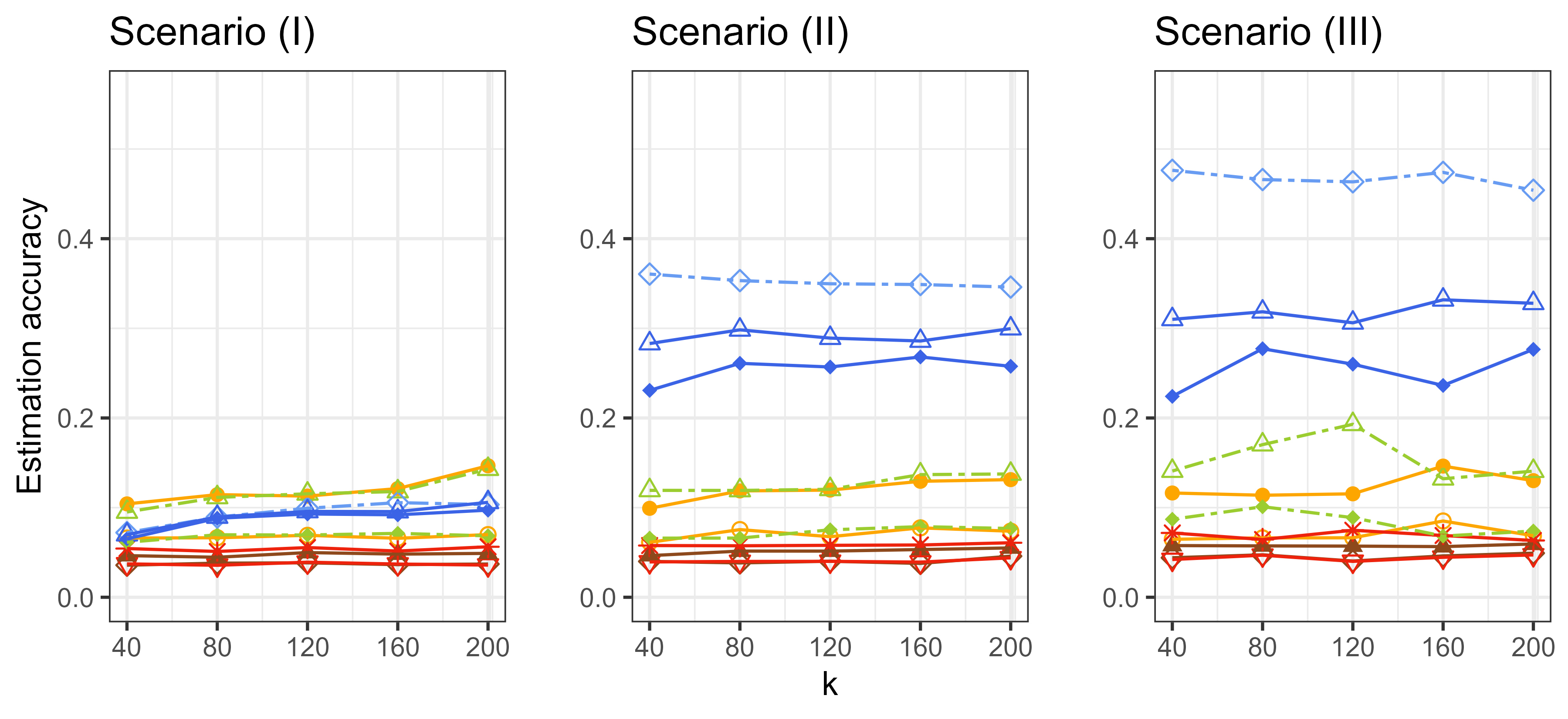}
        \caption{$\Delta=2$}
    \end{subfigure}
\caption{Comparison of changepoint estimation accuracy with different signal strength $\Delta$, signal sparsity levels $k$, and changepoint locations $\tau$ for 
 Scenarios I--III with $(n,p)=(200,200)$ and $\tau/n=0.5$. \label{fig:acc_tau05}}
\end{figure}
\begin{figure}[!ht]
\centering
\begin{subfigure}[b]{0.9\textwidth}
        \centering
        \includegraphics[width=\textwidth]{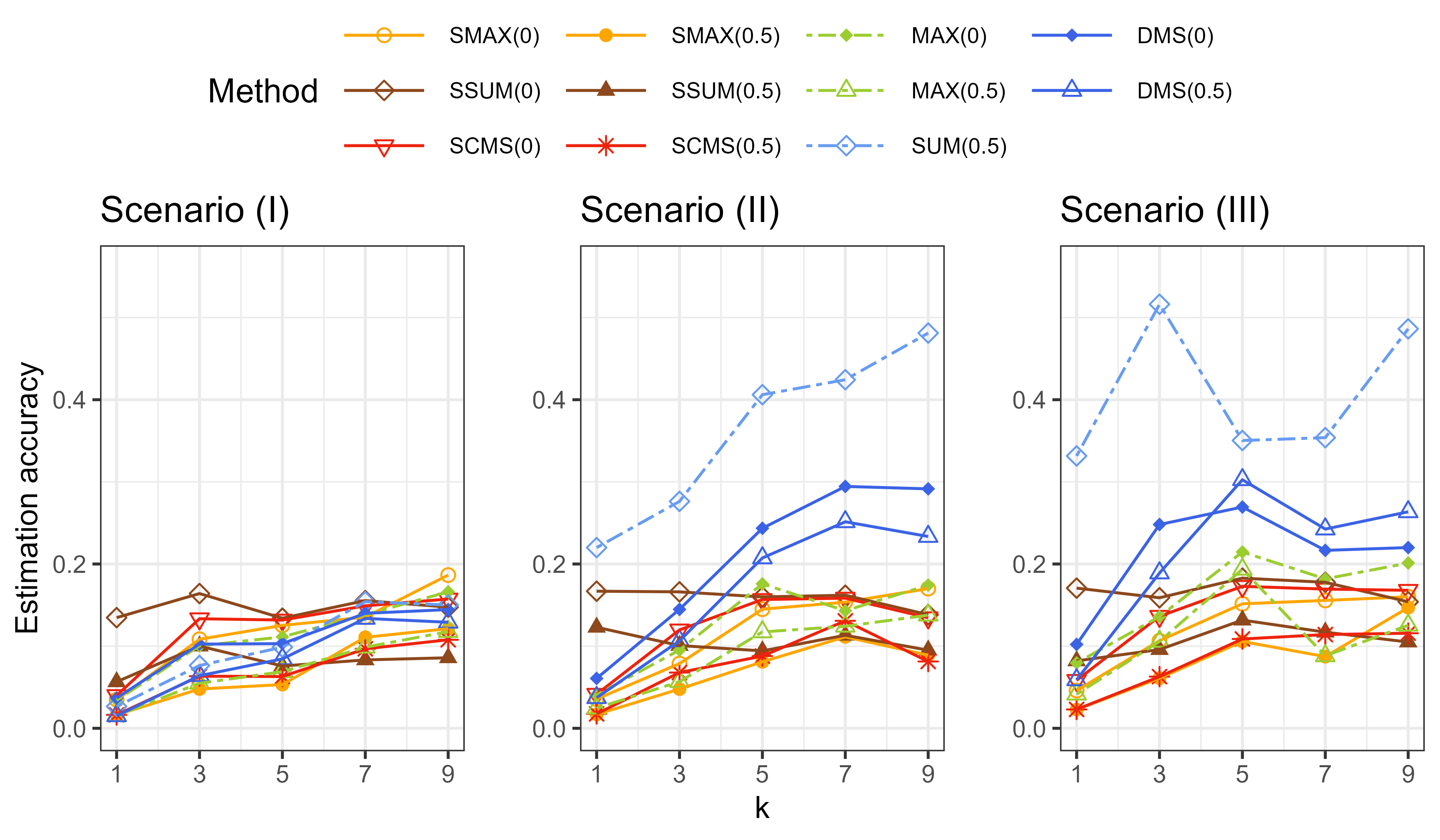}
        \caption{$\Delta=1$}
    \end{subfigure}
    \begin{subfigure}[b]{0.9\textwidth}
        \centering
        \includegraphics[width=\textwidth]{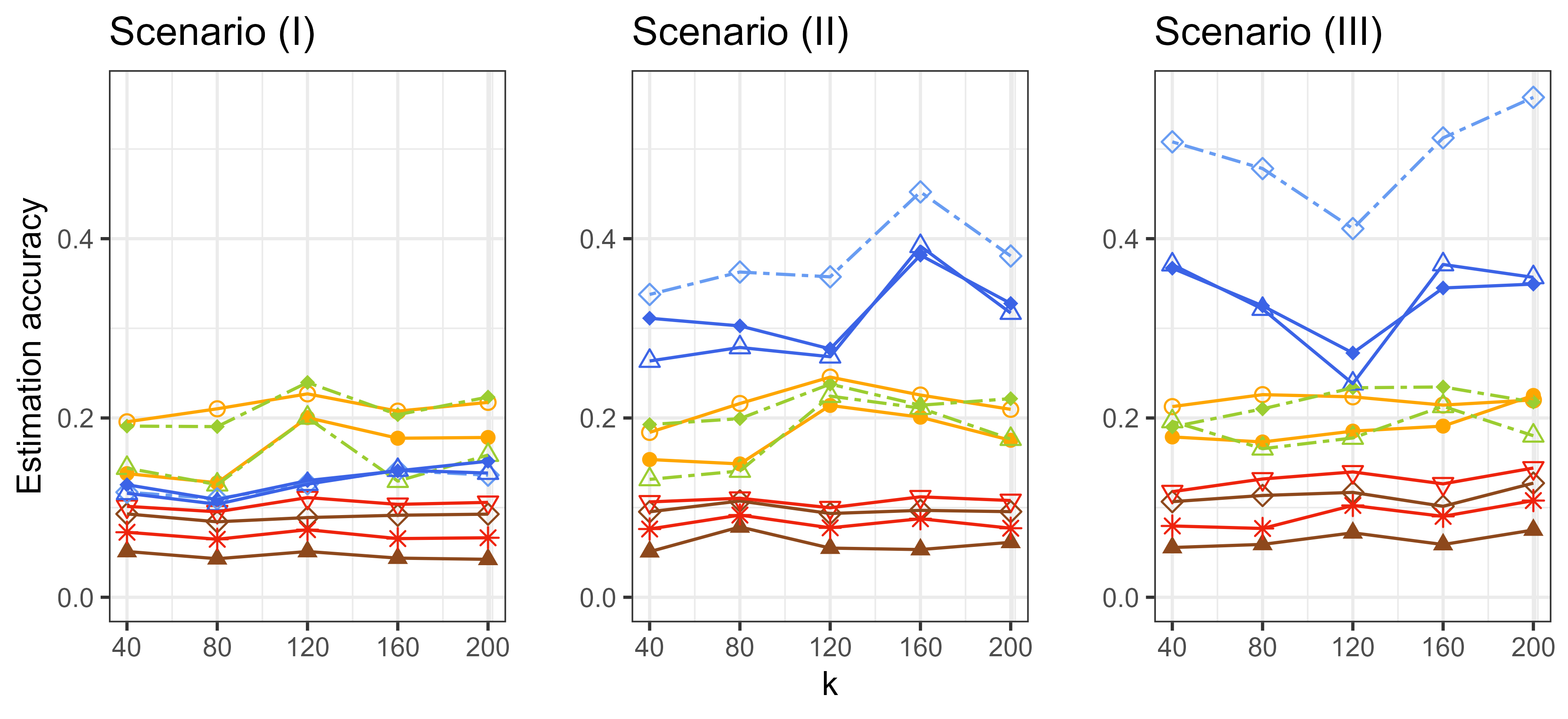}
        \caption{$\Delta=2$}
    \end{subfigure}
\caption{Comparison of changepoint estimation accuracy with different signal strength $\Delta$, signal sparsity levels $k$, and changepoint locations $\tau$ for 
 Scenarios I--III with $(n,p)=(200,200)$ and $\tau/n=0.25$. \label{fig:acc_tau025}}
\end{figure}

Figures \ref{fig:acc_tau05}--\ref{fig:acc_tau025} present the estimation accuracy, defined as the absolute distance between the estimated and true changepoints, scaled by the sample size $n$. 
It is observed that max-type methods are more effective in sparse settings, whereas sum-type methods perform better in dense scenarios. Adaptive methods demonstrate consistent accuracy across different levels of sparsity. When the changepoint is near the center of the sequence, SCMS(0) yields smaller errors, while SCMS(0.5) outperforms SCMS(0) when the changepoint is closer to the boundary. Similar trends are observed for both SMAX and SSUM methods. Notably, under the normality assumption, i.e., Scenario I, the SSUM(0.5) and SSUM(0) methods exhibit superior performance in dense signal settings for $\tau/n=0.25$ and $\tau/n=0.5$, respectively. In sparse signal scenarios, the max-type method shows a slight advantage in Scenario I.
Under heavy-tailed or mixture distributions (Scenarios II and III), the spatial-sign-based methods, particularly SMAX and SCMS, outperform the other methods.

\section{Real data applications}\label{Sec:Real data}
\subsection{US stocks data}

We begin with an analysis of financial data from the Standard \& Poor’s 500 Index (S\&P 500), a widely used benchmark in economics, finance, and statistics.
Comprising 500 large publicly traded companies across diverse sectors, this index reflects overall market trends and is sensitive to macroeconomic conditions, policy shifts, and investor sentiment.
As such, historical S\&P 500 data have been widely used in studies of market volatility, asset pricing, portfolio optimization, and financial risk management.

In this paper, we analyze daily closing prices of the S\&P 500 constituent stocks over the period from January 2019 to October 2024.
Weekly return rates were computed, resulting in 294 observations per stock during this period.
To ensure data consistency, we first excluded companies not continuously listed throughout the entire period, yielding a dataset of 486 stocks.
The weekly return rates were then standardized.
Recognizing the potential presence of autocorrelation in return rates, we applied the Ljung–Box test \citep{ljung1978measure} at the 5\% significance level to test whether each stock exhibited zero autocorrelation. Based on this, 340 stocks were retained for further analysis. It is worth noting that including all 486 stocks would have introduced autocorrelation into the dataset, potentially violating our model assumptions and necessitating further investigation.

Table \ref{tab:SP500} summarizes the $p$-values for testing changepoints in the weekly return rates.
At the 5\% significance level, the DMS(0), DMS(0.5), and LZZL tests fail to reject the null hypothesis. 
In contrast, both SCMS(0) and ZWS yield significantly small $p$-values, leading to a rejection of the null hypothesis and indicating a significant change in weekly return rates. SCMS(0.5) also suggests potential evidence of change, producing a $p$-value close to the significance threshold.
Notably, the max-type tests, SMAX(0) and SMAX(0.5), also detect a significant change, whereas the sum-type tests, SSUM(0) and SSUM(0.5), fail to reject the null. These divergent results imply that the underlying change in weekly return rates is likely sparse rather than dense.

\begin{table}[!htp]
\centering
\footnotesize
\begin{tabular}{ccccc}
\toprule
 SMAX(0)& SSUM(0)& SCMS(0)& SMAX(0.5)& SSUM(0.5) \\
\midrule
0.0049 & 0.2044 & 0.0079 & 0.0197 & 0.4963 \\
\bottomrule
\toprule
SCMS(0.5) & DMS(0) & DMS(0.5) & LZZL & ZWS \\
\midrule
 0.0550 & 0.9041 & 0.9241 &0.6287& 0.0187 
\\
\bottomrule
\end{tabular}
\caption{The $p$-values for testing changepoints in weekly return rates.}\label{tab:SP500}
\end{table}

\subsection{Array comparative genomic hybridization data}

We then analyze an array comparative genomic hybridization (aCGH) dataset, which is used to detect DNA sequence copy number variations in individuals with bladder tumors. The dataset, available in the R package \texttt{ecp}, consists of log-transformed fluorescence intensity ratios of DNA segments across $n=2215$ loci for $p = 43$ individuals.

We apply the changepoint testing procedures to the aCGH dataset and observe that all methods yield significantly small $p$-values, indicating the presence of at least one changepoint.
To localize the changepoints, we adopt the binary segmentation approach used in \citet{liu2020unified,wang2023}.
Specifically, for any interval $[l,r]$, where $l$ and $r$ are integers satisfying $1\leq l<r\leq n$, we first apply the adaptive test to assess the presence of a changepoint. If the null is rejected, we estimate the changepoint location $t$ using the adaptive procedure described in Remark~\ref{rem:ada_esti}, and then divide the interval $[l,r]$ into two subintervals: $[l,t]$ and $[t,r]$. This procedure is recursively applied to each subinterval until no further changepoints are detected.

Following the setup in \citet{liu2020unified,wang2023}, we set $\gamma=0.5$, the boundary parameter $\lambda_n = 40$, and the nominal significance level at 5\%.
The number of detected changepoints by SMAX(0.5), SSUM(0.5), SCMS(0.5), SMAX(0), SSUM(0), and SCMS(0) are 43, 41, 41, 40, 42, and 42, respectively. 
For illustration, Figure~\ref{fig:aCGH} displays the changepoints estimated by SCMS(0.5), which closely align with findings in previous studies \citep{matteson2014nonparametric,liu2020unified,wang2023}, demonstrating the effectiveness of the proposed procedure.

\begin{figure}[!ht]
\centering
\includegraphics[width=\textwidth]{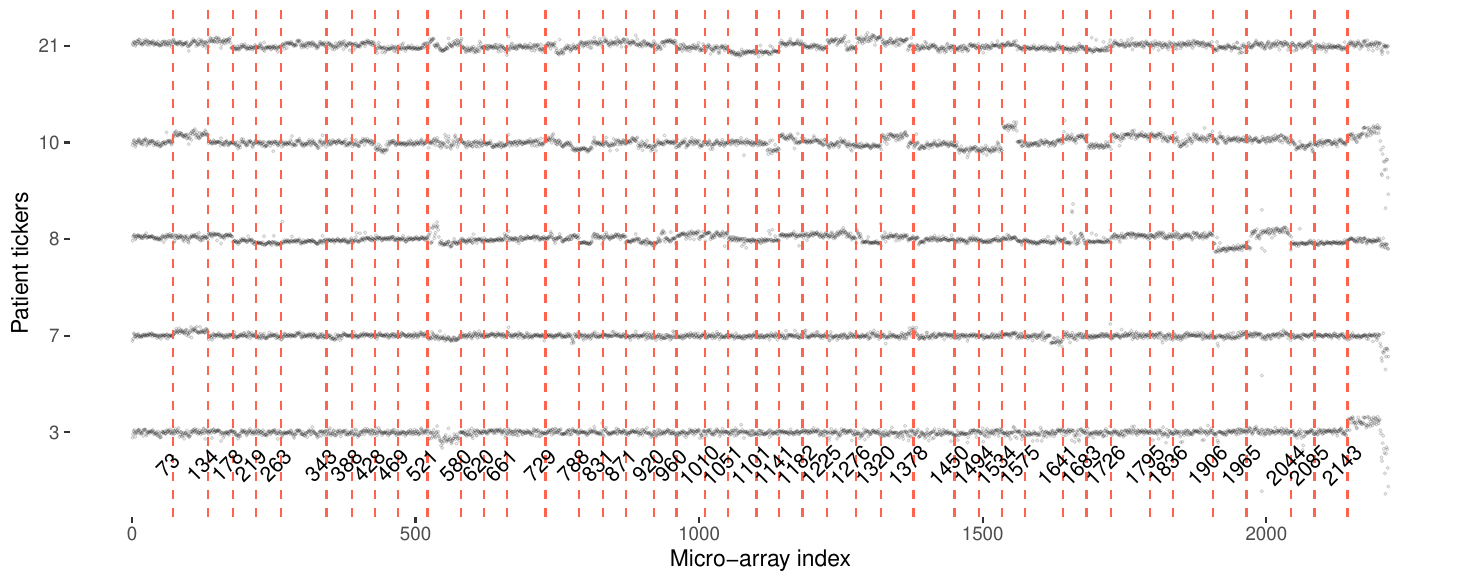}
\caption{Changepoint estimation in the aCGH data using the SCMS(0.5) method with binary segmentation.}\label{fig:aCGH}
\end{figure}

\section{Concluding remarks}\label{Sec:Con Remarks}

This paper introduces a robust and adaptive framework for high-dimensional changepoint detection, particularly suited to heavy-tailed data. Based on spatial medians and spatial signs, we construct max-$L_\infty$-type tests for sparse signals and max-$L_2$-type tests for dense signals. We derive their asymptotic null distributions and establish their asymptotic independence under mild conditions. Building on this, we develop adaptive testing procedures by combining the two test types via Fisher’s method, offering strong power across varying levels of signal sparsity.

Several avenues for future work remain. First, our theoretical results rely on the i.i.d.~assumption. Extending these to dependent settings \citep{chang2024central} is challenging but promising.
Second, our max-$L_2$-type tests consider spatial directions but omit radius information, which has been shown to improve power in other contexts \citep{feng2021inverse,huang2023high}. Incorporating radius-based features while preserving asymptotic properties is an important extension. 
Lastly, enhancing adaptive estimation strategies to accommodate multiple changepoints or structured dependencies may broaden real-world applicability.

\appendix
\setcounter{figure}{0}
\setcounter{table}{0}   
\setcounter{equation}{0}
\setcounter{lemma}{0}

\renewcommand{\thetable}{S\arabic{table}}
\renewcommand{\thefigure}{S\arabic{figure}}
\renewcommand{\theequation}{S\arabic{equation}}
\renewcommand{\thelemma}{S\arabic{lemma}}

\section{Additional numerical studies}
\subsection{Comparison with the mean-based max-$L_2$-type testing}\label{simu:max-L2}
Recall that the spatial-sign based max-$L_2$-type statistics are defined as $S_{n,p}=\max_{1\leq k\leq n}\|\tilde{\C}_{0}(k)\|^2$ and $S_{n,p}^\dagger=\max_{\lambda_n\leq k\leq n-\lambda_n}\|\tilde{\C}_{0.5}(k)\|^2$ if we ignore some constants. We also introduce the mean-based max-$L_2$-type methods with 
$$\breve{C}_{\gamma,j}(k) = \left\{ \frac{k}{n}(1-\frac{k}{n}) \right\}^{-\gamma}\frac{1}{\sqrt{n}}(\breve{S}_{kj}-\frac{k}{n}\breve{S}_{nj})/\check{\sigma}_j,$$
where $\breve{S}_{kj}=\sum_{i=1}^k X_{ij}$ and $\check{\sigma}_j$ is Bartlett's estimators, also used in \cite{wang2023}. Further, $\tr(\R^2)$ can be estimated by $$\widetilde{\tr(\R^2)} = \frac{1}{4(n-3)}\sum_{i=1}^{n-3}\left\{ 
(\X_i-\X_{i+1})^\top \check{\D}^{-1}_{(i,i+1,i+2,i+3)}(\X_{i+2}-\X_{i+3}) \right\}^2,$$ where for any $(i_1,i_2,\ldots,i_m)\subset \{1,2,\ldots,n\}$ with $m\geq 1$, $$\check{\D}_{(i,i+1,i+2,i+3)}=\diag \{\check{\sigma}^2_{1(i_1,\ldots,i_m)},\ldots,\check{\sigma}^2_{p(i_1,\ldots,i_m)}\},$$ and $\check{\sigma}^2_{j(i_1,\ldots,i_m)}=\{2\vert \mathcal{A}_m \vert\}^{-1}\sum_{i\in \mathcal{A}_m}(X_{ij}-X_{i-1,j})^2$ with $\mathcal{A}_m=\{2,3,\ldots,n\}\setminus \{i_1,i_2,\ldots,i_m\}$ for $j=1,2,\ldots,p$, the ratio consistency is shown in \cite{wang+zou+wang+yin-2019-Multiple}. Accordingly, we term them as MSUM(0) when $\gamma=0$ and MSUM(0.5) when $\gamma=0.5$. Similarly, we define the adaptive methods by combining the corresponding $p$-values using Fisher's method.  To wit,
\begin{align*}
    { p}_{MCMS(0)} &:= 1 - F_{\chi^2_4}\Big(-2(\log { p}_{MAX(0)} + \log { p}_{MSUM(0)})\Big)\ \text{and}\\
    { p}_{MCMS(0.5)} &:= 1 - F_{\chi^2_4}\Big(-2(\log { p}_{MAX(0.5)} + \log { p}_{MSUM(0.5)})\Big),
\end{align*}
We term the two adaptive methods as MCMS(0) and MCMS(0.5) respectively.

Figure \ref{fig:power2_tau05}-\ref{fig:power2_tau025} 
present the power comparison for spatial-sign based methods -- SMAX(0), SMAX(0.5), SSUM(0), SSUM(0.5), mean-based max-$L_2$-methods -- MSUM(0), MSUM(0.5), max-$L_\infty$-methods -- MAX(0), MAX(0.5)\citep{wang2023} and sum-$L_2$-method -- SUM(0.5) \citep{wang+zou+wang+yin-2019-Multiple} and corresponding adaptive methods SCMS(0), SCMS(0.5), MCMS(0), MCMS(0.5) and DMS(0), DMS(0.5) methods. The size performance of MSUM(0), MSUM(0.5), MCMS(0) and MCMS(0.5) are shown in Table \ref{tab:size2}.
\begin{figure}[!ht]
\centering
\begin{subfigure}[b]{0.9\textwidth}
        \centering
        \includegraphics[width=\textwidth]{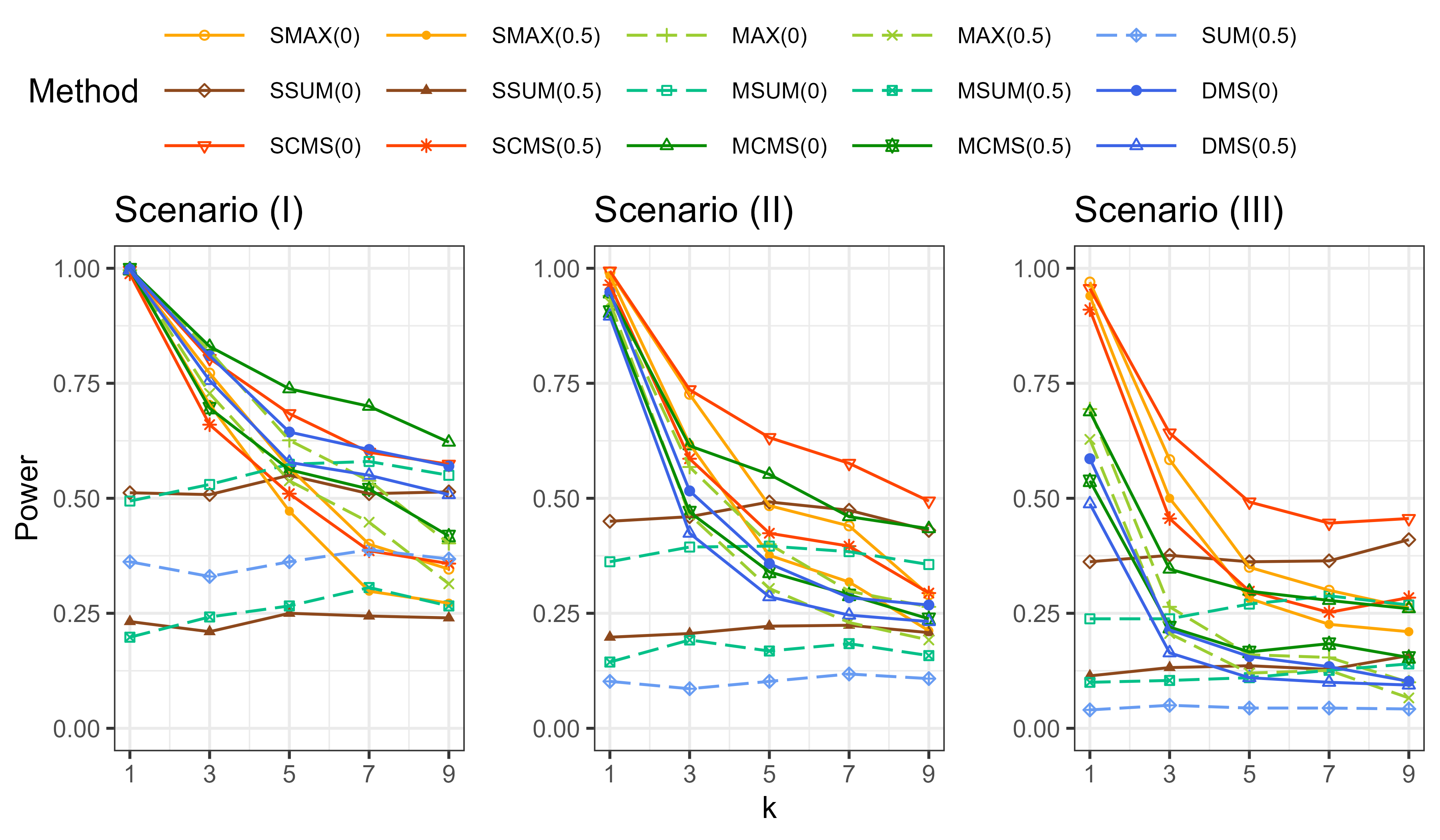}
        \caption{$\Delta=1$}
    \end{subfigure}
    \begin{subfigure}[b]{0.9\textwidth}
        \centering
        \includegraphics[width=\textwidth]{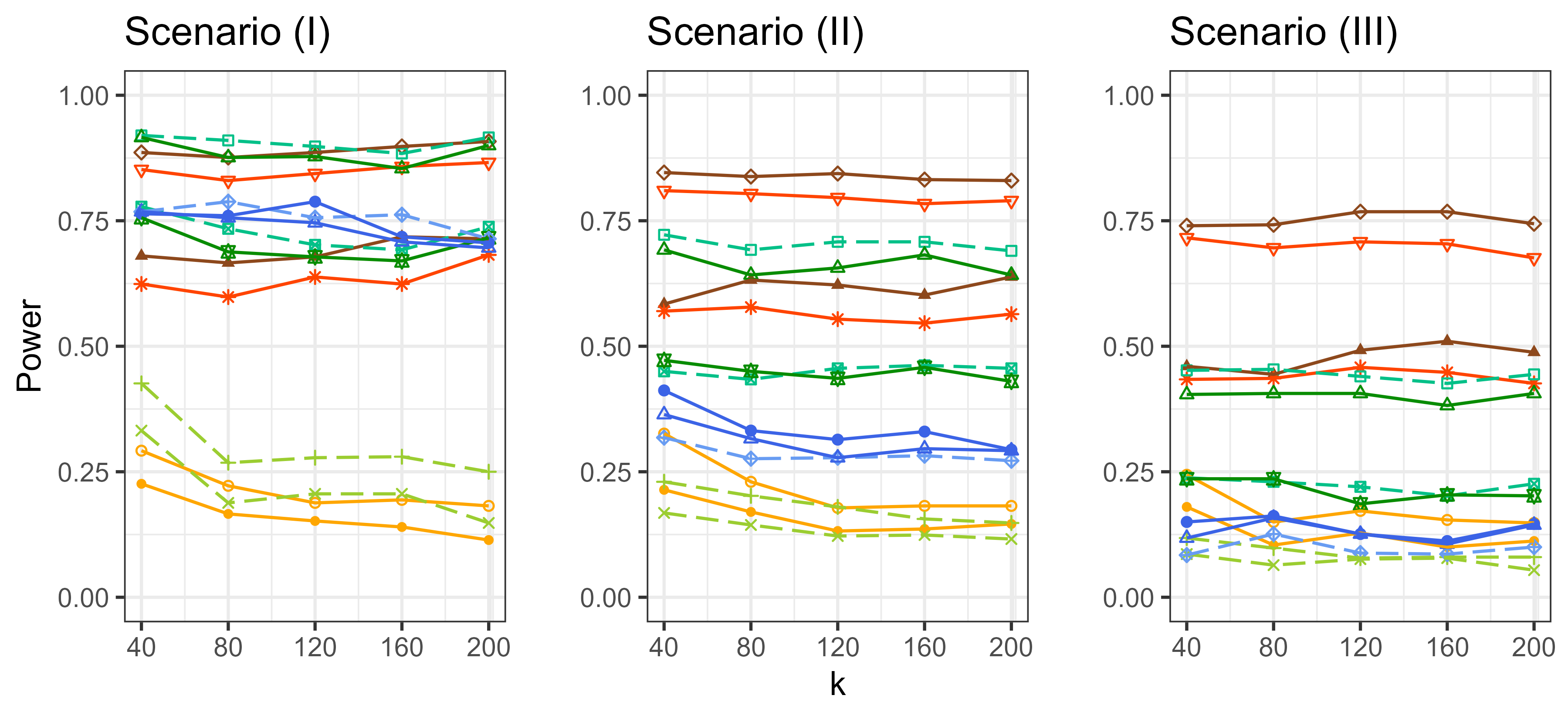}
        \caption{$\Delta=2$}
    \end{subfigure}
\caption{Comparison of the power of max-$L_2$-aggregation and spatial-sign based max-$L_2$-type method with different signal strength for Scenarios I to III over $(n,p)=(200,200)$ and $\tau/n=0.5$. \label{fig:power2_tau05}}
\end{figure}
\begin{figure}[!ht]
\centering
\begin{subfigure}[b]{0.9\textwidth}
        \centering
        \includegraphics[width=\textwidth]{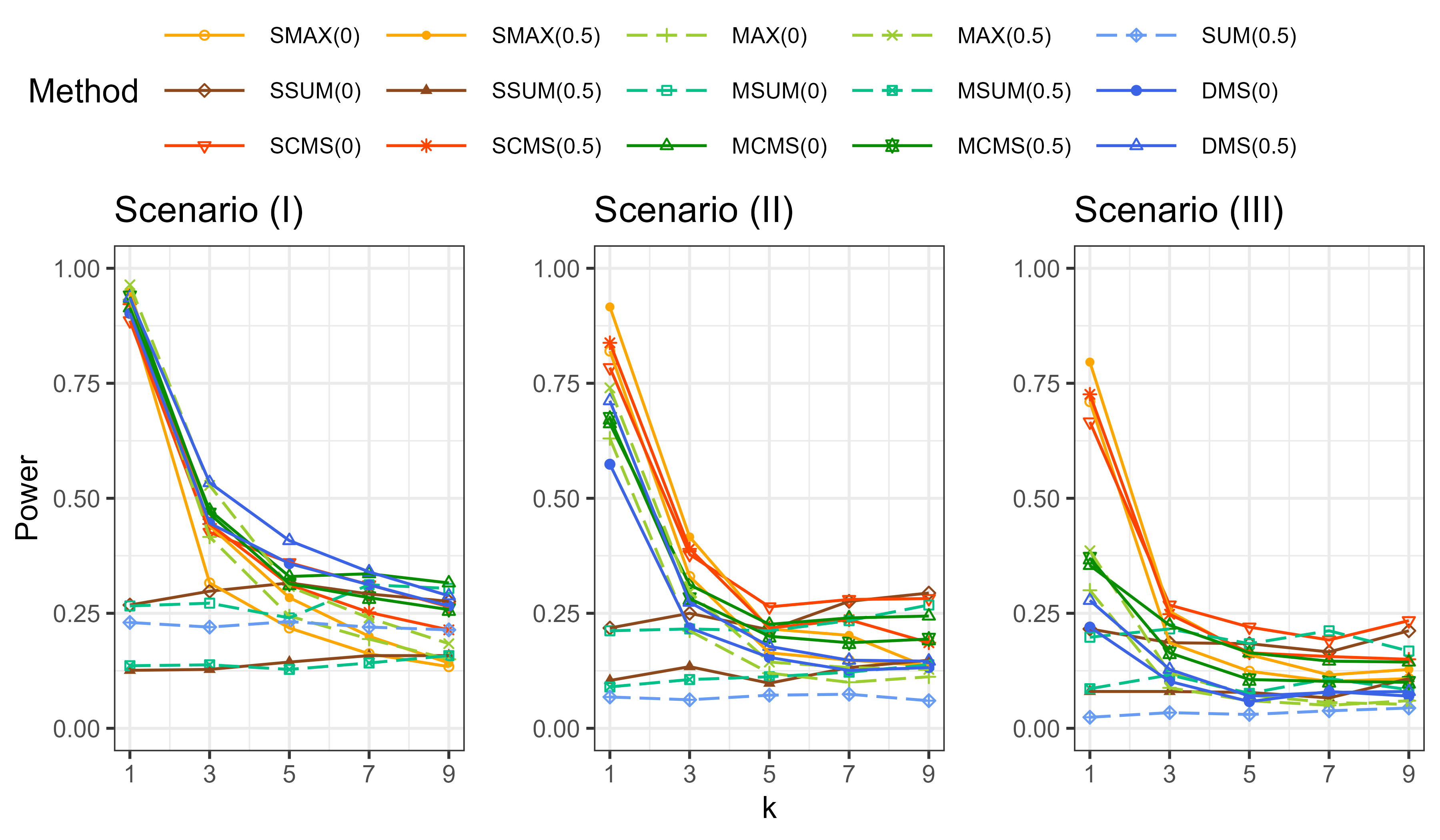}
        \caption{$\Delta=1$}
    \end{subfigure}
    \begin{subfigure}[b]{0.9\textwidth}
        \centering
        \includegraphics[width=\textwidth]{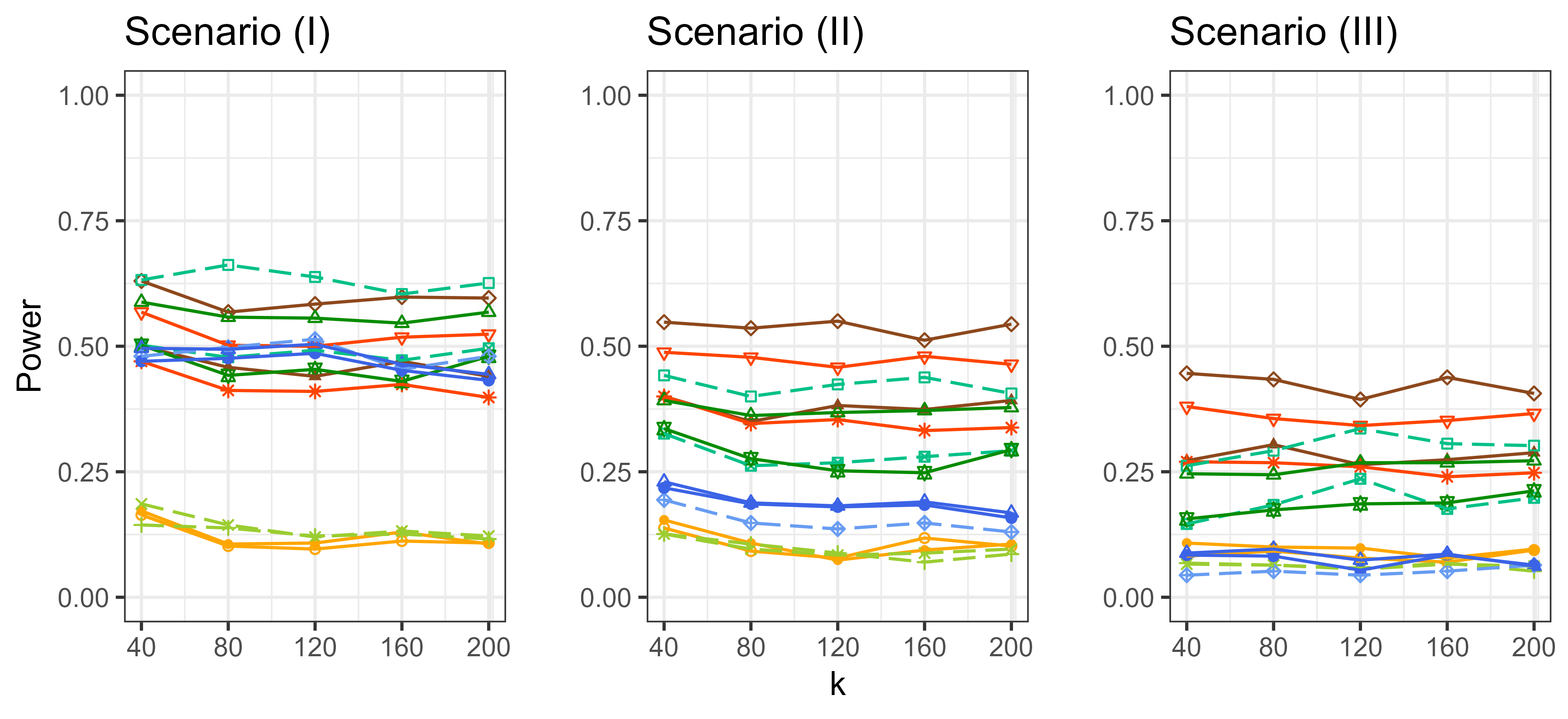}
        \caption{$\Delta=2$}
    \end{subfigure}
\caption{Comparison of the power of max-$L_2$-aggregation and spatial-sign based max-$L_2$-type method with different signal strength for Scenarios I to III over $(n,p)=(200,200)$ and $\tau/n=0.25$. \label{fig:power2_tau025}}
\end{figure}
\begin{table}[!ht]
\centering
\footnotesize
\begin{tabular}{lcccc}
\toprule
$(n,p)$ & MSUM(0)& MCMS(0)&  MSUM(0.5)& MCMS(0.5)\\
\midrule
\multicolumn{5}{l}{\textbf{Scenario (I)}}\\ 
 (200,100)&9.8&11.8&3.6&4.4\\
 (200,200) &8.4&9.4&1.4&5.0\\ 
 (200,300) &7.0&8.6&1.6&5.4\\ 
 (200,400) &7.4&8.8&1.6&3.6\\ 
 \midrule
\multicolumn{5}{l}{\textbf{Scenario (II)}}\\ 
 (200,100)&12.8&12.6&5.8&8.2\\
 (200,200) &9.4&10.0&3.6&5.4\\ 
 (200,300) &13.6&11.8&5.6&7.0 \\
 (200,400) &13.8&13.4&6.2&6.0\\ 
 \midrule
 \multicolumn{5}{l}{\textbf{Scenario (III)}}\\ 
 (200,100)&10.8&9.8&5.0&5.2\\
 (200,200) &14.8&12.0&7.0&7.0\\ 
 (200,300) &16.4&13.6&9.4&10.0\\ 
 (200,400) &21.0&17.2&9.6&10.8\\ 
\bottomrule
\end{tabular}
\caption{Empirical size(in $\%$) performance under Scenarios I to III for max-$L_2$-aggregation methods}\label{tab:size2}
\end{table}

It can be seen that max-$L_\infty$-type methods outperform max-$L_2$-type methods at sparse signal levels, while they fall behind under moderate and dense signal levels. Adaptive methods, on the other hand, demonstrate competitive performance across all levels, which is consistent with the findings in \cite{wang2023}.  We also observe that the MCMS methods perform exceptionally well across Scenarios I--III, consistently achieving higher power than the DMS methods. However, it is worth noting that when $n$ and $p$ are relatively small or the data deviates from normality, the MSUM method shows some inflation in size, which warrants further investigation. Notably, spatial sign-based methods clearly outperform others when the data deviates from normality, highlighting their robustness to heavy-tailed distributions.

\section{Proofs}
In this section, we provide the proofs of all the theorems presented in the paper, along with the main lemmas required for their proofs. We introduce some notations.

Denote $a_n\lesssim b_n$ if there exists constant $C$, $a_n\leq C b_n$ and $a_n \asymp b_n$ if  both $a_n\lesssim b_n$ and $b_n\lesssim a_n$ hold. Let $\psi_{\alpha_0}(x)=\exp \left(x^{\alpha_0}\right)-1$ be a function defined on $[0, \infty)$ for $\alpha_0>0$. Then the Orlicz norm $\|\cdot\|_{\psi_{\alpha_0}}$ of a $\boldsymbol X$ is defined as $\|\boldsymbol X\|_{\psi_{\alpha_0}}=\inf \left\{t>0, \E\left\{\psi_{\alpha_0}(|\boldsymbol X| / t)\right\} \leqslant 1\right\}$. For $d$-dimensional vector $\boldsymbol x=(x_1,\ldots,x_p)^\top$, denote its Euclidean norm and maximum-norm as $\Vert \boldsymbol x\Vert$ and $\Vert \boldsymbol x\Vert_\infty$ respectively. The spatial sign function is defined as $U(\boldsymbol x)=\Vert \boldsymbol x\Vert^{-1}\boldsymbol x\ind{\boldsymbol x\neq 0}$. In particular, the $ i$th component of $U(\boldsymbol{x})$ is given by $U(\boldsymbol x)_i=\Vert \boldsymbol x\Vert^{-1}x_i$, $i=1,\ldots,p$. Let $\operatorname{tr}(\cdot)$ be a trace for matrix, $\lambda_{min}(\cdot)$ and $\lambda_{max}(\cdot )$ be the minimum and maximum eigenvalue for symmetric martix. For a symmetric matrix $\A=(a_{ij})_{p\times p}$, we denote $\Vert \A\Vert_1=\Vert \A\Vert_\infty=\max_{1\leq j\leq p}\sum_{i=1}^p\vert a_{ij}\vert$, $\Vert \A\Vert_F=\left\{\tr(\A^2)\right\}^{1/2}$. $\mathbf I_p$ represents a $p$-dimensional identity matrix, and $ \operatorname{diag}\{v_1,v_2,\ldots,v_p\}$ represents the diagonal matrix with entries $\boldsymbol v=(v_1,v_2,\ldots,v_p)^{\top}$. 

Recall that, for a sequence of $p$-dimensional random noises $\{\boldsymbol\epsilon_i=\nu_i\mathbf\Gamma\W_i\in \mathbb R^p\}_{i=1}^n$, $\W_i=(W_{i,1},\ldots,W_{i,p})^{\top}$, the $\boldsymbol U_i=U(\mathbf D^{-1/2}\boldsymbol\epsilon_i)=(U_{i,1},\ldots,U_{i,p})^\top$ and $R_i=\Vert \mathbf D^{-1/2}\boldsymbol\epsilon_i\Vert$ are the scale-invariant spatial-sign and radius of the random noise is $\boldsymbol\epsilon_i$, respectively, where $\D$ is a diagonal matrix $\D=\diag\{d_1^2,\ldots,d_p^2\}$. The $(\bth,\D)$-estimated version of $\U_i$ is $\hat{\U}_i=U(\hat{\D}^{-1/2}(\X_i-\hat{\bth}_{1:n}))$. The moments of $R_i^{-k}$ are $\zeta_k=\E (R_i^{-k})$, $k=1,2,3,4$.

\subsection{Proof of main lemmas}

\begin{lemma}\label{lemma:UWi1}
    Under Assumption \ref{ass:max1}, we have for any $1\leq l\neq k\leq p$, \\
    (i) $\E \{U(\W_i)_l^2 \}=p^{-1}$; and \\(ii) $\E \{U(\W_i)_l U(\W_i)_k\}=O(p^{-5/2}).$
\end{lemma}

\begin{proof}
    (i) By symmetry, all components of $U(\W_i)$ have the same marginal distribution. Since $\sum_{j=1}^p U(\W_i)_j^2 = U(\W_i)^{\top}U(\W_i) = 1$, we have
    \[
        \E\{U(\W_i)_l^2\} = p^{-1} \E\left\{ \sum_{j=1}^p U(\W_i)_j^2 \right\} = p^{-1},
    \]
    for any $1\leq l\leq p$.

(ii) Let 
\[
\mathcal{A}_{1i} = \left\{ p - \varsigma p^{(1+\eta_0)/2} \leq \|\W_i\|^2 \leq p + \varsigma p^{(1+\eta_0)/2} \right\},
\]
for some fixed $0 < \varsigma < 1$. Using Lemmas \ref{LemmaA1}--\ref{LemmaA2}, Assumption \ref{ass:max1}, and the inequality 
\[
\frac{1}{p(p-1)} \sum_{1 \leq l \neq k \leq p} W_{i,l} W_{i,k} \leq \frac{1}{p} \sum_{j=1}^p W_{i,j}^2,
\]
we obtain
\begin{align*}
& \quad \E \{U(\W_i)_l U(\W_i)_k\} = \E\left\{\frac{W_{i,l} W_{i,k}}{\|\W_i\|^2}\right\} \\
&= \E\left\{ \frac{1}{p(p-1)} \sum_{1 \leq l \neq k \leq p} \frac{W_{i,l} W_{i,k}}{\|\W_i\|^2} \right\} \\
&= \E\left\{ \frac{1}{p(p-1)} \sum_{1 \leq l \neq k \leq p} W_{i,l} W_{i,k} \left( \|\W_i\|^{-2} - \frac{1}{p} \right) \right\} \\
&= -p^{-1} \E\left\{ \frac{1}{p(p-1)} \sum_{1 \leq l \neq k \leq p} W_{i,l} W_{i,k} \|\W_i\|^{-2} (\|\W_i\|^2 - p) \right\} \\
&= -p^{-1} \E\left\{ \frac{1}{p(p-1)} \sum_{1 \leq l \neq k \leq p} W_{i,l} W_{i,k} \|\W_i\|^{-2} (\|\W_i\|^2 - p) \ind{\mathcal{A}_{1i}} \right\} \\
&\quad - p^{-1} \E\left\{ \frac{1}{p(p-1)} \sum_{1 \leq l \neq k \leq p} W_{i,l} W_{i,k} \|\W_i\|^{-2} (\|\W_i\|^2 - p) \ind{\mathcal{A}_{1i}^c} \right\} \\
& \leq p^{-1}\{p - \varsigma p^{(1+\eta_0)/2}\}^{-1} \left[ \E \left\{ \frac{1}{p(p-1)} \sum_{1 \leq l \neq k \leq p} W_{i,l} W_{i,k} \right\}^2 \right]^{1/2} \left\{ \E (\|\W_i\|^2 - p)^2 \right\}^{1/2} \\
& \quad + p^{-2}\E\left| \|\W_i\|^2 - p \right| \ind{\mathcal{A}_{1i}^c} \\
& = p^{-1} \{p-\varsigma p^{(1+\eta_0)/2}\}^{-1}\left\{p(p-1)\right\}^{-1/2} O(p^{1/2})+p^{-2}O(p^{1/2}) c_1^{1/2}\exp\left\{ -c_2 p^{\eta_0\alpha_0 /(4\alpha_0+4)} \right\}\\
& = O(p^{-5/2})\,.
\end{align*}
We finish the proof of this lemma.
\end{proof}

\begin{lemma}\label{lemma1_like_2016}
    Under Assumption \ref{ass:max1}, for any nonrandom symmetric matrix $\mathbf{M}$, we have \\
    (i) $\E \left[\{ U(\bm W_i)^\top \mathbf{M}  U(\boldsymbol{W}_i)\}^2\right]=O\{p^{-2}\tr(\mathbf{M}^\top \mathbf{M})\}$; \\
    (ii)$\E \left[\{ U(\bm W_i)^\top \mathbf{M}  U(\boldsymbol{W}_i)\}^4\right]=O\{p^{-4}\tr^2(\mathbf{M}^\top \mathbf{M})\}$; and \\
    (iii) $\E \left[\{ U(\bm W_i)^\top \mathbf{M}  U(\boldsymbol{W}_i)\}^8\right]=O\{p^{-8}\tr^4(\mathbf{M}^\top \mathbf{M})\}$.
\end{lemma}
\begin{proof}
(i) 
    By Cauchy–Schwarz inequality and Assumption \ref{ass:max1}, we have 
\begin{equation}\label{eq:UW}
    \begin{aligned}
    &\E \left\{U(\bm W_i)_l^2 U(\bm W_i)_k^2\right\} \leq p^{-2}\E\left\{\sum_{s=1}^p\sum_{t=1}^p U(\bm W_i)_s^2 U(\bm W_i)_t^2\right\}=p^{-2},\\
    &\E \left\{U(\bm W_i)_l^4\right\} \leq p^{-1}\E\left\{\sum_{s=1}^p U(\bm W_i)_s^4\right\} \leq p^{-1}\E\left\{\sum_{s=1}^p\sum_{t=1}^p U(\bm W_i)_s^2 U(\bm W_i)_t^2\right\}= p^{-1},
    \end{aligned}
\end{equation}
and consequently
$$
\E\left\{ U(\bm W_i)_l U(\bm W_i)_k U(\bm W_i)_s U(\bm W_i)_t\right\} \leq \sqrt{\E\left\{U(\bm W_i)_l^2 U(\bm W_i)_k^2\right\} \E\left\{U(\bm W_i)_s^2 U(\bm W_i)_t^2\right\}}\leq p^{-2}.
$$

Let $\mathbf{M} = (m_{lk})_{p \times p}$. Using Cauchy–Schwarz again,
$$
\sum_{l,k,s,t}m_{lk}m_{st}\leq \sqrt{\sum_{l,k}m_{lk}^2\sum_{s,t}m_{st}^2}\\
\leq \sqrt{\sum_{l,k}^p m_{lk}^2\sum_{s,t}^pm_{st}^2}=\text{tr}(\mathbf M^\top \mathbf M).
$$
Combining the above,
$$
\begin{aligned}
&\E\left[\left\{U( \W_i)^\top \mathbf{M} U(\W_i)\right\}^2\right]\\
=&\sum_{1\leq l\not= k\leq p} \sum_{1\leq s\not= t\leq p} m_{l k} m_{s t} \E\left\{ U(\bm W_i)_l U(\bm W_i)_k U(\bm W_i)_s U(\bm W_i)_t\right\} +\sum_{l=1}^p \sum_{s=1}^p m_{l l} m_{s s} \E\left\{ U(\bm W_i)_l^2 U(\bm W_i)_s^2 \right\} \\
\leq & p^{-2}\frac{p^4-p^2}{p^4}\text{tr}(\mathbf{M}^\top \mathbf M)+ p^{-1}\frac{p^2}{p^4}\text{tr}(\mathbf M^\top \mathbf M)=O\{p^{-2}\text{tr}(\bf M^\top \bf M)\}.
\end{aligned}
$$

(ii)
Similarly, by Assumption \ref{ass:max1}, we have
\begin{equation*}
    \begin{aligned}
            &\E\left\{ U(\bm W_i)_l^8\right\} \leq p^{-1}\E\left\{\sum_{s=1}^p U(\bm W_i)_s^8\right\} \leq p^{-1}\E\left\{\sum_{s=1}^p U(\bm W_i)_s^2\right\}^4= O(p^{-1}),\\
&\E \left\{U(\bm W_i)_{t_1}^6 U(\bm W_i)_{t_2}^2\right\} \leq O(p^{-2})\E\left\{\sum_{s=1}^p U(\bm W_i)_s^2\right\}^4=O(p^{-2}),\\
&\E \left\{U(\bm W_i)_{t_1}^4 U(\bm W_i)_{t_2}^4\right\} \leq O(p^{-2})\E\left\{\sum_{s=1}^p U(\bm W_i)_s^2\right\}^4=O(p^{-2}),\\
&\E \left\{U(\bm W_i)_{t_1}^4 U(\bm W_i)_{t_2}^2U(\bm W_i)_{t_3}^2\right\} \leq O(p^{-3})\E\left\{\sum_{s=1}^p U(\bm W_i)_s^2\right\}^4=O(p^{-3}),\\
&\E \left\{U(\bm W_i)_{t_1}^2 U(\bm W_i)_{t_2}^2U(\bm W_i)_{t_3}^2U(\bm W_i)_{t_4}^2\right\} \leq O(p^{-4})\E\left\{\sum_{s=1}^p U(\bm W_i)_s^2\right\}^4= O(p^{-4}),\\
    \end{aligned}
\end{equation*}
and by Cauchy inequality, 
\begin{equation*}
    \begin{aligned}
        &\E \left\{U(\bm W_i)_{t_1}^7 U(\bm W_i)_{t_2}\right\} \leq \sqrt{\E \left\{U(\bm W_i)_{t_1}^8\right\}  \E\left\{U(\bm W_i)_{t_1}^6 U(\bm W_i)_{t_2}^2\right\}}=O( p^{-3/2}),\\
&\E \left\{U(\bm W_i)_{t_1}^5 U(\bm W_i)_{t_2}^3 \right\}\leq \sqrt{E\left\{U(\bm W_i)_{t_1}^6 U(\bm W_i)_{t_2}^2\right\}  \E\left\{U(\bm W_i)_{t_1}^4 U(\bm W_i)_{t_2}^4\right\}}=O( p^{-2}),\\
&\E \left\{U(\bm W_i)_{t_1}^5 U(\bm W_i)_{t_2}^2U(\bm W_i)_{t_3}\right\}  \leq \sqrt{\E\left\{U(\bm W_i)_{t_1}^4 U(\bm W_i)_{t_2}^2U(\bm W_i)_{t_3}^2\right\}  \E\left\{U(\bm W_i)_{t_1}^6 U(\bm W_i)_{t_2}^2\right\}}=O( p^{-5/2}).\\
    \end{aligned}
\end{equation*}
Similarly, we calculate the terms and show the results as follows,
\begin{equation*}
    \begin{aligned}
&\E \left\{U(\bm W_i)_{t_1}^5 U(\bm W_i)_{t_2}U(\bm W_i)_{t_3}U(\bm W_i)_{t_4}\right\} \leq O(p^{-5/2}),\\
&\E \left\{U(\bm W_i)_{t_1}^4 U(\bm W_i)_{t_2}^2U(\bm W_i)_{t_3}U(\bm W_i)_{t_4} \right\}\leq O(p^{-3}),\\
&\E \left\{U(\bm W_i)_{t_1}^3 U(\bm W_i)_{t_2}^3U(\bm W_i)_{t_3}U(\bm W_i)_{t_4}\right\} \leq O(p^{-3}),\\
&\E \left\{U(\bm W_i)_{t_1}^3 U(\bm W_i)_{t_2}^2U(\bm W_i)_{t_3}^2U(\bm W_i)_{t_4}\right\} \leq O(p^{-7/2}),\\
&\E \left\{U(\bm W_i)_{t_1}^4 U(\bm W_i)_{t_2}U(\bm W_i)_{t_3}U(\bm W_i)_{t_4}U(\bm W_i)_{t_5}\right\} \leq O(p^{-3}),\\
&\E \left\{U(\bm W_i)_{t_1}^3 U(\bm W_i)_{t_2}^2U(\bm W_i)_{t_3}U(\bm W_i)_{t_4}U(\bm W_i)_{t_5}\right\} \leq O(p^{-7/2}),\\
&\E \left\{U(\bm W_i)_{t_1}^2 U(\bm W_i)_{t_2}^2U(\bm W_i)_{t_3}^2U(\bm W_i)_{t_4}U(\bm W_i)_{t_5}\right\} \leq O(p^{-4}),\\
&\E \left\{U(\bm W_i)_{t_1}^3 U(\bm W_i)_{t_2}U(\bm W_i)_{t_3}U(\bm W_i)_{t_4}U(\bm W_i)_{t_5} U(\bm W_i)_{t_6}\right\}\leq O(p^{-7/2}),\\
&\E \left\{U(\bm W_i)_{t_1}^2 U(\bm W_i)_{t_2}^2U(\bm W_i)_{t_3}U(\bm W_i)_{t_4}U(\bm W_i)_{t_5} U(\bm W_i)_{t_6}\right\}\leq O(p^{-4}),\\
&\E \left\{U(\bm W_i)_{t_1}^2 U(\bm W_i)_{t_2}U(\bm W_i)_{t_3}U(\bm W_i)_{t_4}U(\bm W_i)_{t_5} U(\bm W_i)_{t_6}U(\bm W_i)_{t_7}\right\}\leq O(p^{-4}),\\
&\E \left\{U(\bm W_i)_{t_1} U(\bm W_i)_{t_2}U(\bm W_i)_{t_3}U(\bm W_i)_{t_4}U(\bm W_i)_{t_5} U(\bm W_i)_{t_6}U(\bm W_i)_{t_7}U(\bm W_i)_{t_8}\right\}\leq O(p^{-4}),\\
\end{aligned}
\end{equation*}
where  $t_1,t_2,\ldots ,t_8\in \{1,2,\ldots,p\}$ are not equal and $l\in \{1,2,\ldots,p\}$.

 By Cauchy inequality, 
\begin{equation*}
    \begin{aligned}
 &\sum_{i_1,i_2,i_3,i_4=1}^p\sum_{j_1,j_2,j_3,j_4=1}^p m_{i_1 j_1}m_{i_2 j_2}m_{i_3 j_3}m_{i_4 j_4}\\
 \leq & \frac{1}{4}\sum_{i_1,i_2,i_3,i_4=1}^p\sum_{j_1,j_2,j_3,j_4=1}^p (m_{i_1 j_1}^2+m_{i_2 j_2}^2)(m_{i_3 j_3}^2+m_{i_4 j_4}^2)\\
 =& \sum_{i_1,i_2,j_1,j_2=1}^p m_{i_1,j_1}^2 m_{i_2,j_2}^2=\text{tr}^2(\mathbf M^\top \mathbf M).
    \end{aligned}
\end{equation*}
Thus, we get
$$
\begin{aligned}
&\E\left[\left\{U(\bm W_i)^\top \mathbf{M} U(\bm W_i)\right\}^4\right]\\
=&\sum_{l_1,l_2,\ldots,l_8=1}^p  m_{l_1 l_2} m_{l_3 l_4} m_{l_5 l_6}m_{l_7 l_8}\E\left\{ U(\bm W_i)_{l_1} \cdots U(\bm W_i)_{l_8}\right\} +\sum_{l=1}^p \sum_{s=1}^p m_{l l} m_{s s} \E\left\{ U(\bm W_i)_l^2 U(\bm W_i)_s^2 \right\} \\
\leq & p^{-4}\frac{p^8-O(p^6)}{p^8}\text{tr}^2(\mathbf{M}^\top \mathbf M) \\
&\quad+\frac{p^1p^{-1}+p^2p^{-3/2}+p^3p^{-5/2}+p^4p^{-5/2}+p^5p^{-3}+p^6p^{-7/2}}{p^8}\text{tr}^2(\mathbf M^\top \mathbf M)\\
=&O\{p^{-4}\text{tr}^2(\bf M^\top \bf M)\}.
\end{aligned}
$$

(iii) Using similar techniques as in (ii), with higher-order moments and more combinatorial terms, we can show part (iii).
\end{proof}

\subsection{Proof of Theorem \ref{thm:Max-Max}}\label{sec:proofMM}
According to Lemma 1 in \cite{liu+feng+wang+2024}, we can approximate $\C_{\gamma}(k)$ as
\begin{align*}
\C_{\gamma}(k)=\left\{\frac{k}{n}\left(1-\frac{k}{n}\right)\right\}^{-\gamma}\frac{1}{\sqrt{n}}\zeta_1^{-1}\left(\S_{k}-\frac{k}{n}\S_{n}\right)+\boldsymbol J_{n;k}^\gamma,
\end{align*}
where $\S_k=\sum_{i=1}^k\U_i$. Then, $M_{np}$ and $M_{np}^\dagger$ can be decomposed as
\begin{equation}\label{eq:Jngamma}
    \begin{aligned}
        M_{np}&=\max_{\lambda_n\leq k\leq n-\lambda_n}\frac{1}{\sqrt{n}}\zeta_1^{-1}\left(\S_{k}-\frac{k}{n}\S_{n}\right) +J_{n}^0,\\
        M_{np}^\dagger&=\max_{\lambda_n\leq k\leq n-\lambda_n}\left\{\frac{k}{n}\left(1-\frac{k}{n}\right)\right\}^{-1/2}\frac{1}{\sqrt{n}}\zeta_1^{-1}\left(\S_{k}-\frac{k}{n}\S_{n}\right)+J_{n}^{1/2},
    \end{aligned}
\end{equation}
and the details of $J_n^{\gamma}$ are provided later. From the proof of Theorem 1 in \cite{wang2023}, we see that the conclusion holds when $\U_i$ follows a multivariate normal distribution, i.e. 
\begin{equation*}
    \pr(p^{1/2}\zeta_1\max_{\lambda_n\leq k\leq n-\lambda_n}\Vert \C^{\text{Nor}}_0(k)\Vert_\infty\leq u_p\{\exp(-x)\})\rightarrow \exp\{-\exp(-x)\}.
\end{equation*}
as $n\rightarrow\infty$ and $\lambda_n/n\rightarrow 0$. By Lemma \ref{LemmaA4}, we see that the vectors $\U_i$ are i.i.d. and each of their components follows a sub-exponential distribution. We follows the Steps 1--3 in the proof of Theorem C.4 in \cite{Jirak-2015}, and acquire 
\begin{align}
    \notag & \pr(p^{1/2}\zeta_1\max_{\lambda_n\leq k\leq n-\lambda_n}\Vert \C_0(k)\Vert_\infty\leq u_p\{\exp(-x)\}) \\
    \notag & ~~~~~~~~~ -\pr(p^{1/2}\zeta_1\max_{\lambda_n\leq k\leq n-\lambda_n}\Vert \C_0^{\text{Nor}}(k)\Vert_\infty\leq u_p\{\exp(-x)\})\rightarrow 0,
\end{align}
where $\C_0^{\text{Nor}}(k)=\left\{\frac{k}{n}\left(1-\frac{k}{n}\right)\right\}^{-\gamma}\frac{1}{\sqrt{n}}\zeta_1^{-1}\left(S^{\text{Nor}}_{k}-\frac{k}{n}S^{\text{Nor}}_{n}\right)$ and $S^{\text{Nor}}_k=\sum_{i=1}^k\Y_i$ with $\Y_i\sim N(0,\R/p)$.

We next to show that the remainders shown in Equation (\ref{eq:Jngamma}) is $ J_{n}^\gamma=o_p(1)$. By the Bahadur representation of $\hat{\boldsymbol\theta}_{1:k}$ and $\hat{\boldsymbol\theta}_{k+1:n}$, we have

\begin{equation*}
    \begin{aligned}
        J_{n}^0&=\max_{\lambda_n\leq k\leq n-\lambda_n}\max_{1\leq j\leq p}( E_1+ E_2+ E_3),\\
        J_{n}^{1/2}&=\max_{\lambda_n\leq k\leq n-\lambda_n}\max_{1\leq j\leq p}\left\{\frac{k}{n}\left(1-\frac{k}{n}\right)\right\}^{-1/2}(E_1+E_2+E_3).
    \end{aligned}
\end{equation*}
where 
\begin{equation*}
    \begin{aligned}
        E_1&=-n^{-1/2}(1-\frac{k}{n})\zeta_1^{-1}\sum_{i=1}^k \varsigma_{1,i,1:k}U_{i,j}+n^{-1/2}\frac{k}{n}\zeta_1^{-1}\sum_{i=k+1}^n \varsigma_{1,i,k+1:n}U_{i,j},\\
        E_2&=n^{-1/2}(1-\frac{k}{n})\sum_{i=1}^k \zeta_1^{-1} R_i^{-1} \{\U_i^\top (\hat{\boldsymbol\theta}_{1:k}-\boldsymbol\theta)\} U_{i,j}-n^{-1/2}\frac{k}{n}\sum_{i=k+1}^n \zeta_1^{-1} R_i^{-1} \{\U_i^\top (\hat{\boldsymbol\theta}_{k+1:n}-\boldsymbol\theta)\} U_{i,j},\\
        E_3&=-n^{-1/2}(1-\frac{k}{n})\zeta_1^{-1}\sum_{i=1}^k \left\{R_i^{-1}(1+\varsigma_{1,i,1:k}+\varsigma_{2,i,1:k})-1\right\}(\hat\theta_{1:k,j}-\theta)\\
        &\qquad\qquad+n^{-1/2}\frac{k}{n}\zeta_1^{-1}\sum_{i=k+1}^n \left\{R_i^{-1}(1+\varsigma_{1,i,k+1:n}+\varsigma_{2,i,k+1:n})-1\right\}(\hat\theta_{k+1:n,j}-\theta),
    \end{aligned}
\end{equation*}
where $\varsigma_{1,i,1:k}\lesssim R_i^{-2}\Vert\hat{\boldsymbol\theta}_{1:k}-\boldsymbol\theta\Vert^2\{1+O_p(R_i^{-1}\Vert\hat{\boldsymbol\theta}_{1:k}-\boldsymbol\theta\Vert)\}=O_p(k^{-1})$ and $\varsigma_{2,i,1:k}=R_i^{-1}\W_i^\top (\hat{\boldsymbol\theta}_k-\boldsymbol\theta)-2^{-1}R_i^{-2}\Vert\hat{\boldsymbol\theta}_k-\boldsymbol\theta\Vert^2$ which are proved in \cite{cheng2023statistical}.

When $\gamma=0$, the first term is
\begin{equation*}
\begin{aligned}
    \max_{\lambda_n\leq k\leq n-\lambda_n}\max_{1\leq j\leq p}E_1=&\max_{\lambda_n\leq k\leq n-\lambda_n}\max_{1\leq j\leq p}n^{-1/2}(1-\frac{k}{n})\sum_{i=1}^k\zeta_1^{-1}\varsigma_{1,i,1:k}U_{i,j}\\
    \lesssim&n^{-1/2}\max_{\lambda_n\leq k\leq n-\lambda_n}k\Vert\hat{\boldsymbol\theta}_{1:k}-\boldsymbol\theta\Vert^2\max_{1\leq k\leq n}\left\vert\frac{1}{k}\sum_{i=1}^k\zeta_1^{-1}R_i^{-2}\U_i\right\vert_\infty\\
     =& O_p\left\{n^{-1/2}\left(\log n\right)^2\log(np)\right\},\\
\end{aligned}
\end{equation*}
by Lemma \ref{LemmaA4}, Cauchy inequality and  $\varsigma_{1,i,1:k}=O_p(R_i^{-2}\Vert\hat{\boldsymbol\theta}_{1:k}-\bth\Vert^2)$. For $s\rightarrow\infty$, we have $s^{1/2}(\hat{\boldsymbol\theta}_{1:s}-\boldsymbol\theta)\stackrel{d}{\rightarrow}N(0,\Sigma_{\boldsymbol\theta})$ and $\Vert\hat{\boldsymbol\theta}_{1:s}-\boldsymbol\theta\Vert^2=O_p(\zeta_1^{-2}s^{-1})$. Taking the same procedure as in the proof of the Lemma \ref{lemma2_dj}, we have, $\max_{s\leq k\leq n}k\Vert\hat{\boldsymbol\theta}_{1:k}-\boldsymbol\theta\Vert^2=O_p\{\zeta_1^{-2}(\log n)^2\}$ as $s\rightarrow\infty$.

Similarly, we decompose the second term
\begin{equation*}
    \begin{aligned}
         &\max_{\lambda_n\leq k\leq n-\lambda_n}\max_{1\leq j\leq p}E_2 \\
         &\leq  \max_{\lambda_n\leq k\leq n-\lambda_n}n^{-1/2}(1-\frac{k}{n})\left\vert\sum_{i=1}^k\zeta_1^{-1} \{R_i^{-1} \U_i\U_i^\top-\E (R_i^{-1} \U_i\U_i^\top)  \}(\hat{\boldsymbol\theta}_{1:k}-\boldsymbol\theta)\right\vert_\infty\\
         &\qquad+\max_{\lambda_n\leq k\leq n-\lambda_n}n^{-1/2}\frac{k}{n}\left\vert\sum_{i=k+1}^n\zeta_1^{-1} \{R_i^{-1} \U_i\U_i^\top -\E (R_i^{-1} \U_i\U_i^\top )\}(\hat{\boldsymbol\theta}_{k+1:n}-\boldsymbol\theta)\right\vert_\infty\\
          &\qquad+  \max_{\lambda_n\leq k\leq n-\lambda_n}n^{-1/2}\frac{k(n-k)}{n}\zeta_1^{-1}\left\vert\E (R_i^{-1} \U_i\U_i^\top)\left\{(\hat{\boldsymbol\theta}_{1:k}-\boldsymbol\theta)-(\hat{\boldsymbol\theta}_{k+1:n}-\boldsymbol\theta)\right\}\right\vert_\infty\\
          :&= E_{21}+E_{22}+E_{23}.
    \end{aligned}
\end{equation*}

For $E_{21}$,
\begin{equation*}
    \begin{aligned}
        E_{21}\leq& n^{-1/2}\zeta_1^{-1}\max_{1\leq k\leq n}(1-\frac{k}{n})k^{-1/2}\left\vert\sum_{i=1}^k \left\{R_i^{-1} \U_i\U_i^\top-\E (R_i^{-1} \U_i\U_i^\top  )\right\}\right\vert_1 \cdot\max_{\lambda_n\leq k\leq n-\lambda_n} k^{1/2}\left\vert\hat{\boldsymbol\theta}_{1:k}-\boldsymbol\theta\right\vert_\infty\\
        =:&n^{-1/2}\zeta_1^{-1}E_{211}\cdot E_{212},\\
    \end{aligned}
\end{equation*}
where
\begin{equation*}
    \begin{aligned}
         E_{211}:=&\max_{1\leq k\leq n}(1-\frac{k}{n})k^{-1/2}\left\vert\sum_{i=1}^k \left\{R_i^{-1} \U_i\U_i^\top-\E (R_i^{-1} \U_i\U_i^\top  )\right\}\right\vert_1 \\
         =&\sum_{j=1}^p\max_{1\leq k\leq n}\max_{1\leq l\leq p}\left\vert\sum_{i=1}^{k}(1-\frac{k}{n})k^{-1/2}\left\{R_i^{-1}U_{ij}U_{il}-\E (R_i^{-1}U_{ij}U_{il})\right\}\right\vert,\\
         E_{212}:=&\max_{\lambda_n\leq k\leq n-\lambda_n} k^{1/2}\vert \hat{\boldsymbol\theta}_{1:k}-\boldsymbol\theta\vert_\infty.
    \end{aligned}
\end{equation*}

To bounding $E_{211}$, we define
\begin{equation*}
    \begin{aligned}
        \phi_j^2 &=\max_{1\leq k\leq n}\max_{1\leq l\leq p}(1-\frac{k}{n})^2\E\left\{R_i^{-1}U_{ij}U_{il}-\E (R_i^{-1}U_{ij}U_{il})\right\}^2\\
        &\leq \max_{1\leq l\leq p}\E (R_i^{-2}U_{ij}^2 U_{il}^2),\\
        M_j & = \max_{1\leq k\leq n,1\leq i\leq k}\max_{1\leq l\leq p}\left\vert (1-\frac{k}{n})k^{-1/2}\left\{R_i^{-1}U_{ij}U_{il}-\E (R_i^{-1}U_{ij}U_{il})\right\}\right\vert.
    \end{aligned}
\end{equation*}
For $\phi_j$,
\begin{equation*}
    \begin{aligned}
        \sum_{j=1}^p \phi_j\leq& \max_{1\leq l\leq p}\sum_{j=1}^p\frac{p^{-1}\E (R_i^{-2}U_{il}^2)+\E (R_i^{-2}U_{ij}^2U_{il}^2)}{2\sqrt{p^{-1}\E (R_i^{-2}U_{il}^2)}}\\
        =&\max_{1\leq l\leq p}p\sqrt{p^{-1}\E (R_i^{-2}U_{il}^2)}\leq\zeta_1\max_{1\leq l\leq p}\sigma_{ll}^{1/2},
    \end{aligned}
\end{equation*}
where the last inequality is indicated by taking the same procedure as in the proof of Lemma A3 in \cite{cheng2023statistical}, we have, $\E (R_i^{-2}U_{il}^2)\lesssim \zeta_1 p^{-3/2} \sigma_{ll}+\zeta_1p^{-5/3}+\zeta_1 p^{-3/2-\eta_0/2}$.

For $M_j$,

\begin{equation*}
    \begin{aligned}
        M_j&=\max_{1\leq k\leq n}\max_{1\leq l\leq p}\left\vert k^{-1/2}\left\{R_{k}^{-1}U_{kj}U_{kl}-\E (R_{k}^{-1}U_{kj}U_{kl})\right\}\right\vert\\
        &\leq \max_{1\leq k\leq n}\max_{1\leq l\leq p}\left\vert k^{-1/2}R_{k}^{-1}U_{kj}U_{kl}\right\vert+\max_{1\leq l\leq p}\left\vert \E (R_{1}^{-1}U_{1j}U_{1l})\right\vert\\
        &\leq \max_{1\leq k\leq n}k^{-1/2}R_k^{-1}\max_{1\leq k\leq n}\max_{1\leq l\leq p}\left\vert U_{kj}U_{kl}\right\vert+\max_{1\leq l\leq p}\left\vert \E (R_{1}^{-1}U_{1j}U_{1l})\right\vert\\
        &\lesssim \zeta_1\max_{1\leq k\leq n}\max_{1\leq l\leq p}\left\vert U_{kj}U_{kl}\right\vert+\zeta_1p^{-1}\max_{1\leq l\leq p}\vert\sigma_{jl}\vert,
    \end{aligned}
\end{equation*}
where the last inequality holds by the proof of Lemma A3 in \cite{cheng2023statistical},
\begin{equation}\label{eq:ERUU_1}
    \begin{aligned}
        &\left\vert\E \left\{k^{-1}\sum_{i=1}^k R_i^{-1}U_{ij}U_{il} \right\}-\E \left\{ k^{-1}p^{-1/2}\sum_{i=1}^k \nu_i^{-1}\R^{1/2}_j U(\W_i)U(\W_i)^\top \R^{1/2\top}_l \right\} \right\vert\\
      \quad & \lesssim \zeta_1 p^{-1-\eta_0/2}+\zeta_1 p^{-7/6},
    \end{aligned}
\end{equation}
and
\begin{equation}\label{eq:ERUU_2}
    \begin{aligned}
        \E \left\{ k^{-1}p^{-1/2}\sum_{i=1}^k \nu_i^{-1}\R^{1/2}_j U(\W_i)U(\W_i)^\top \R^{1/2\top}_l \right\}\lesssim \zeta_1 p^{-1}\vert \sigma_{jl}\vert +\zeta_1 p^{-3/2}.
    \end{aligned}
\end{equation}

By the properties of $\psi_{\alpha_0}$ norm, we have 
\begin{equation*}
    \begin{aligned}
        \left\Vert \max_{1\leq k\leq n}\max_{1\leq l\leq p}\left\vert\zeta_1^{-2}U_{kj}U_{kl} \right\vert\right\Vert_{\psi_{\alpha_0/2}}&\leq\left\Vert \max_{1\leq k\leq n}\max_{1\leq j,l\leq p}\left\vert\zeta_1^{-2}U_{kj}U_{kl} \right\vert\right\Vert_{\psi_{\alpha_0/2}}\\
        &\lesssim \left\Vert \max_{1\leq k\leq n}\max_{1\leq l\leq p}\left\vert\zeta_1^{-1}U_{kl} \right\vert^2\right\Vert_{\psi_{\alpha_0/2}}\\
        &= \left\Vert \max_{1\leq k\leq n}\max_{1\leq l\leq p}\left\vert\zeta_1^{-1}U_{kl} \right\vert\right\Vert_{\psi_{\alpha_0/2}}^2\lesssim\log^2(np).
    \end{aligned}
\end{equation*}

It follows that,
\begin{equation*}
    \begin{aligned}
        \Vert M_j\Vert_{\psi_{\alpha_0/2}}\lesssim \zeta_1\left\Vert  \max_{1\leq k\leq n}\max_{1\leq l\leq p}\left\vert U_{kj}U_{kl}\right\vert\right\Vert_{\psi_{\alpha_0/2}}+\zeta_1 p^{-1}\max_{1\leq l\leq p}\vert\sigma_{jl}\vert\lesssim \zeta_1 p^{-1}\log^2(np).
    \end{aligned}
\end{equation*}

By the Lemma \ref{LemmaE.1central}, we have
\begin{equation*}
    \begin{aligned}
        \E(E_{211})\leq &\sum_{j=1}^p \phi_j\sqrt{\log (np)}+ \sum_{j=1}^p\sqrt{\E (M_j^2)}\log(np)\\
         &\qquad \leq\zeta_1 \log^{1/2}(np)+\zeta_1\log^3(np)\lesssim\zeta_1\log^3(np).
    \end{aligned}
\end{equation*}

Similarly,  
\begin{equation*}
    \begin{aligned}
        \E\left\{ \max_{1\leq k\leq n} k^{1/2}\left\vert \zeta_1^{-1}\frac{1}{k}\sum_{i=1}^k \U_i \right\vert_\infty\right\}\lesssim \log^2(np).
    \end{aligned}
\end{equation*}

For $E_{212}$, similar with the proof of the Lemma 1 in \cite{cheng2023statistical}, as $s\rightarrow\infty$, we obtain
\begin{equation*}
    \begin{aligned}
        s^{1/2}\left\vert \hat{\boldsymbol\theta}_{1:s}-\boldsymbol\theta\right\vert_\infty\lesssim s^{1/2}\left\vert \zeta_1^{-1}s^{-1}\sum_{i=1}^s \U_i \right\vert_\infty +\zeta_1^{-1} \left\vert s^{-1}\sum_{i=1}^s R_i^{-1}\U_i\U_i^\top\right\vert_1 s^{1/2}\left\vert \hat{\boldsymbol\theta}_{1:s}-\boldsymbol\theta\right\vert_\infty,
    \end{aligned}
\end{equation*}
then, we have
\begin{equation*}
    \begin{aligned}
        &~~~~\max_{s\leq k\leq n}k^{1/2}\left\vert \hat{\boldsymbol\theta}_{1:k}-\boldsymbol\theta\right\vert_\infty\\
        &\lesssim \max_{s\leq k\leq n}k^{1/2}\left\vert \zeta_1^{-1}k^{-1}\sum_{i=1}^k \U_i \right\vert_\infty +\max_{s\leq k\leq n}\zeta_1^{-1} \left\vert k^{-1}\sum_{i=1}^k R_i^{-1}\U_i\U_i^\top\right\vert_1 k^{1/2}\left\vert \hat{\boldsymbol\theta}_{1:k}-\boldsymbol\theta\right\vert_\infty\\
        &\lesssim \max_{s\leq k\leq n}k^{1/2}\left\vert \zeta_1^{-1}k^{-1}\sum_{i=1}^k \U_i \right\vert_\infty +s^{-1/2}\max_{s\leq k\leq n}k^{1/2}\zeta_1^{-1} \left\vert k^{-1}\sum_{i=1}^k R_i^{-1}\U_i\U_i^\top\right\vert_1 \max_{s\leq k\leq n}k^{1/2}\left\vert \hat{\boldsymbol\theta}_{1:k}-\boldsymbol\theta\right\vert_\infty\\
        &\lesssim\log^2(np)+s^{-1/2}\log^3(np)\max_{s\leq k\leq n}k^{1/2}\left\vert \hat{\boldsymbol\theta}_{1:k}-\boldsymbol\theta\right\vert_\infty.
    \end{aligned}
\end{equation*}
Let $s\asymp n^{\lambda}$, we have
\begin{equation*}
    \begin{aligned}
        \max_{s\leq k\leq n}k^{1/2}\left\vert \hat{\boldsymbol\theta}_{1:k}-\boldsymbol\theta\right\vert_\infty\lesssim \log^2(np).
    \end{aligned}
\end{equation*}

Then,
\begin{equation}\label{eq:max_theta}
    \begin{aligned}
         E_{212}=\max_{\lambda_n\leq k\leq n-\lambda_n}k^{1/2}\left\vert \hat{\boldsymbol\theta}_{1:k}-\boldsymbol\theta\right\vert_\infty\lesssim\log^2(np).
    \end{aligned}
\end{equation}

Thus, we have
\begin{equation*}
    \begin{aligned}
        E_{21}\lesssim n^{-1/2}\log^{5}(np)=o_p(1).
    \end{aligned}
\end{equation*}

Similarly, $E_{22}\lesssim n^{-1/2}\log^{5}(np)=o_p(1)$. For $  E_{23}$,
\begin{equation*}
    \begin{aligned}
        E_{23}\leq& \max_{1\leq k\leq n}\zeta_1^{-1}\left\vert\E (R_i^{-1} \U_i\U_i^\top)\right\vert_1\left\{\left\vert k^{1/2}(\hat{\boldsymbol\theta}_{1:k}-\boldsymbol\theta)\right\vert_\infty+\left\vert (n-k)^{1/2}(\hat{\boldsymbol\theta}_{k+1:n}-\boldsymbol\theta)\right\vert_\infty\right\}\\
    \leq &\max_{1\leq k\leq n}\zeta_1^{-1}\left\vert\E (R_i^{-1} \U_i\U_i^\top)\right\vert_1\left\{\max_{1\leq k\leq n}\left\vert k^{1/2}(\hat{\boldsymbol\theta}_{1:k}-\boldsymbol\theta)\right\vert_\infty+\max_{1\leq k\leq n}\left\vert (n-k)^{1/2}(\hat{\boldsymbol\theta}_{k+1:n}-\boldsymbol\theta)\right\vert_\infty\right\}\\
    \lesssim & p^{-(1/6\wedge \eta_0/2)}\log^7(np)=o_p(1),
    \end{aligned}
\end{equation*}
where the last inequality holds by Equations \eqref{eq:ERUU_1}--\eqref{eq:max_theta}. Then we obtain $\max_{1\leq k\leq n}\max_{1\leq j\leq p}E_2=o_p(1)$. Taking the same procedure, we can also show that $\max_{1\leq k\leq n}\max_{1\leq j\leq p}E_3=o_p(1)$. The result is as follows. Similarly, we can proof the conclusion for $M_{n,p}^\dagger$. The proof is completed.

\subsection{Proof of Theorem \ref{thm:Max-Sum}}\label{sec:proofMS}
For max-$L_2$-type test with $\gamma=0$, 
\begin{equation*}
    \begin{aligned}
        \frac{1}{\sqrt{2\tr(\R^2)}}S_{n,p}&=\frac{1}{\sqrt{2\tr(\R^2)}}\max_{\lambda_n\leq k\leq n-\lambda_n}\left\{\frac{p}{n}(\hat{\S}_k-\frac{k}{n}\hat{\S}_n)^\top(\hat{\S}_k-\frac{k}{n}\hat{\S}_n)-\frac{k(n-k)p}{n^2}\right\}\\
        &:=\max_{\lambda_n\leq k\leq n-\lambda_n}\sum_{i\not= j}\upsilon_{i,k}\upsilon_{j,k} \hat{\U}_i^\top \hat{\U}_j\\
        &=\max_{\lambda_n\leq k\leq n-\lambda_n}\left\{\sum_{i\not= j}\upsilon_{i,k}\upsilon_{j,k} \U_i^\top \U_j+\sum_{i\not= j}\upsilon_{i,k}\upsilon_{j,k} (\hat{\U}_i^\top \hat{\U}_j- \U_i^\top \U_j)\right\},\\
    \end{aligned}
\end{equation*}
where 
\begin{equation*}
    \upsilon_{i,k}= \left\{\frac{p}{n\sqrt{2\tr(\R^2)}}\right\}^{1/2}\frac{n-k}{n},i\leq k;\ 
    \upsilon_{i,k}= -\left\{\frac{p}{n\sqrt{2\tr(\R^2)}}\right\}^{1/2}\frac{k}{n},i> k.\\
\end{equation*}

We first consider the $\hat{\U}_i$, by Taylor expansion, we have
\begin{equation}\label{eq:Ui_expansion}
    \begin{aligned}
    & U\big(\hat{\D}^{-1/2}(\X_i-\hat{\boldsymbol\theta}_{1:n})\big)\\
    =&U\left( \D^{-1/2}(\boldsymbol X_i-\boldsymbol\theta)-\D^{-1/2}(\hat{\boldsymbol\theta}_{1:n}-\boldsymbol\theta)+(\hat{\D}^{-1/2}-\D^{-1/2})(\boldsymbol X_i-\boldsymbol\theta)-(\hat{\D}^{-1/2}-\D^{-1/2})(\hat{\boldsymbol\theta}_{1:n}-\boldsymbol\theta) \right)\\
        =&\left\{ \U_i-R_i^{-1}\D^{-1/2}(\hat{\boldsymbol\theta}_{1:n}-\boldsymbol\theta)+R_i^{-1}(\hat{\D}^{-1/2}-\D^{-1/2})(\boldsymbol X_i-\boldsymbol\theta)\right.\\
        &\qquad\qquad\qquad\qquad\left.-R_i^{-1}(\hat{\D}^{-1/2}-\D^{-1/2})(\hat{\boldsymbol\theta}_{1:n}-\boldsymbol\theta) \right\}(1+\alpha_i)^{-1/2}.\\
    \end{aligned}
\end{equation}
Thus, for $\hat{\U}_i^\top \hat{\U}_j$, we have

\begin{equation*}
    \begin{aligned}
        &\hat{\U}_i^\top \hat{\U}_j\\
        =&\left\{ \U_i-R_i^{-1}\D^{-1/2}(\hat{\boldsymbol\theta}_{1:n}-\boldsymbol\theta)+R_i^{-1}(\hat{\D}^{-1/2}-\D^{-1/2})(\boldsymbol X_i-\boldsymbol\theta)-R_i^{-1}(\hat{\D}^{-1/2}-\D^{-1/2})(\hat{\boldsymbol\theta}_{1:n}-\boldsymbol\theta) \right\}\\
        &\cdot \left\{ \U_j-R_j^{-1}\D^{-1/2}(\hat{\boldsymbol\theta}_{1:n}-\boldsymbol\theta)+R_j^{-1}(\hat{\D}^{-1/2}-\D^{-1/2})(\boldsymbol X_j-\boldsymbol\theta)-R_j^{-1}(\hat{\D}^{-1/2}-\D^{-1/2})(\hat{\boldsymbol\theta}_{1:n}-\boldsymbol\theta) \right\}\\
        &\qquad\qquad\qquad\qquad\qquad\qquad\cdot(1+\alpha_i)^{-1/2}(1+\alpha_j)^{-1/2}\\
        =&\U_i^\top \U_j+R_i^{-1}R_j^{-1}\left\{\D^{-1/2}(\hat{\boldsymbol\theta}_{1:n}-\boldsymbol\theta)\right\}^\top \D^{-1/2}(\hat{\boldsymbol\theta}_{1:n}-\boldsymbol\theta)\\
        &~~~~~~~~~~~~~~~~~~~~~~~~~~~~~~-R_j^{-1}\U_i^\top \D^{-1/2}(\hat{\boldsymbol\theta}_{1:n}-\boldsymbol\theta)-R_i^{-1}\U_j^\top \D^{-1/2}(\hat{\boldsymbol\theta}_{1:n}-\boldsymbol\theta)\\
        &+\U_i^\top\U_j\left\{(1+\alpha_i)^{-1/2}(1+\alpha_j)^{-1/2}-1\right\}-R_j^{-1}\U_i^\top \D^{-1/2}(\hat{\boldsymbol\theta}_{1:n}-\boldsymbol\theta)\left\{(1+\alpha_i)^{-1/2}(1+\alpha_j)^{-1/2}-1\right\}\\
        &-R_i^{-1}\U_j^\top \D^{-1/2}(\hat{\boldsymbol\theta}_{1:n}-\boldsymbol\theta)\left\{(1+\alpha_i)^{-1/2}(1+\alpha_j)^{-1/2}-1\right\}\\
        &+R_i^{-1}R_j^{-1}\left\{\D^{-1/2}(\hat{\boldsymbol\theta}_{1:n}-\boldsymbol\theta)\right\}^\top \D^{-1/2}(\hat{\boldsymbol\theta}_{1:n}-\boldsymbol\theta)\left\{(1+\alpha_i)^{-1/2}(1+\alpha_j)^{-1/2}-1\right\}\\
        &+\left\{R_i^{-1}(\hat{\D}^{-1/2}-\D^{-1/2})(\boldsymbol X_i-\boldsymbol\theta)-R_i^{-1}(\hat{\D}^{-1/2}-\D^{-1/2})(\hat{\boldsymbol\theta}_{1:n}-\boldsymbol\theta)\right\}\\
        &\qquad\cdot\left\{R_j^{-1}(\hat{\D}^{-1/2}-\D^{-1/2})(\boldsymbol X_j-\boldsymbol\theta)-R_j^{-1}(\hat{\D}^{-1/2}-\D^{-1/2})(\hat{\boldsymbol\theta}_{1:n}-\boldsymbol\theta)\right\}(1+\alpha_i)^{-1/2}(1+\alpha_j)^{-1/2}\\
        &+\left\{\U_j-R_j^{-1}\D^{-1/2}(\hat{\boldsymbol\theta}_{1:n}-\boldsymbol\theta)\right\}\left\{R_i^{-1}(\hat{\D}^{-1/2}-\D^{-1/2})(\boldsymbol X_i-\boldsymbol\theta)-R_i^{-1}(\hat{\D}^{-1/2}-\D^{-1/2})(\hat{\boldsymbol\theta}_{1:n}-\boldsymbol\theta)\right\}\\
        &\qquad\qquad\qquad\qquad\qquad\cdot(1+\alpha_i)^{-1/2}(1+\alpha_j)^{-1/2}\\
        &+\left\{\U_i-R_i^{-1}\D^{-1/2}(\hat{\boldsymbol\theta}_{1:n}-\boldsymbol\theta)\right\}\left\{R_j^{-1}(\hat{\D}^{-1/2}-\D^{-1/2})(\boldsymbol X_j-\boldsymbol\theta)-R_j^{-1}(\hat{\D}^{-1/2}-\D^{-1/2})(\hat{\boldsymbol\theta}_{1:n}-\boldsymbol\theta)\right\}\\
        &\qquad\qquad\qquad\qquad\qquad\cdot(1+\alpha_i)^{-1/2}(1+\alpha_j)^{-1/2},\\
    \end{aligned}
\end{equation*}
where $\alpha_i=2\U_i^\top R_i^{-1}(\hat{\D}^{-1/2}-\D^{-1/2})(\boldsymbol X_i-\boldsymbol\theta)-2\U_i^\top R_i^{-1}(\hat{\D}^{-1/2}-\D^{-1/2})(\hat{\boldsymbol\theta}_{1:n}-\boldsymbol\theta)+R_i^{-2}\Vert R_i^{-1}(\hat{\D}^{-1/2}-\D^{-1/2})(\boldsymbol X_i-\boldsymbol\theta)-R_i^{-1}(\hat{\D}^{-1/2}-\D^{-1/2})(\hat{\boldsymbol\theta}_{1:n}-\boldsymbol\theta)\Vert^2+2R_i^{-2}\U_i^\top \D^{-1/2}(\hat{\boldsymbol\theta}_{1:n}-\boldsymbol\theta)+R_i^{-2}(\hat{\boldsymbol\theta}_{1:n}-\boldsymbol\theta)^\top \D^{-1}(\hat{\boldsymbol\theta}_{1:n}-\boldsymbol\theta)$, by the Assumption \ref{ass:max3}(iv) and the same procedure of Theorem 2 in \cite{feng2016}, we have $\alpha_i=O_p\{n^{-1/2}(\log p)^{1/2}\}$. It implies that,
\begin{equation*}
    \begin{aligned}
        \hat{\U}_i^\top \hat{\U}_j=\U_i^\top \U_j&+R_i^{-1}R_j^{-1}\left\{\D^{-1/2}(\hat{\boldsymbol\theta}_{1:n}-\boldsymbol\theta)\right\}^\top \D^{-1/2}(\hat{\boldsymbol\theta}_{1:n}-\boldsymbol\theta)\\
        &-R_j^{-1}\U_i^\top \D^{-1/2}(\hat{\boldsymbol\theta}_{1:n}-\boldsymbol\theta)-R_i^{-1}\U_j^\top \D^{-1/2}(\hat{\boldsymbol\theta}_{1:n}-\boldsymbol\theta)+Q_{n,i,j},\\
    \end{aligned}
\end{equation*}
where $Q_{n,i,j}=o_p\{n^{-1}p^{-1}\sqrt{2\tr(\R^2)}\}$. Then, we have
\begin{equation*}
\begin{aligned}
    &\sum_{1\leq i\not= j\leq n}\upsilon_{i,k}\upsilon_{j,k} \left(\hat\U_i^\top \hat\U_j-\U_i^\top \U_j\right)=\sum_{1\leq i, j\leq n}\upsilon_{i,k}\upsilon_{j,k} \left(\hat\U_i^\top \hat\U_j-\U_i^\top \U_j\right)\\
    =& \sum_{1\leq i\not= j\leq k}\upsilon_{i,k}\upsilon_{j,k} \hat\U_i^\top \hat\U_j+\sum_{k+1\leq i\not= j\leq n}\upsilon_{i,k}\upsilon_{j,k} \hat\U_i^\top \hat\U_j+2\sum_{1\leq i\leq k}\sum_{k+1\leq j\leq n}\upsilon_{i,k}\upsilon_{j,k} \hat\U_i^\top \hat\U_j\\
    :=& \sum_{1\leq i\not= j\leq n}\upsilon_{i,k}\upsilon_{j,k}\U_i^\top \U_j-\sum_{1\leq i\not= j\leq n}\upsilon_{i,k}\upsilon_{j,k}\U_i^\top \U_j\\
    &\qquad-\sum_{1\leq i,j\leq n}\upsilon_{i,k}\upsilon_{j,k}R_i^{-1}(\hat{\boldsymbol\theta}_{1:n}-\boldsymbol\theta)^\top\D^{-1/2}\U_j-\sum_{1\leq i,j\leq n}\upsilon_{i,k}\upsilon_{j,k}R_j^{-1}\U_i^\top\D^{-1/2}(\hat{\boldsymbol\theta}_{1:n}-\boldsymbol\theta)  \\
    &\qquad +\sum_{1\leq i\not= j\leq n} \upsilon_{i,k}\upsilon_{j,k}R_i^{-1}R_j^{-1} (\hat{\boldsymbol\theta}_{1:n}-\boldsymbol\theta)^\top \D^{-1/2}\D^{-1/2}(\hat{\boldsymbol\theta}_{1:n}-\boldsymbol\theta)+\sum_{1\leq i\not= j\leq n}Q_{n,k,i,j}\\
    \lesssim&
    -2\sum_{1\leq i,j\leq n}\upsilon_{i,k}\upsilon_{j,k}R_i^{-1}(\hat{\boldsymbol\theta}_{1:n}-\boldsymbol\theta)^\top\D^{-1/2}\U_j\\
    &\qquad +\sum_{1\leq i\not= j\leq n} \upsilon_{i,k}\upsilon_{j,k}R_i^{-1}R_j^{-1} (\hat{\boldsymbol\theta}_{1:n}-\boldsymbol\theta)^\top \D^{-1/2}\D^{-1/2}(\hat{\boldsymbol\theta}_{1:n}-\boldsymbol\theta).\\
\end{aligned}
\end{equation*}
For two parts, taking the same procedure as in the proof of Lemma A.2 in \cite{feng2016multivariate}, we have
\begin{equation*}
    \begin{aligned}
        &\max_{\lambda_n\leq k\leq n-\lambda_n}\sum_{1\leq i,j\leq n}\upsilon_{i,k}\upsilon_{j,k}R_i^{-1}(\hat{\boldsymbol\theta}_{1:n}-\boldsymbol\theta)^\top\D^{-1/2}\U_j\\
        =&\max_{\lambda_n\leq k\leq n-\lambda_n}\frac{k^2(n-k)^2p}{n^3\sqrt{2\tr(\R^2)}} \left(\frac{1}{k}\sum_{i=1}^k R_i^{-1}-\frac{1}{n-k}\sum_{i=k+1}^n R_i^{-1}\right)\\
        &\qquad\qquad\qquad\qquad\qquad\cdot\left\{\D^{-1/2}(\hat{\boldsymbol\theta}_{1:n}-\boldsymbol\theta)\right\}^\top \left(\frac{1}{k}\sum_{i=1}^k \U_i-\frac{1}{n-k}\sum_{i=k+1}^n\U_i\right)\\
        =&\max_{\lambda_n\leq k\leq n-\lambda_n}\frac{k^2(n-k)^2p}{n^3\sqrt{2\tr(\R^2)}} \left(\frac{1}{k}\sum_{i=1}^k R_i^{-1}-\frac{1}{n-k}\sum_{i=k+1}^n R_i^{-1}\right)\\
        &\qquad\qquad\qquad\qquad\qquad\cdot\{1+o_p(1)\}\left\{\D^{-1/2}(\hat{\boldsymbol\theta}_{1:n}-\boldsymbol\theta)\right\}^\top \left(\frac{1}{k}\sum_{i=1}^k \U_i-\frac{1}{n-k}\sum_{i=k+1}^n\U_i\right)\\
    \end{aligned}
\end{equation*}
and
\begin{equation*}
    \begin{aligned}
        &\left(\frac{1}{n}\sum_{i=1}^n \U_i\right)^\top \left(\frac{1}{k}\sum_{i=1}^k \U_i-\frac{1}{n-k}\sum_{i=k+1}^n\U_i\right)\\
        =&\frac{1}{kn}\sum_{1\leq i\not=j\leq k}\U_i^\top \U_j-\frac{1}{(n-k)n}\sum_{k+1\leq i\not= j\leq n}\U_i^\top \U_j+\frac{n-2k}{k(n-k)n}\sum_{1\leq i\leq k}\sum_{k+1\leq j\leq n}\U_i^\top \U_j.
    \end{aligned}
\end{equation*}

By the proof of Theorem 5 in \cite{liu+feng+wang+2024}, we have
\begin{equation*}
    \begin{aligned}
        \frac{p}{k\sqrt{2\tr(\R^2)}}\sum_{1\leq i\not= j\leq k}\U_i^\top \U_j\stackrel{d}{\rightarrow}N(0,1).
    \end{aligned}
\end{equation*}
Similarly, we have
\begin{equation*}
    \begin{aligned}
        \frac{p}{k^{1/2}(n-k)^{1/2}\sqrt{2\tr(\R^2)}}\sum_{1\leq i\leq k}\sum_{k+1\leq j\leq n}\U_i^\top \U_j\stackrel{d}{\rightarrow}N(0,1).
    \end{aligned}
\end{equation*}
Then,
\begin{equation*}
    \begin{aligned}
        &\max_{\lambda_n\leq k\leq n-\lambda_n}\left(\frac{1}{n}\sum_{i=1}^n \U_i\right)^\top \left(\frac{1}{k}\sum_{i=1}^k \U_i-\frac{1}{n-k}\sum_{i=k+1}^n\U_i\right)\\
        \lesssim &\frac{\sqrt{2\tr\R^2}}{p}\left(n^{-1}\sqrt{\log n} +n^{-1}\sqrt{\log n}+n^{-1/2}\sqrt{\log n}\right)\\
        =&\frac{\sqrt{2\tr\R^2}}{p}n^{-1}\sqrt{\log n}.
        \end{aligned}
\end{equation*}
Similarly, 
\begin{equation*}
    \begin{aligned}
    &\max_{\lambda_n\leq k\leq n-\lambda_n}\frac{1}{k}\sum_{i=1}^k R_i^{-1}-\frac{1}{n-k}\sum_{i=k+1}^n R_i^{-1}\\
    \leq &\max_{\lambda_n\leq k\leq n-\lambda_n}\left\vert\frac{1}{k}\sum_{i=1}^k R_i^{-1}-\zeta_1\right\vert+\max_{\lambda_n\leq k\leq n-\lambda_n}\left\vert\frac{1}{n-k}\sum_{i=k+1}^n R_i^{-1}-\zeta_1\right\vert,
    \end{aligned}
\end{equation*}
and 
\begin{equation*}
    \begin{aligned}
        &\max_{\lambda_n\leq k\leq n-\lambda_n}k^{1/2}\left\vert\frac{1}{k}\zeta_1^{-1}\sum_{i=1}^k R_i^{-1}-1 \right\vert=O_p\left(\sqrt{\log n}\right),
    \end{aligned}
\end{equation*}
by the proof of lemma 3 in \cite{liu+feng+wang+2024}. Thus,
\begin{small}
\begin{equation*}
    \begin{aligned}
    &\max_{\lambda_n\leq k\leq n-\lambda_n}\frac{k^2(n-k)^2p}{n^3\sqrt{2\tr(\R^2)}} \{1+o_p(1)\}\left(\frac{1}{k}\sum_{i=1}^k R_i^{-1}-\frac{1}{n-k}\sum_{i=k+1}^n R_i^{-1}\right)\\
    &\qquad\qquad\qquad\qquad\qquad\qquad\cdot \left(\zeta_1^{-1}\frac{1}{n}\sum_{i=1}^n\U_i\right)^\top \left(\frac{1}{k}\sum_{i=1}^k \U_i-\frac{1}{n-k}\sum_{i=k+1}^n\U_i\right)\\
    \leq&O_p\left\{\max_{\lambda_n\leq k\leq n-\lambda_n}\frac{k^2(n-k)^2}{n^3} \left( k^{-1/2}\sqrt{\log n}+(n-k)^{-1/2}\sqrt{\log n} \right)\left( n^{-1}+\frac{1}{k^{1/2}(n-k)^{1/2}}\right)\sqrt{\log n} \right\}\\
    =&O_p\left\{\max_{\lambda_n\leq k\leq n-\lambda_n} \log n\left( \frac{k^{3/2}(n-k)^2}{n^4}+\frac{k(n-k)^{3/2}}{n^3}+\frac{k^{2}(n-k)^{3/2}}{n^4}+\frac{k^{3/2}(n-k)}{n^3} \right)\right\}\\
    \leq &O_p(n^{-1/2}\log n). 
    \end{aligned}
\end{equation*}
\end{small}
 For the second part, by Assumption \ref{ass:sum_R4_0},
 \begin{small}
\begin{equation*}
    \begin{aligned}
        &\max_{\lambda_n\leq k\leq n-\lambda_n}\sum_{1\leq i, j\leq n} \upsilon_{i,k}\upsilon_{j,k}R_i^{-1}R_j^{-1} (\hat{\boldsymbol\theta}_{1:n}-\boldsymbol\theta)^\top \D^{-1/2}(\mathbf I_p-\U_i\U_i^\top)(\mathbf I_p-\U_j\U_j^\top)\D^{-1/2}(\hat{\boldsymbol\theta}_{1:n}-\boldsymbol\theta)\\
        =&\max_{\lambda_n\leq k\leq n-\lambda_n}\frac{k^2(n-k)^2p}{n^3\sqrt{2\tr(\R^2)}} \{1+o_p(1)\}\left(\zeta_1\frac{1}{n}\sum_{i=1}^n\U_i\right)^\top\left(\frac{1}{k}\sum_{i=1}^k R_i^{-1}-\frac{1}{n-k}\sum_{i=k+1}^n R_i^{-1}\right)^2\left(\zeta_1\frac{1}{n}\sum_{i=1}^n\U_i\right)\\
        \leq &O_p\left(\max_{\lambda_n\leq k\leq n-\lambda_n} \frac{k^2(n-k)^2p}{n^3\sqrt{2\tr(\R^2)}} n^{-1} \{k^{-1/2}\log k+(n-k)^{-1/2}\log(n-k)\}^{2}\right)\\
        =&O_p\left(\frac{p\log n}{n\sqrt{\tr(\R^2)}}\right)=o_p(1).
    \end{aligned}
\end{equation*}
\end{small}

Thus, we get,
\begin{equation*}
    \max_{\lambda_n\leq k\leq n-\lambda_n}\sum_{i\not= j}\upsilon_{i,k}\upsilon_{j,k} (\hat{\U}_i^\top \hat{\U}_j- \U_i^\top \U_j)=o_p(1),
\end{equation*}
i.e.
\begin{equation}\label{eq:Snp0_ram}
\begin{aligned}
  &\qquad\qquad\frac{1}{\sqrt{2\tr(\R^2)}}S_{np}\\
   =&\frac{1}{\sqrt{2\tr(\R^2)}}\max_{\lambda_n\leq k\leq n-\lambda_n}\frac{p}{n}\left(\sum_{i=1}^k \U_i-\frac{k}{n}\sum_{i=1}^n \U_i\right)^\top \left(\sum_{i=1}^k \U_i-\frac{k}{n}\sum_{i=1}^n \U_i\right)-\frac{k(n-k)p}{n^2\sqrt{2\tr(\R^2)}}+o_p(1).
\end{aligned}
\end{equation}

For $\gamma=0.5$, 
\begin{equation*}
    \begin{aligned}
        \frac{1}{\sqrt{2\tr(\R^2)}}S^\dagger_{n,p}=&\max_{\lambda_n\le k\le n-\lambda_n}\tilde \C_{0.5}(k)^\top \tilde \C_{0.5}(k)-p\\
        =&\max_{\lambda_n\le k\le n-\lambda_n}\left\{\frac{k}{n}(1-\frac{k}{n})\right\}^{-1}\sum_{i\not= j}\upsilon_{i,k}\upsilon_{j,k}\hat\U_i^\top \hat{\U}_j.
    \end{aligned}
\end{equation*}

Taking the same procedure of $\gamma=0$, we have
\begin{equation}\label{eq:Snp05_ram}
    \max_{\lambda_n\le k\le n-\lambda_n}\left\{\frac{k}{n}(1-\frac{k}{n})\right\}^{-1}\sum_{i\not= j}\upsilon_{i,k}\upsilon_{j,k}\left(\hat\U_i^\top \hat{\U}_j-\U_i^\top \U_j\right)=o_p(1).
\end{equation}

\subsubsection{The limit distribution for $S_{np}$ under $H_0$}
We next consider the limit distribution for $S_{np}$ under $H_0$. By Equation \eqref{eq:Snp0_ram}, the $S_{np}$ can be rewritten as
\begin{equation*}
    \frac{1}{\sqrt{2\tr(\R^2)}}S_{np}=\max_{\lambda_n\leq k\leq n-\lambda_n} W(k)+o_p(1),
\end{equation*}
where
\begin{equation*}
    \begin{aligned}
      W(k) =&\frac{k^2(n-k)^2p}{n^3\sqrt{2\text{tr}(\R^2)}}(\bar \U_{1k}-\bar \U_{(k+1)n})^\top (\bar \U_{1k}-\bar \U_{(k+1)n})-\frac{k(n-k)p}{n^2 \sqrt{2\text{tr}(\R^2)}},\\
    \end{aligned}
\end{equation*}
and 
\begin{equation*}
    \tilde{\upsilon}_{i,t}= \begin{cases}
    \frac{1}{\lfloor n t\rfloor}, &i \leq \lfloor n t\rfloor,\\
        \frac{-1}{n-\lfloor n t\rfloor}, &i> \lfloor n t\rfloor. 
    \end{cases}
\end{equation*}

Then $W(\lfloor n t\rfloor)$ can be rewritten as, 
\begin{equation}\label{eq:Wnt}
    W(\lfloor n t\rfloor)=2g(\lfloor n t\rfloor)\sum_{i<j}\tilde{\upsilon}_{i,t}\tilde{\upsilon}_{j,t} \U_i^\top \U_j,
\end{equation}
where $g(k)=\frac{k^2(n-k)^2p}{n^3\sqrt{2\text{tr}(\R^2)}}$.
Taking the same procedure as Theorem 5 in \cite{liu+feng+wang+2024}, we only need to calculate the term,
$\sum_{j=2}^n\sum_{i=1}^{j-1}\tilde{\upsilon}_{i,t}^2\tilde{\upsilon}_{j,t}^2$ and replace the term $\sum_{j=2}^n\sum_{i=1}^{j-1}\{n^2(n-1)^2\}^{-1}$. By some calculations, we have
\begin{equation*}
    \begin{aligned}
    &\sum_{j=2}^n\sum_{i=1}^{j-1}\frac{1}{n^2(n-1)^2}=\frac{1}{2}\frac{1}{n(n-1)},\sum_{j=2}^n\sum_{i=1}^{j-1}\tilde{\upsilon}_{i,t}^2\tilde{\upsilon}_{j,t}^2=\frac{1}{2}\frac{1}{n^2t^2(1-t)^2}.
    \end{aligned}
\end{equation*}

Thus, we get
\begin{equation*}
    \frac{t(1-t)\cdot 2\sum_{i<j}\tilde{\upsilon}_{i,t}\tilde{\upsilon}_{j,t} \U_i^\top \U_j}{\sqrt{\frac{2\text{tr}(\R^2)}{n^2p^2}}}\rightarrow N(0,1),
\end{equation*}
i.e.
\begin{equation*}
     W( n t)=2g( n t)\sum_{i<j}\tilde{\upsilon}_{i,t}\tilde{\upsilon}_{j,t} \U_i^\top \U_j,\rightarrow N(0,t^2(1-t)^2).
\end{equation*}
By the fact that $\lfloor n t\rfloor/n\rightarrow t$, we finally acquire, $W(\lfloor n t\rfloor)\rightarrow N(0,t^2(1-t)^2)$.

We next consider the two time points $t$ and $s$ with $s<t$. and consider the limit distribution of $a W(n t)+b W(n s)$,
\begin{equation*}
    \begin{aligned}
        &a W(n t)+b W(n s)
        =\frac{2n p}{\sqrt{2\text{tr}(\R^2)}}\sum_{i<j}\{at^2(1-t)^2\tilde{\upsilon}_{i,t}\tilde{\upsilon}_{j,t}+b s^2(1-s)^2\tilde{\upsilon}_{i,s}\tilde{\upsilon}_{j,s}\}\U_i^\top \U_j.
    \end{aligned}
\end{equation*}
Similarly, 
\begin{equation*}
    \begin{aligned}
        &\sum_{j=2}^n\sum_{i=1}^{j-1}\left\{at^2(1-t)^2\tilde{\upsilon}_{i,t}\tilde{\upsilon}_{j,t}+b s^2(1-s)^2\tilde{\upsilon}_{i,s}\tilde{\upsilon}_{j,s}\right\}^2\\
        =&\frac{1}{2n^2}\left\{\frac{a^2t^4(1-t)^4}{t^2(1-t)^2}+\frac{b^2 s^4(1-s)^4}{s^2(1-s)^s}+2\frac{a b t^2(1-t)^2s^2(1-s)^2}{t^2(1-s)^2}\right\}+o(1)\\
        =&\frac{1}{2n^2}\{a^2t^2(1-t^2+b^2s^2(1-s)^2+2abs^2(1-t)^2\}+o(1),
    \end{aligned}
\end{equation*}
which means that,
\begin{equation*}
    (W(\lfloor n t\rfloor),W(\lfloor n s\rfloor))^\top \stackrel{d}{\rightarrow}N_2(0,\mathbf\Omega_2),
\end{equation*}
and
\begin{equation*}
\mathbf\Omega_2=\begin{pmatrix}
t^2(1-t)^2 & s^2(1-t)^2 \\
s^2(1-t)^2 & s^2(1-s)^2
\end{pmatrix}.
\end{equation*}

A set of three or more times points can be treated in the same way and therefore the finite-dimensional distributions of $W(\lfloor n t\rfloor)$ converge properly. We next prove the tightness of $W(\lfloor n t\rfloor)$. Since $W(\lfloor n 0\rfloor)=0$, the Equation (15.17) in \cite{billingsley+1968} is satisfied.
\begin{equation*}
    \begin{aligned}
        W(n t)=&\frac{t^2(1-t)^2 n p}{\sqrt{2\text{tr}(\R^2)}}\sum_{i<j}\tilde{\upsilon}_{i t}\tilde{\upsilon}_{j t} \U_i^\top \U_j\\
        =& \frac{p}{n\sqrt{2\text{tr}(\R^2)}}\left\{(1-t)^2\sum_{i<j<n t}\U_i^\top \U_j+t^2\sum_{n t<i<j}\U_i^\top \U_j-t(1-t)\sum_{i<n t<j}\U_i^\top \U_j\right\}\\
        =&\frac{p}{n\sqrt{2\text{tr}(\R^2)}}\left\{-t(1-t)\sum_{i<j}\U_i^\top \U_j+(1-t)\sum_{i<j<n t}\U_i^\top \U_j+t\sum_{n t<i<j}\U_i^\top \U_j\right\}\\
        :=& -t(1-t)K_1+(1-t)K_{2,t}+t K_{3,t}.
    \end{aligned}
\end{equation*}

By Theorem 5 in \cite{liu+feng+wang+2024}, $\E\vert K_1\vert^2 \leq C$ for a constant $C$, we see that $t(1-t)K_1$ is tight. We next prove tightness of $(1-t)K_{2,t}$. It's equivalent to show that, for each positive $\varsigma$ and $\vartheta$, there exists a $\varphi$, $0<\varphi<1$, and an integer $n_0$, such that
\begin{equation}\label{eq:tight}
    \pr(\sup_{\vert s-t\vert<\varphi}\left\{(1-t)\sum_{i<j<n t}\U_i^\top \U_j-(1-s)\sum_{i<j<n s}\U_i^\top \U_j\right\}/\{n p^{-1}\sqrt{2\text{tr}(\R^2)}\}\geq \varsigma)\leq \vartheta,~n\geq n_0.
\end{equation}

We rewrite the term in Equation \eqref{eq:tight} as,
\begin{equation*}
    \begin{aligned}
&\left\vert(1-t)\sum_{i<j<n t}\U_i^\top \U_j-(1-s)\sum_{i<j<n s}\U_i^\top \U_j\right\vert\\
=&\left\vert-(t-s)\sum_{i<j<n s}\U_i^\top \U_j+(1-t)\left\{\sum_{1<i<j<n t}\U_i^\top \U_j-\sum_{1<i<j<n s}\U_i^\top \U_j\right\}\right\vert\\
\leq &(t-s)\left\vert\sum_{i<j<n s}\U_i^\top \U_j\right\vert+\left\vert\sum_{1<i<j<n t}\U_i^\top \U_j-\sum_{1<i<j<n s}\U_i^\top \U_j\right\vert.
    \end{aligned}
\end{equation*}

\underline{Proof of Equation (\ref{eq:tight}):}
\begin{equation*}
    \begin{aligned}
        &\pr\left(\sup_{\vert s-t\vert<\varphi}\left\{(1-t)\sum_{i<j<n t}\U_i^\top \U_j-(1-s)\sum_{i<j<n s}\U_i^\top \U_j\right\}/\{n p^{-1}\sqrt{2\text{tr}(\R^2)}\}\geq \varsigma\right)\\
        \leq &\pr\left(\sup_{\vert s-t\vert<\varphi}(t-s)\vert\sum_{i<j<n s}\U_i^\top \U_j\vert/\{n p^{-1}\sqrt{2\text{tr}(\R^2)}\}\geq \varsigma/2\right)\\
        &+\pr\left(\sup_{\vert s-t\vert<\varphi}\vert\sum_{1<i<j<n t}\U_i^\top \U_j-\sum_{1<i<j<n s}\U_i^\top \U_j\vert/\{n p^{-1}\sqrt{2\text{tr}(\R^2)}\}\geq \varsigma/2\right)\\
        :=&K_{2,t,1}+K_{2,t,2}.
    \end{aligned}
\end{equation*}
By Doob's martingale inequality and some discussions for $\sum_{i<j< n}\U_i^\top \U_j$ before, we have  
\begin{equation*}
    \begin{aligned}
        K_{2,t,1}\leq &\pr\left(\vartheta\sup_{1<k\leq n}\vert\sum_{i<j<k}\U_i^\top \U_j\vert/\{n p^{-1}\sqrt{2\text{tr}(\R^2)}\}\geq \varsigma/2\right)\\
        \leq & \frac{4\varphi^2}{\varsigma^2}\E\left(\big[\vert\sum_{i<j< n}\U_i^\top \U_j\vert/\{n p^{-1}\sqrt{2\text{tr}(\R^2)}\}\big]^2\right)\leq \frac{C\varphi^2}{\varsigma^2}=\vartheta,
    \end{aligned}
\end{equation*}
where $\varphi=\varsigma\vartheta^{1/2}/C$, $C$ is a constant and do not depends on $\varsigma$ and $\vartheta$.

For $K_{2,t,2}$, by the Theorem 8.4 of \cite{billingsley+1968}, it reduces to check the following condition: for any $\varsigma>0$, there exists a $\vartheta>1$ and an integer $n_0$ such that for all $k$
\begin{equation*}
    \pr\left(\max_{m\leq n}\vert\sum_{1<i<j<k+m}\U_i^\top \U_j-\sum_{1<i<j<k}\U_i^\top \U_j\vert/\{n p^{-1}\sqrt{2\text{tr}(\R^2)}\}\geq \vartheta\right)\leq \frac{\varsigma}{\vartheta^2},n\geq n_0.
\end{equation*}
Since $\U_i$'s are $i.i.d$, it can further reduces to
\begin{equation*}
    \begin{aligned}
        \pr\left(\max_{m\leq n}\vert\sum_{1<i<j<m}\U_i^\top \U_j\vert/\{n p^{-1}\sqrt{2\text{tr}(\R^2)}\}\geq \vartheta\right)\leq \frac{\varsigma}{\vartheta^2},n\geq n_0.
    \end{aligned}
\end{equation*}
By Doob's martingale inequality, the result is as follows. From the Theorem 15.5 of \cite{billingsley+1968}, we see the limiting process $V(t)$ is continuous. The proof for $S_{np}$ is completed.

\subsubsection{The limit distribution for $S_{np}^\dagger$ under $H_0$}

We next consider the limit distribution for $S_{np}^\dagger$ under $H_0$. By Equation \eqref{eq:Snp05_ram}, we see that
 \begin{equation*}
     \begin{aligned}
         S_{np}^\dagger =& \max_{\lambda_n\leq k\leq n-\lambda_n} \frac{np}{k(n-k)}(\S_k-\frac{k}{n}\S_n)^\top(\S_k-\frac{k}{n}\S_n)-p+o_p(1)\\
         :=& \max_{\lambda_n\leq k\leq n-\lambda_n} H_{np}(k)+o_p(1),
     \end{aligned}
 \end{equation*}
 where
 \begin{equation*}
     \begin{aligned}
         H_{np}(k) =& \frac{np}{k(n-k)}(\S_k-\frac{k}{n}\S_n)^\top(\S_k-\frac{k}{n}\S_n)\\
         =& \frac{n}{k(n-k)}\sum_{i=1}^p \left\{ p(S_{ik}-\frac{k}{n}S_{in})^2-\frac{k(n-k)}{n}\right\}.
     \end{aligned}
 \end{equation*}
 
 We first give the approximation for the process, 
 \begin{equation}\label{eq:Z_pk}
 \begin{aligned}
      Z_p(k)=&\frac{p}{\sqrt{2\text{tr}(\R^2)}}(\S_k^\top \S_k-k)
      = \frac{2p}{\sqrt{2\text{tr}(\R^2)}}\sum_{i<j}\U_i^\top \U_j.\\
        =& \frac{2p}{\sqrt{2\text{tr}(\R^2)}}\sum_{i<j}U(\W_i)^\top \R U(\W_j)+O_p(k^{1/2}),\\
        :=&\tilde{Z}_p(k)+O_p(k^{1/2})
 \end{aligned}
 \end{equation}
where the last equation holds by taking the same procedure as in the proof of theorem 2 in \cite{feng2016}.
 The main step is to show that:
 \begin{lemma}\label{lemma:approxi_wiener}
     Suppose Assumptions \ref{ass:max1}-\ref{ass:max3}, \ref{ass:sum_R4} hold, then for each $n$ and $p$ we can define Wiener process $\{W_{n,p}(k),1\leq k\leq n\}$ such that,
     \begin{equation*}
         \max_{1\leq k\leq n}\vert Z_p(k) - W_{np}(2k^2) \vert/k^{3/4+\omega_1}=O_p(1).
     \end{equation*}
 \end{lemma}
 \begin{proof}The proof of Lemma \ref{lemma:approxi_wiener} is based on the Skorokhod representation of martingales\citep{2014-hall+Heyde-martingale}. By Equation \eqref{eq:Z_pk}, it suffices to show $ \max_{1\leq k\leq n}\vert \tilde{Z}_p(k) - W_{np}(2k^2) \vert/k^{3/4+\omega_1}=O_p(1)$. Rewrite $\tilde{Z}_p(k)$ as
     \begin{equation*}
         \begin{aligned}
             \tilde{Z}_p(k)=\sum_{j=1}^k v_j,~v_j={\frac{2p}{\sqrt{2\text{tr}(\R^2)}}\{\R^{1/2}U(\W_j)\}^\top \S^U_{j-1}},
         \end{aligned}
     \end{equation*}
where $\S^U_j=\R^{1/2}\sum_{i=1}^jU(\W_i)$. Let $\mathcal{F}_k=\sigma(\U_j,1\leq j\leq k)$, By Assumption \ref{ass:max1}, $\E(Z_p(k)\mid \mathcal{F}_{k-1})=Z_p(k-1)$, so $\{Z_p(k),\mathcal{F}_k\}$ is a martingale. By the Skorokhod representation theorem for martingales\citep{2014-hall+Heyde-martingale}, we can define a Wiener process $W$ and random variables $\tau_1,\tau_2,\ldots$ satisfying Equations \eqref{eqs:prop_Skorokhod}(i)--(iv) below. Let $w_j=W(T_j)-W(T_{j-1})$ with $T_j= \sum_{l=1}^j\tau_l$, $T_0=0$ and $\mathcal{G}_k=\sigma(w_j,1\leq j\leq k)$. The Wiener process defined by the Skorokhod construction has the following properties:
\begin{equation}\label{eqs:prop_Skorokhod}
    \begin{aligned}
        &(i):\ \left\{\sum_{j=1}^k v_j,1\leq k\leq n\right\}\stackrel{d}{=} \left\{\sum_{j=1}^k w_j,1\leq k\leq n\right\},\\
        &(ii):\ T_k\in \mathcal{G}_k,\\
        &(iii):\ \E(\tau_k\mid \mathcal{G}_{k-1})=\E(w_k^2\mid \mathcal{G}_{k-1}),a.s.,\\
        &(iv):\ \E(\tau_k^r\mid \mathcal{G}_{k-1})\leq C_r^* \E(\vert w_k\vert^{2r}\mid \mathcal{G}_{k-1})\text{ for any } r\geq 1, \text{ where }C^*_r \text{ only depends on }r.
    \end{aligned}
\end{equation}
Due to the modulus of continuity of $W$\citep{2014-csorgo+Revese-strong}, it's enough to show the $T_k$ is approximately $2k^2$. We start with the decomposition
\begin{equation*}
    \begin{aligned}
        T_k=\sum_{l=1}^k \left\{\tau_l-\E (\tau_l\mid \mathcal{G}_{l-1})\right\}+\E (\tau_l\mid \mathcal{G}_{l-1}).
    \end{aligned}
\end{equation*}
It follows from Equation  \eqref{eqs:prop_Skorokhod}(ii) that $\{\sum_{l=1}^k \left\{\tau_l-\E (\tau_l\mid \mathcal{G}_{l-1})\right\},1\leq k\leq n\}$ is a martingale. On account of Equation \eqref{eqs:prop_Skorokhod}(i) and (iv), we have $\E \{\tau_l-\E (\tau_l\mid \mathcal{G}_{l-1})\}^2\leq C_1\E (v_l^4)$.
\begin{small}
\begin{equation*}
    \begin{aligned}
        &~~~~\E (v_j^4)\\
        & = \frac{4p^4}{\tr^2(\R^2)}\left\{\sum_{l=1}^{j-1}\E[\{U(\W_j)^\top \R U(\W_l)\}^4]+6\sum_{l_1\not= l_2}\E[\{U(\W_j)^\top\R U(\W_{l_1})\}^2 \{U(\W_j)^\top \R U(\W_{l_2})\}^2]\right\}\\
        &=\frac{4p^4}{\tr^2(\R^2)}\left\{(j-1)\E[\{U(\W_1)^\top \R U(\W_2)\}^4]+3j(j-1)\E[\{U(\W_1)^\top \R U(\W_2)\}^2 \{U(\W_1)^\top\R  U(\W_3)\}^2]\right\}\\
        &\leq   \frac{28(j-1)^2p^4}{\tr^2(\R^2)}\E[\{U(\W_1)^\top\R  U(\W_2)\}^4].
    \end{aligned}
\end{equation*}
\end{small}

By lemma \ref{lemma1_like_2016}, we have, $\E[ \{U(\W_1)^\top\R U(\W_2)\}^4]=O\{p^{-4}\tr^2(\R^2)\}$. So $\E (v_j^4)\leq C_2 (j-1)^2$. The H{\'a}jek-R{\'e}nyi inequality for martingales\citep{chow2012probability-martingale} yields for all $x>0$,

\begin{equation*}
    \begin{aligned}
        \pr\left\{ \max_{1\leq k\leq n}k^{-(3/2+\omega_1)}\left\vert \sum_{l=1}^k\{\tau_l-\E(\tau_l\mid \mathcal{G}_{l-1})\} \right\vert \geq x\right\}\leq \frac{C_1}{x^2}\sum_{l=1}^n\frac{\E (v_l^4)}{l^{3+2\omega_1}}\leq \frac{C}{x^2},
    \end{aligned}
\end{equation*}
which means that, 
\begin{equation}\label{eq:lemmaA_42}
    \max_{1\leq k\leq n}k^{-(3/2+\omega_1)}\left\vert \sum_{l=1}^k(\tau_l-\E\{\tau_l\mid \mathcal{G}_{l-1})\} \right\vert=O_p(1).
\end{equation}

From Equation \eqref{eqs:prop_Skorokhod} (i) and (iii), we have $\{\E(\tau_l\mid \mathcal{G}_{l-1},1\leq l\leq n)\}\stackrel{d}{=}\{\E(v_l^2\mid \mathcal{F}_{l-1},1\leq l\leq n)\}$. Hence, 
\begin{equation*}
    \begin{aligned}
        &\sum_{l=1}^k\E(\tau_l\mid \mathcal{G}_{l-1})=\sum_{l=1}^k\E(v_l^2\mid \mathcal{F}_{l-1})\\
       =&\sum_{l=1}^k \frac{4p^2}{2\tr(\R^2)}\left\{\sum_{i=1}^{l-1}U(\W_i)^\top \R\Sigma_u\R U(\W_i)+\sum_{i\not= j}U(\W_i)^\top \R\Sigma_u\R U(\W_j)\right\}\\
       =& \frac{4p^2}{2\tr(\R^2)}\left\{\sum_{i=1}^{k-1}(k-i)U(\W_i)^\top \R\Sigma_u\R U(\W_i)+2\sum_{i=1}^{k-2}\sum_{j=i+1}^{k-1}(k-j)U(\W_i)^\top \R\Sigma_u\R U(\W_j)\right\}\\
       :=& 2k v_{1,k}-2v_{2,k}+4kv_{3,k}-4v_{4,k},
    \end{aligned}
\end{equation*}
where $\mathbf\Sigma_u=\E \{U(\W_i)U(\W_i)^\top\}$ and
\begin{equation*}
    \begin{aligned}
        v_{1,k}&=\frac{p^2}{\tr(\R^2)}\sum_{i=1}^{k-1}U(\W_i)^\top \R\mathbf\Sigma_u\R U(\W_i),\ v_{2,k}=\frac{p^2}{\tr(\R^2)}\sum_{i=1}^{k-1}i U(\W_i)^\top \R\mathbf\Sigma_u\R U(\W_i)\\
        v_{3,k}&=\frac{p^2}{\tr(\R^2)}\sum_{i=1}^{k-2}\sum_{j=i+1}^{k-1} U(\W_i)^\top \R\mathbf\Sigma_u\R U(\W_j),\  v_{4,k}=\frac{p^2}{\tr(\R^2)}\sum_{i=1}^{k-2}\sum_{j=i+1}^{k-1}j U(\W_i)^\top \R\mathbf\Sigma_u\R U(\W_j).\\
    \end{aligned}
\end{equation*}
 We note that $\E \{U(\W_i)^\top \R\mathbf\Sigma_u\R U(\W_i)\}= \tr\{(\R\mathbf\Sigma_u)^2\}$ and by Lemma \ref{lemma:UWi1}

\begin{equation*}
    \begin{aligned}
        \tr\{(\R\mathbf\Sigma_u)^2\}=&\sum_{1\leq i,j\leq p}p^{-2}\sigma_{ij}^2+\sum_{1\leq i,j,\leq p}(\sum_{l=1}^p \sigma_{il})^2 O(p^{-5})\\
        &+2\sum_{1\leq i,j\leq p} p^{-1} \sigma_{ij}\sum_{l=1}^p O(p^{-5/2})\\
        =&p^{-2}\tr(\R^2)+O(p^{-1})p^{-2}\tr(\R^2)+2O(p^{-1/2})p^{-2}\tr(\R^2)\\
        =&p^{-2}\tr(\R^2)\{1+O(p^{-1/2})\}=p^{-2}\tr(\R^2)\{1+o(1)\},
    \end{aligned}
\end{equation*}
where 
\begin{equation*}
    \begin{aligned}
        \sum_{1\leq i,j,\leq p}(\sum_{l=1}^p \sigma_{il})^2 O(p^{-5})=&O(p^{-4})\sum_{i=1}^p(\sum_{l=1}^p \sigma_{il})^2\leq O(p^{-3})\sum_{i=1}^p\sum_{l=1}^p \sigma_{il}^2,
    \end{aligned}
\end{equation*}
and
\begin{equation*}
    \sum_{1\leq i,j\leq p} p^{-1} \sigma_{ij}\sum_{l=1}^p O(p^{-5/2})\leq \sqrt{\sum_{1\leq i,j\leq p}p^{-2}\sigma_{ij}^2}\sqrt{\sum_{1\leq i,j\leq p}(\sum_{l=1}^p \sigma_{il})^2 O(p^{-5})}=O(p^{-5/2})\sum_{i=1}^p\sum_{l=1}^p \sigma_{il}^2.
\end{equation*}

Similarly, $\tr\{(\R\mathbf\Sigma_u)^4\}=p^{-4}\tr(\R^4)\{1+o(1)\}$. For $v_{1,k}$, $\E (v_{1,k})=\{p^{-2}\tr(\R^2)\}^{-1}\tr\{(\R\mathbf\Sigma_u)^2\}=1+O(p^{-1/2})$ and 
\begin{equation*}
\begin{aligned}
     &\pr\left\{ \max_{1\leq k\leq n}k^{-(1/2+\omega_1)}\vert v_{1,k}-(k-1)\vert>x\right\}\\
     &\leq  \pr\left\{ \max_{1\leq k\leq n}k^{-(1/2+\omega_1)}\vert v_{1,k}-(k-1)\E (v_{1,k})\vert>x\right\}+\pr\left\{ n^{1/2-\omega_1} O(p^{-1/2})>x\right\}.\\
\end{aligned}
\end{equation*}
 Since ${v}_{1,k}$ is sum of independent random variables, by  H{\'a}jek-- R{\'e}nyi inequality, 
\begin{equation}\label{eq:v1k}
    \begin{aligned}
        &\pr\left\{ \max_{1\leq k\leq n}k^{-(1/2+\omega_1)}\vert v_{1,k}-(k-1)\E (v_{1,k})\vert>x\right\}\\
        &\leq \frac{1}{x^2p^{-4}\tr^2(\R^2)}\sum_{k=1}^n\frac{\var \{U(\W_k)^\top \R\mathbf\Sigma_u\R U(\W_k)\}}{k^{1+2\omega_1}}\lesssim\frac{1}{x^2}\sum_{k=1}^n\frac{1}{k^{1+2\omega_1}}.
    \end{aligned}
\end{equation}
So for $v_{1,k}$, we have $\pr\left\{ \max_{1\leq k\leq n} k^{-(1/2+\omega_1)}\vert v_{1,k}-(k-1)\vert>x\right\}\lesssim x^{-2}\sum_{k=1}^n k^{-1-2\omega_1}$as $x>O(n^{1/2-\omega_1}p^{-1/2})$.

Similarly, for $v_{k,2}$, we have
\begin{equation}\label{eq:v2k}
    \begin{aligned}
        \pr(\max_{1\leq k\leq n} k^{-(3/2+\omega_1)}\vert v_{2,k}-k(k-1)/2 \vert>x)\lesssim \frac{1}{x^2p^{-4}\tr^2(\R^2)}\sum_{k=1}^n\frac{k^2 p^{-4}\tr(\R^4)}{k^{3+2\omega_1}}\lesssim\frac{1}{x^2}\sum_{k=1}^n\frac{1}{k^{1+2\omega_1}},
    \end{aligned}
\end{equation}
as $x>O(n^{1/2-\omega_1}p^{-1/2})$. 
For $v_{k,3}$ and $v_{k,4}$, by  H{\'a}jek--R{\'e}nyi inequality for martingales and Assumption \ref{ass:sum_R4}, 
\begin{equation}\label{eq:v34k}
    \begin{aligned}
        &\pr(\max_{1\leq k\leq n} k^{-(1/2+\omega_1)}\vert v_{3,k} \vert>x)\lesssim \frac{1}{x^2\tr^2(\R^2)}\sum_{k=1}^n\frac{(k-1)\tr(\R^4)}{k^{1+2\omega_1}}\lesssim \frac{1}{x^2} \frac{\tr(\R^4)n^{1-2\omega_1}}{\tr^2(\R^2)}=O(\frac{1}{x^2}),\\
        &\pr(\max_{1\leq k\leq n} k^{-(3/2+\omega_1)}\vert v_{4,k} \vert>x)\lesssim \frac{1}{x^2\tr^2(\R^2)}\sum_{k=1}^n\frac{k^2(k-1)\tr(\R^4)}{k^{3+2\omega_1}}=O(\frac{1}{x^2}).
    \end{aligned}
\end{equation}

Combing Equation \eqref{eq:v1k} and \eqref{eq:v2k} and Assumption \ref{ass:sum_R4}, we have
\begin{equation}\label{eq:lemmaA_43}
    \begin{aligned}
        \max_{1\leq k\leq n} k^{-(3/2+\omega_1)}\left\vert \sum_{l=1}^k\{\E(\tau_l\mid \mathcal{G}_{l-1})-2(l-1)\} \right\vert=O_p(1).
    \end{aligned}
\end{equation}

Due to the modulus of cotinuity of $W$, we get
\begin{equation}\label{eq:Tk_2k2}
    \max_{1\leq k \leq n} k^{-(3/2+\omega_1)} \vert T_k-2k^2\vert =O_p(1),
\end{equation}
by putting together Equation \eqref{eq:lemmaA_42} and \eqref{eq:lemmaA_43}. Let $G(C^*)$ be the event defined by 
\begin{equation*}
    G(C^*)=\{\omega: \vert T_k-2k^2 \vert\leq C^* k^{3/2+\omega_1}\text{ for all }\ 1\leq k\leq n\}.
\end{equation*}
It follows from Equation \eqref{eq:Tk_2k2} that $\lim_{C^*\rightarrow\infty} \pr(G(C^*))=1$. By the Markov property and the scale transformation of $W$, we have
\begin{equation*}
    \begin{aligned}
        &\pr\left\{\max_{1\leq k\leq n}k^{-(3/4+\omega_1)}\left\vert W(T_k)-W(2k^2)  \right\vert>x\text{ and }G(C^*)\right\}\\
        \leq & \pr\left\{ \max_{1\leq k\leq n} k^{-(3/4+\omega_1)} \sup_{\vert h\vert \leq C^* k^{3/2+\omega_1}}\left\vert W(2k^2+h)-W(2k^2)  \right\vert>x \right\}\\
        \leq & 2\sum_{k=0}^\infty \pr\left\{ \sup_{0\leq h\leq C^* k^{3/2+\omega_1}} \left\vert W(2k^2+h)-W(2k^2)  \right\vert>x k^{3/4+\omega_1}   \right\}\\
        =& 2\sum_{k=0}^\infty \pr\left\{ \sup_{0\leq h\leq C^* k^{3/2+\omega_1}} \left\vert W(h)  \right\vert>x k^{3/4+\omega_1}   \right\}\\
        =& 2\sum_{k=0}^\infty \pr\left\{  (C^* k^{3/2+\omega_1})^{1/2}\sup_{0<t<1} \left\vert W(t)  \right\vert>x k^{3/4+\omega_1}   \right\}\\
        =& 2\sum_{k=0}^\infty \pr\left\{ \sup_{0<t<1} \left\vert W(t)  \right\vert>x k^{\omega_1/2} (C^*)^{-1/2}  \right\}\\
        \leq & C\sum_{k=1}^\infty \exp\left( -\frac{x^2k^{\omega_1}}{3C^*}\right)\rightarrow 0,
    \end{aligned}
\end{equation*}
as $x\rightarrow\infty$, where $C$ is a constant and in the last step we used Lemma 1.2.1 of \cite{2014-csorgo+Revese-strong}. The proof of Lemma \ref{lemma:approxi_wiener} is completed. 
 \end{proof}

 We decompose the $H_{np}(k)/\sqrt{2\tr(\R^2)}$ as,
 \begin{equation*}
     \frac{H_{np}(k)}{\sqrt{2\tr(\R^2)}}= \frac{n}{k(n-k)}Z_{np}(k)-H_{np}^{(1)}(k)+H_{np}^{(2)}(k),
 \end{equation*}
 where
 \begin{equation*}
     H_{np}^{(1)}(k)=\frac{2p}{(n-k)\sqrt{2\tr(\R^2)}}(\S_k^\top \S_n-k), H_{np}^{(2)}(k)=\frac{kp}{n(n-k)\sqrt{2\tr(\R^2)}}(\S_n^\top \S_n-n).
 \end{equation*}
 By the definition of $H_{np}^{(2)}(k)$, we have
 \begin{equation*}
     \max_{1\leq k\leq n/2}\frac{n}{k}\vert H_{np}^{(2)}(k) \vert\leq \frac{2p}{n\sqrt{2\tr(\R^2)}}\sum_{1\leq i<j\leq n} \U_i^\top \U_j.
 \end{equation*}
 By the proof of lemma 4 in \cite{liu+feng+wang+2024}, we see that, $2p/n\sum_{1\leq i<j\leq n}\U_i^\top \U_j/\sqrt{2\tr(\R^2)}\stackrel{d}{\rightarrow}N(0,1)$, thus 
 \begin{equation}\label{eq:H_npk^2}
     \max_{1\leq k\leq n/2}\frac{n}{k}\vert H_{np}^{(2)}(k)\vert=O_p(1).
 \end{equation}

By the definition of $H_{np}^{(1)}(k)$, for any $M\leq n/2$, we have
\begin{equation*}
    \max_{1\leq k\leq M} \frac{p}{n\sqrt{2\tr(\R^2)}}\sum_{j=1}^k\sum_{1\leq l\not= j\leq n}\U_j^\top \U_l.
\end{equation*}

By Rosenthal inequality\citep{petrov1995limit}, for any $k_1<k_2$, 
\begin{equation*}
    \begin{aligned}
        &\E \left\vert\sum_{j=k_1}^{k_2}\sum_{1\leq l\not= j\leq n} \U_j^\top \U_l\right\vert^4\\
    \lesssim & \sum_{j=k_1}^{k_2}\sum_{1\leq l\not= j\leq n} \E \left\vert \U_j^\top \U_l\right\vert^4+ \left(\sum_{j=k_1}^{k_2}\sum_{1\leq l\not= j\leq n} \E \left\vert \U_j^\top \U_l\right\vert^2\right)^2\\
    =& (k_2-k_1)(n-1)\frac{\tr^2(\R^2)}{p^4}+(k_2-k_1)^2(n-1)^2\left\{\frac{\tr(\R^2)}{p^2}\right\}^2\\
    \lesssim &\frac{(k_2-k_1)^2n^2\tr^2(\R^2)}{p^4}.
    \end{aligned}
\end{equation*}

 By the maximal inequality of \cite{moricz1982moment}, for all $M\geq 1$,
 \begin{equation}\label{eq:H_npk^1}
     \E\left( \max_{1\leq k\leq M} \frac{p}{n\sqrt{2\tr(\R^2)}}\sum_{j=1}^k\sum_{1\leq l\not= j\leq n}\U_j^\top \U_l \right)^4\lesssim \frac{p^4}{n^4\tr^2(\R^2)}\frac{M^2n^2\tr^2(\R^2)}{p^4}=\frac{M^2}{n^2}
 \end{equation}
 
 Let $M=\lambda_n,n/\lambda_n,n/2$, where $\lambda_n\sim n^{\lambda}$, we get
 \begin{equation}
     \begin{aligned}
          \max_{1\leq k\leq \lambda_n}\vert H_{n p}^{(1)}(k)\vert &=O_p(n^{-(1-\lambda)/2}),\\
         \max_{1\leq k\leq n/\lambda_n}\vert H_{n p}^{(1)}(k)\vert &=O_p(n^{-\lambda/2}),\\
         \max_{1\leq k\leq  n/2}\vert H_{n p}^{(1)}(k)\vert &=O_p(1).\\
     \end{aligned}
 \end{equation}
 
 By Equation \eqref{eq:H_npk^2}, \eqref{eq:H_npk^1} and Lemma \ref{lemma:approxi_wiener}, we have
\begin{equation}\label{eq:H_npk_Wnp2k2}
     \max_{1\leq k\leq n/2}\vert H_{np}(k)-\frac{n}{k(n-k)}W_{np}(2k^2)\vert =O_p(1).
 \end{equation}
 
 According to the law of iterated algorithm, 
 \begin{equation}\label{eq:limsup}
     \limsup_{k\rightarrow\infty} \{4k^2\log\log(2k^2)\}^{-1/2} \vert W(2k^2)\vert=1,\text{ a.s.,}
 \end{equation}
 where $W$ stands for a Wiener process. By the Darling-Erd{\"o}s law\citep{csorgo1997limit}, we have
\begin{equation}\label{eq:W2k2_k_dt_n2}
     \max_{n/\lambda_n\leq k\leq n/2} \vert W(2k^2)\vert/k=O_p\{(\log\log\log n)^{1/2}\},
 \end{equation}
 and
\begin{equation}\label{eq:W2k2_k_bt_dt}
     (\log\log n)^{-1/2}\max_{\lambda_n\leq k\leq n/\lambda_n}\vert W(2k^2)\vert/k\stackrel{p}{\rightarrow}1.
 \end{equation}
 
Since the distribution of $W_{np}$ does not depend on $n$ and $p$, we get
\begin{equation}\label{eq:comp_lambdan_le}
    \max_{1\leq k\leq \lambda_n}\vert H_{np}(k)\vert=O_p\{(\log\log\log n)^{1/2}\},
\end{equation}
by Equation \eqref{eq:H_npk_Wnp2k2} and \eqref{eq:limsup}.

Combining Equation \eqref{eq:H_npk_Wnp2k2}, \eqref{eq:W2k2_k_dt_n2} and \eqref{eq:W2k2_k_bt_dt}, we get,
\begin{equation}\label{eq:comp_lambdan_ri}
    \max_{n/\lambda_n\leq k\leq n/2} \vert H_{np}(k)\vert =O_p\{(\log\log\log n)^{1/2}\},
\end{equation}
and
\begin{equation}\label{eq:comp_lambdan_mid}
    (\log\log n)^{-1/2}\max_{\lambda_n\leq k\leq n/\lambda_n}\vert H_{np}(k)\vert \stackrel{p}{\rightarrow}1.
\end{equation}

Let $\xi_{np}$ denotes the location of the maximum of $H_{np}(k)$ on $[1,n/2]$, then by Equation \eqref{eq:comp_lambdan_le}-\eqref{eq:comp_lambdan_mid}, we have
\begin{equation*}
    \pr(\lambda_n\leq\xi_{np}\leq n/\lambda_n)\rightarrow 1,
\end{equation*}
as $\min(n,p)\rightarrow\infty$, i.e.
\begin{equation}\label{eq:appro_Hnpk_all_bd}
    \pr\left\{ \max_{1\leq k\leq n/2}\vert H_{np}(k)\vert=\max_{\lambda_n\leq k\leq n/\lambda_n}\vert H_{np}(k) \vert \right\}\rightarrow 1.
\end{equation}

By Equation \eqref{eq:H_npk^2} and \eqref{eq:H_npk^1},
\begin{equation}\label{eq:appro_Hnpc_Znpk}
    \max_{\lambda_n\leq k\leq n/\lambda_n} \left\vert H_{np}(k)-k^{-1}Z_{np}(k) \right\vert=O_p(n^{-\lambda/2}).
\end{equation}

Combining Equation \eqref{eq:appro_Hnpk_all_bd} and \eqref{eq:appro_Hnpc_Znpk}, we get
\begin{equation}
    \max_{1\leq k\leq n/2}\vert H_{np}(k)\vert =\max_{\lambda_n\leq k\leq n/\lambda_n}\left\vert k^{-1}\frac{p}{\sqrt{2\tr(\R^2)}}\left(\S_k^\top \S_k-k\right)\right\vert +O_p(n^{-\lambda/2}).
\end{equation}

By symmetric,
\begin{equation}
\begin{aligned}
     &\max_{n/2\leq k\leq n}\vert H_{np}(k)\vert\\
     =&\max_{n-n/\lambda_n\leq k\leq n-\lambda_n}\left\vert (n-k)^{-1}\frac{p}{\sqrt{2\tr(\R^2)}}\left\{(\S_n-\S_k)^\top (\S_n-\S_k)-(n-k)\right\}\right\vert +O_p(n^{-\lambda/2}).
     \end{aligned}
\end{equation}

By the above two equations, it is enough to consider the limit distribution of 
\begin{equation*}
    \begin{aligned} 
   & Q_{np}^{(1)}=\max\left\{ \max_{\lambda_n\leq k\leq n/\lambda_n} \left\vert k^{-1}\frac{p}{\sqrt{2\tr(\R^2)}}\left(\S_k^\top \S_k-k\right)\right\vert , \right.\\
    &\left.\max_{n-n/\lambda_n\leq k\leq n-\lambda_n}\left\vert (n-k)^{-1}\frac{p}{\sqrt{2\tr(\R^2)}}\left\{(\S_n-\S_k)^\top (\S_n-\S_k)-(n-k)\right\}\right\vert \right\}.
        \end{aligned}
\end{equation*}

We notice that, $\{\S_k,1\leq k\leq n/2\}$ and $\{\S_n-\S_k,n/2<k\leq n\}$ are independent. By lemma \ref{lemma:approxi_wiener} we have, 
\begin{equation*}
Q_{np}^{(1)}\stackrel{d}{=}Q_{p}^{(2)}+O_p(n^{-\lambda(1/4-\omega_1)}),
\end{equation*}
where
\begin{equation*}
    Q_{p}^{(2)}=\max\left\{ \max_{\lambda_n,n-n/\lambda_n} \vert W^{(1)}(2k^2) \vert/k,\max_{\lambda_n,n-n/\lambda_n} \vert W^{(2)}(2k^2) \vert/k\right\},
\end{equation*}
where $W^{(1)}$ and $W^{(2)}$ are independent Wiener processes. By the same argument with minor modification in \citep{2013-chan+Horvath+Huskova-p955},  the conclusion follows.

\subsection{Proof of Theorem  \ref{thm:ind_H0}}
\subsubsection{For Gaussian type}
From Sections \ref{sec:proofMM}--\ref{sec:proofMS}, we verify
\begin{equation*}
    M_{n,p}=\max_{\lambda_n\leq k\leq n-\lambda_n}\Vert \C^{U}_{0}(k)\Vert_{\infty}+o_p(1),\ S_{n,p}=\max_{\lambda_n\leq k\leq n-\lambda_n}\Vert \C^{U}_{0}(k)\Vert^2/\sqrt{2\text{tr}(\R^2)}+o_p(1),
\end{equation*}
where $\C^U_0(k)=n^{-1/2}\zeta_1^{-1}\left(\S_{k}-k/n \S_{n}\right)$, $\S_k=\sum_{i=1}^k \U_i$.

We first investigate the asymptotic independence of $p^{1/2}\zeta_1\max_{1\leq k\leq n}\Vert \C^{U}_{0}(k)\Vert_{\infty}$ and $p^{1/2}\zeta_1$
$\max_{\lambda_n\leq k\leq n-\lambda_n}\Vert \C^{U}_{0}(k)\Vert^2$ if $\U_i\sim N(0,\R/p)$. We define $A_p=\{p^{1/2}\zeta_1\max_{\lambda_n\leq k\leq n-\lambda_n}\Vert \C^{U}_{0}(k)\Vert^2\leq \sqrt{2\text{tr}(\R^2)}x\}$ and $B_{j}:=B_j(y)=\{p^{1/2}\zeta_1\max_{\lambda_n\leq k\leq n-\lambda_n}\vert C^U_{0,j}(k)\vert>u_p\{\exp(-y)\}\}$, $j=1,\ldots,p$. Our goal is to prove that,
\begin{equation*}
    \begin{aligned}
    &\pr\left(p^{1/2}\zeta_1\max_{\lambda_n\leq k\leq n-\lambda_n}\Vert \C^{U}_{0}(k)\Vert^2\leq \sqrt{2\text{tr}(\R^2)}x,p^{1/2}\zeta_1\max_{\lambda_n\leq k\leq n-\lambda_n}\Vert \C^{U}_{0}(k)\Vert_{\infty}\leq u_p\{\exp(-x)\}\right)\\
    \rightarrow & F_V(x)\cdot \exp\{-\exp(-y)\},
    \end{aligned}
\end{equation*}
or equivalently,
\begin{equation*}
    \pr(\bigcup_{j=1}^p A_p B_j)\rightarrow F_V(x)\cdot \exp\{-\exp(-y)\}.
\end{equation*}

Let for each $d\geq 1$,
\begin{equation*}
    \zeta(p,d):=\sum_{1\leq j_1,<\cdots<j_d\leq p} \vert \pr(A_p B_{j_1}\cdots B_{j_d})-\pr(A_p)\pr(B_{j_1}\cdots B_{j_d})\vert,
\end{equation*}
and 
\begin{equation*}
    H(p,d):=\sum_{1\leq j_1,<\cdots<j_d\leq p}\pr(B_{j_1}\cdots B_{j_d}).
\end{equation*}

By the inclusion-exclusion principle, we see that, for any integer $k\geq 1$,
\begin{equation*}
    \begin{aligned}
         \pr(\bigcup_{j=1}^p A_p B_j)\leq &\sum_{1\leq j_1\leq p}\pr(A_p B_{j_1})-\sum_{1\leq j_1<j_2\leq p}\pr(A_p B_{j_1}B_{j_2})+\cdots\\
         &\qquad +\sum_{1\leq j_1<\cdots<j_{2k+1}\leq p}\pr(A_p B_{j_1}\cdots B_{j_{2k+1}}),
    \end{aligned}
\end{equation*}
and
\begin{equation*}
\begin{aligned}
    \pr(\bigcup_{j=1}^p B_j)\geq \sum_{1\leq j_1\leq p}&\pr(B_{j_1})-\sum_{1\leq j_1<j_3\leq p}\pr(B_{j_1}B_{j_2})+\cdots\\
    &-\sum_{1\leq j_1<\cdots<j_{2k}\leq p}\pr(B_{j_1}\cdots B_{j_{2k}}).
\end{aligned}
\end{equation*}

Then, we have
\begin{equation*}
    \begin{aligned}
        \pr(\bigcup_{j=1}^p A_p B_j)\leq &\pr(A_p)\left\{\sum_{1\leq j_1\leq p}\pr(B_{j_1})-\sum_{1\leq j_1<j_2\leq p}\pr(B_{j_1}B_{j_2})+\cdots \right.\\
      & \left.- \sum_{1\leq j_1<\cdots<j_{2k}}\pr(B_{j_1}\cdots B_{j_{2k}}) \right\}+\sum_{d=1}^{2k}\zeta(p,d)+H(p,2K+1)\\
       \leq &\pr(A_p)\pr(\bigcup_{j=1}^p B_j)+\sum_{d=1}^{2k}\zeta(p,d)+H(p,2k+1).
    \end{aligned}
\end{equation*}

By fixing $k$ and letting $p\rightarrow \infty$, and combining Lemma \ref{lemma_zetapd}, we obtain
\begin{equation*}
    \limsup_{p\rightarrow\infty }\pr(\bigcup_{j=1}^p A_p B_j)\leq F_V(x)[1-\exp \{-\exp(-y)\}]+\lim_{p\rightarrow\infty }H(p,2k+1).
\end{equation*}

According to the Equation (S.5) and (S.6) in \cite{wang2023}, and $p^{1/2 }\U_i\sim N(0,\R)$, we have $\min_{p\rightarrow \infty } H(p,d)=\frac{1}{d!}\exp(-d x/2)$. By letting $k\rightarrow\infty$ , we have
\begin{equation*}
    \limsup_{p\rightarrow\infty }\pr(\bigcup_{j=1}^p A_p B_j)\leq F_V(x)\left[1-\exp \{-\exp(-y)\}\right].
\end{equation*}

Using the similar arguments, we acquire
\begin{equation*}
    \begin{aligned}
 \pr(\bigcup_{j=1}^p A_p B_j)\geq &\sum_{1\leq j_1\leq p}\pr(A_p B_{j_1})-\sum_{1\leq j_1<j_2\leq p}\pr(A_p B_{j_1}B_{j_2})+\cdots\\
         &\qquad -\sum_{1\leq j_1<\cdots<j_{2k}\leq p}\pr(A_p B_{j_1}\cdots B_{j_{2k}}),    \end{aligned}
\end{equation*}
and 
\begin{equation*}
\begin{aligned}
    \pr(\bigcup_{j=1}^p B_j)\leq \sum_{1\leq j_1\leq p}&\pr(B_{j_1})-\sum_{1\leq j_1<j_3\leq p}\pr(B_{j_1}B_{j_2})+\cdots\\
    &+\sum_{1\leq j_1<\cdots<j_{2k-1}\leq p}\pr(B_{j_1}\cdots B_{j_{2k-1}}).
\end{aligned}
\end{equation*}

We obtain,
\begin{equation*}
    \liminf_{p\rightarrow \infty } \pr(\bigcup_{j=1}^p A_p B_j)\geq F_V(x) [1-\exp\{-\exp(-y)\}].
\end{equation*}
\begin{lemma}\label{lemma_zetapd}
    Suppose the assumptions in Theorem \ref{thm:ind_H0} holds, then for each $d\geq 1$, $\zeta(p,d)\rightarrow 0$.
\end{lemma}
\begin{proof}
For convenience, we define $\tilde{\U}_i=p^{1/2}\U_i\sim N(0,\R)$. For each $i=1, \ldots, n$, let $\tilde{\U}_{i,(1)}=\left(\tilde{\U}_{i, j_1}, \ldots, \tilde{\U}_{i, j_d}\right)^{\top}$ and $\tilde{\U}_{i,(2)}=\left(\tilde{\U}_{i, j_{d+1}}, \ldots, \tilde{\U}_{i, j_p}\right)^{\top}$, and $\R_{k l}=\operatorname{Cov}\left(\tilde{\U}_{i,(k)}, \tilde{\U}_{i,(l)}\right)$ for $k, l \in\{1,2\}$. By Lemma \ref{lemma_normaldecom}, $\tilde{\U}_{i,(2)}$ can be decomposed as $\tilde{\U}_{i,(2)}=$ $\V_i+\T_i$, where $\V_i:=\tilde{\U}_{i,(2)}-\R_{21} \R_{11}^{-1} \tilde{\U}_{i,(1)}$ and $\T_i:=\R_{21} \R_{11}^{-1} \tilde{\U}_{i,(1)}$ satisfying that $\V_i \sim N\left(0, \R_{22}-\R_{21} \R_{11}^{-1} \R_{12}\right), \T_i \sim N\left(0, \R_{21} \R_{11}^{-1} \R_{12}\right)$ and
$\V_i$ and $\tilde{\U}_{i,(1)}$ are independent.
    Let $MS_{n,p}=n^{-1}\sum_{1\leq j\leq p}\max_{\lambda_n\leq k\leq n-\lambda_n}$ $(\sum_{i=1}^k p^{1/2}U_{i,j}-\frac{k}{n}\sum_{i=1}^n p^{1/2}U_{i,j})^2$ and we can decompose it as,
 \begin{equation*}
     \begin{aligned}
         MS_{n,p}=n^{-1}\sum_{l\in\{1,2,\ldots,n-d\}}\max_{\lambda_n\leq k\leq n-\lambda_n}(\sum_{i=1}^k V_{i,l}-\sum_{i=1}^k k/n V_{i,l})^2+\Theta:=MS_{n,p}^{*}+\Theta,
     \end{aligned}
 \end{equation*}
 where 
 \begin{equation*}
 \begin{aligned}
     \Theta\leq&n^{-1}\sum_{j\in\{j_1,\ldots,j_d\}}\max_{\lambda_n\leq k\leq n-\lambda_n}(\sum_{i=1}^k \tilde{U}_{i,j}-\sum_{i=1}^k\frac{k}{n}\tilde{U}_{i,j})^2+\\
     &n^{-1}\sum_{l\in\{1,2,\ldots,n-d\}}\max_{\lambda_n\leq k\leq n-\lambda_n}(\sum_{i=1}^k T_{i,l}-\sum_{i=1}^k\frac{k}{n}T_{i,l})^2+\\
     &\qquad 2n^{-1}\sum_{l\in\{1,2,\ldots,n-d\}}\max_{\lambda_n\leq k\leq n-\lambda_n}(\sum_{i=1}^k T_{i,l}-\sum_{i=1}^k\frac{k}{n}T_{i,l})(\sum_{i=1}^k V_{i,l}-\sum_{i=1}^k\frac{k}{n}V_{i,l})\\
     :=& \Theta_1+\Theta_2+\Theta_3.
     \end{aligned}
 \end{equation*}
 
We claim that, for any $\varsigma>0$, there exists a sequence of positive constant $t=: t_p>0$ with $t_p\rightarrow\infty$ such that,
\begin{equation}\label{eq:theta}
    \pr(\vert\Theta_i\vert\geq \sqrt{2\text{tr}(\R^2)}\varsigma)\leq p^{-t},i=1,2,3.
\end{equation}

Consequently, $\pr\{\vert\Theta\vert>\varsigma \sqrt{2\text{tr}(\R^2)}\}\leq p^{-t}$ for some $t\rightarrow\infty$ and sufficiently large $p$. $A_p(x)$ can rewritten as $A_p=\{MS_{n,p}^*/\sqrt{2\text{tr}(\R^2)}+\Theta/\sqrt{2\text{tr}(\R^2)}\leq x\}$. By Lemma \ref{lemma_normaldecom}, we have 
\begin{equation*}
\begin{aligned}
    \pr(A_p(x)B_{j_1}\cdots B_{j_d})&\leq \pr(A_p(x)B_{j_1}\cdots B_{j_d},\vert\Theta\vert/\sqrt{2\text{tr}(\R^2)}\leq \varsigma)+p^{-t}\\
    &\leq \pr(MS_{n,p}^*/\sqrt{2\text{tr}(\R^2)}\leq x+\varsigma,B_{j_1}\cdots B_{j_d})+p^{-t}\\
    &=\pr(MS_{n,p}^*/\sqrt{2\text{tr}(\R^2)}\leq x+\varsigma)\pr(B_{j_1}\cdots B_{j_d})+p^{-t}.
\end{aligned}
\end{equation*}

We also have 
\begin{equation*}
    \begin{aligned}
        \pr(MS_{n,p}^*/\sqrt{2\text{tr}(\R^2)}\leq x+\varsigma)&\leq \pr(MS_{n,p}^*/\sqrt{2\text{tr}(\R^2)}\leq x+\varsigma,\vert\Theta\vert/\sqrt{2\text{tr}(\R^2)}<\varsigma)+p^{-t}\\
        &\leq \pr(A_{p}(x+2\varsigma))+p^{-t}.
    \end{aligned}
\end{equation*}

Thus, we have
\begin{equation}\label{eq:MS_l}
    \pr(A_p(x)B_{j_1}\cdots B_{j_d})\leq \pr(A_p(x+2\varsigma))\pr(B_{j_1}\cdots B_{j_d})+2p^{-t}.
\end{equation}

On the other hand, we consider
\begin{equation*}
    \begin{aligned}
        &\quad \pr(MS_{n,p}^*/\sqrt{2\text{tr}(\R^2)}\leq x-\varsigma)\leq x-\varsigma)\pr(B_{j_1}\cdots B_{j_d})\\
        &=\pr(MS_{n,p}^*/\sqrt{2\text{tr}(\R^2)}\leq x-\varsigma,B_{j_1}\cdots B_{j_d})\\
        &\leq \pr(MS_{n,p}^*/\sqrt{2\text{tr}(\R^2)}\leq x-\varsigma,B_{j_1}\cdots B_{j_d},\vert\Theta\vert/\sqrt{2\text{tr}(\R^2)}<\varsigma)+p^{-t},
    \end{aligned}
\end{equation*}
and 
\begin{equation*}
    \begin{aligned}
        \pr(A_p(x-2\varsigma))&\leq \pr(A_p(x-2\varsigma),\vert\Theta\vert/\sqrt{2\text{tr}(\R^2)})+p^{-t}\\
        &\leq \pr(MS_{n,p}^*/\sqrt{2\text{tr}(\R^2)}\leq x-\varsigma)+p^{-t}.
    \end{aligned}
\end{equation*}

Thus, we have
\begin{equation}\label{eq:MS_r}
    \pr(A_p(x)B_{j_1}\cdots B_{j_d})\geq \pr(A_p(x-2\varsigma))\pr(B_{j_1}\cdots B_{j_d})-2p^{-t}.
\end{equation}

Combining Equation (\ref{eq:MS_l}) and (\ref{eq:MS_r}), we conclude that
\begin{equation*}
    \vert \pr(A_p(x)B_{j_1}\cdots B_{j_d})-\pr(A_p(x))\pr(B_{j_1}\cdots B_{j_d})\vert\leq\Delta_{p,\varsigma}\pr(B_{j_1}\cdots B_{j_d})+2p^{-t},
\end{equation*}
for sufficiently large $p$, where 
\begin{equation*}
    \begin{aligned}
        \Delta_{p,\varsigma}=\pr(A_p(x+2\varsigma))-\pr(A_p(x-2\varsigma)),
    \end{aligned}
\end{equation*}
since $\pr(A_p(x))$ is increasing in $x$. Thus $\zeta(p,d)$ follows,
\begin{equation*}
    \zeta(p,d)\leq \Delta_{p,\varsigma}H(p,d)+2C_p^d p^{-t}.
\end{equation*}
where $C^d_p=p!/\{d!(p-d)!\}$ and $k!=\prod_{\ell=1}^k \ell$ for $k=1,2,\cdots$.

Since $\pr(A_p)\rightarrow F_V(x)$, $\Delta_{p,\varsigma}\rightarrow F_V(x+2\varsigma)-F_V(x-2\varsigma)$ as $p\rightarrow\infty$, which implies that $\lim_{\varsigma\rightarrow 0}\limsup_{p\rightarrow\infty}\Delta_{p,\varsigma}=0$. For each $d\geq 1$, $H(p,d)\rightarrow\frac{1}{d!}\exp(-d x/2)$ as $p\rightarrow \infty$, we get $\limsup_{p\rightarrow\infty}H(p,d)<\infty$. By some basic calculation, it easy to get $C_p^d p^{-t}\leq p^{d-t}$ for fixed $d\geq 1$. By letting $p\rightarrow\infty$ and then $\varsigma\rightarrow 0$, $\zeta(p,d)\rightarrow 0$ for each $d\geq 1$.

\underline{Proof of Equation (\ref{eq:theta}):}
\begin{equation*}
\begin{aligned}
        \Theta_1=&n^{-1}\max_{\lambda_n\leq k\leq n-\lambda_n}\left(\sum_{i=1}^k\tilde{\U}_{i,(1)}-\frac{k}{n}\sum_{i=1}^n\tilde{\U}_{i,(1)}\right)^\top \left(\sum_{i=1}^k\tilde{\U}_{i,(1)}-\frac{k}{n}\sum_{i=1}^n\tilde{\U}_{i,(1)}\right)\\
        :=&n^{-1}\max_{\lambda_n\leq k\leq n-\lambda_n}\left(\sum_{i=1}^n \breve{\upsilon}_{i,k}\tilde{\U}_{i,(1)}\right)^\top \left(\sum_{i=1}^n \breve{\upsilon}_{i,k}\tilde{\U}_{i,(1)}\right)\\
\end{aligned}
\end{equation*}
For $\Theta_1$,
\begin{equation*}
\begin{aligned}
&\pr(\vert\Theta_{1}\vert>\sqrt{2\text{tr}(\R^2)}\varsigma)\\
=
&\pr\left(n^{-1}\max_{\lambda_n\leq k\leq n-\lambda_n}\big(\sum_{i=1}^n \breve{\upsilon}_{i,k}\tilde{\U}_{i,(1)}\big)^\top \big(\sum_{i=1}^n \breve{\upsilon}_{i,k}\tilde{\U}_{i,(1)}\big)>\sqrt{2\text{tr}(\R^2)}\varsigma\right)\\
\leq & \pr\left(\max_{\lambda_n\leq k\leq n-\lambda_n}\big\{\sum_{i=1}^n \breve{\upsilon}_{i,k}\tilde{\U}_{i,(1)}/(\sum_{j=1}^n\breve{\upsilon}_{i,k}^2)^{1/2}\big\}^\top \big\{\sum_{i=1}^n \breve{\upsilon}_{i,k}\tilde{\U}_{i,(1)}/(\sum_{j=1}^n\breve{\upsilon}_{i,k}^2)^{1/2}\big\}>\sqrt{2\text{tr}(\R^2)}\varsigma\right)\\
\leq& n \pr(\vert\tilde{\U}_{1,(1)}^\top \tilde{\U}_{1,(1)}\vert>C_\varsigma \sqrt{2\text{tr}(\R^2)})\\
\leq & n\exp(-C_\varsigma d^{-1}p^{1/2}),
    \end{aligned}
\end{equation*}
where the last inequality holds by Lemma S.7 in \cite{Feng2022AsymptoticIO}, or the proof of Theorem 4 in \cite{wang2023} and $C_{\varsigma}$ denotes some
positive constant depending on $\varsigma$. Similarly, for $\Theta_2$ and $\Theta_3$,
\begin{equation*}
\begin{aligned}
        \Theta_2=&n^{-1}\max_{\lambda_n\leq k\leq n-\lambda_n}\left(\sum_{i=1}^k\T_i-\frac{k}{n}\sum_{i=1}^n \T_i\right)^\top \left(\sum_{i=1}^k\T_i-\frac{k}{n}\sum_{i=1}^n \T_i\right)\\
        :=&n^{-1}\max_{\lambda_n\leq k\leq n-\lambda_n}\left(\sum_{i=1}^n \breve{\upsilon}_{i,k}\T_{i}\right)^\top \left(\sum_{i=1}^n \breve{\upsilon}_{i,k}\T_{i}\right)\\
\end{aligned}
\end{equation*}
\begin{equation*}
\begin{aligned}
        \Theta_3=&n^{-1}\max_{\lambda_n\leq k\leq n-\lambda_n}\left(\sum_{i=1}^k\T_i-\frac{k}{n}\sum_{i=1}^n \T_i\right)^\top \left(\sum_{i=1}^k \V_i-\frac{k}{n}\sum_{i=1}^n \V_i\right)\\
                :=&n^{-1}\max_{\lambda_n\leq k\leq n-\lambda_n}\left(\sum_{i=1}^n \breve{\upsilon}_{i,k}\T_{i}\right)^\top \left(\sum_{i=1}^n \breve{\upsilon}_{i,k}\V_{i}\right)\\
\end{aligned}
\end{equation*}

\begin{equation*}
    \begin{aligned}
        &~~~~\pr(\vert\Theta_{2}\vert>\sqrt{2\text{tr}(\R^2)}\varsigma)\\
&=\pr\left(n^{-1}\max_{\lambda_n\leq k\leq n-\lambda_n}\big(\sum_{i=1}^n \breve{\upsilon}_{i,k}\T_{i,(1)}\big)^\top \big(\sum_{i=1}^n \breve{\upsilon}_{i,k}\T_{i,(1)}\big)>\sqrt{2\text{tr}(\R^2)}\varsigma\right)\\
        &\leq  \pr\left(\max_{\lambda_n\leq k\leq n-\lambda_n}\big\{\sum_{i=1}^n \breve{\upsilon}_{i,k}\T_{i}/(\sum_{j=1}^n\breve{\upsilon}_{i,k}^2)^{1/2}\big\}^\top \big\{\sum_{i=1}^n \breve{\upsilon}_{i,k}\T_{i}/(\sum_{j=1}^n\breve{\upsilon}_{i,k}^2)^{1/2}\big\}>\sqrt{2\text{tr}(\R^2)}\varsigma\right)\\
        &\leq n \pr(\vert \T_1^\top \T_1\vert>C_{\varsigma}\sqrt{\text{tr}(\R^2)})\\
        &\leq n\exp\left\{-C_\varsigma \frac{\sqrt{2\text{tr}(\R^2)}}{\lambda_{\max}(\R)} \right\},
    \end{aligned}
\end{equation*}
and
\begin{equation*}
    \begin{aligned}
        &~~~~\pr(\vert\Theta_{3}\vert>\sqrt{2\text{tr}(\R^2)}\varsigma)\\
&=\pr\left(n^{-1}\max_{\lambda_n\leq k\leq n-\lambda_n}\big(\sum_{i=1}^n \breve{\upsilon}_{i,k}\T_{i,(1)}\big)^\top \big(\sum_{i=1}^n \breve{\upsilon}_{i,k}\V_{i,(1)}\big)>\sqrt{2\text{tr}(\R^2)}\varsigma\right)\\
        &\leq  \pr\left(\max_{\lambda_n\leq k\leq n-\lambda_n}\big\{\sum_{i=1}^n \breve{\upsilon}_{i,k}\T_{i}/(\sum_{j=1}^n\breve{\upsilon}_{i,k}^2)^{1/2}\big\}^\top \big\{\sum_{i=1}^n \breve{\upsilon}_{i,k}\V_{i}/(\sum_{j=1}^n\breve{\upsilon}_{i,k}^2)^{1/2}\big\}>\sqrt{2\text{tr}(\R^2)}\varsigma\right)\\
        &\leq n \pr(\vert \T_1^\top \V_1\vert>C_{\varsigma}\sqrt{\text{tr}(\R^2)})\\
        &\leq n\exp\left\{-C_\varsigma \frac{\sqrt{2\text{tr}(\R^2)}}{\lambda_{\max}(\R)} \right\},
    \end{aligned}
\end{equation*}

It is then easy to see that the Equation (\ref{eq:theta}) holds.
\end{proof}

\subsubsection{For non-Gaussian type}
From the Section \ref{sec:proofMM}-\ref{sec:proofMS}, we verify
\begin{equation*}
    S_{n,p}=p n^{-1}\max_{\lambda_n\leq k\leq n-\lambda_n}2\sum_{i<j} \breve{\upsilon}_{i,k}\breve{\upsilon}_{j,k}\U_i^\top \U_j/\sqrt{2\text{tr}(\R^2)}+o_p(1),
\end{equation*}
\begin{equation*}
    M_{n,p}=p^{1/2}n^{-1/2}\max_{\lambda_n\leq k\leq n-\lambda_n}\max_{1\leq j\leq p}\vert\sum_{i=1}^n \breve{\upsilon}_{i,k}U_{i,l}\vert+o_p(1),
\end{equation*}
where $\breve{\upsilon}_{i,k}=\ind{i\leq k}-k/n$.
It suffice to show that:
\begin{equation}\label{eq:H0ind}
\begin{aligned}
    \pr&\left(p n^{-1}\max_{\lambda_n\leq k\leq n-\lambda_n}2\sum_{i<j} \breve{\upsilon}_{i,k}\breve{\upsilon}_{j,k}\U_i^\top \U_j/\sqrt{2\text{tr}(\R^2)}\leq x,\right.\\
    &\qquad\qquad\qquad\qquad\left. p^{1/2}\zeta_1\max_{\lambda_n\leq k\leq n-\lambda_n}\max_{1\leq j\leq p}\vert\sum_{i=1}^n \breve{\upsilon}_{i,k}U_{i,l}\vert\leq u_p\{\exp(-y)\}\right)\\
    &\rightarrow F_V(x)\exp\{-\exp(-y)\}.
    \end{aligned}
\end{equation}

For $\boldsymbol z=(z_1,\ldots,z_q)^\top\in \mathbb R^q$, we consider a smooth approximation of the maximum function, namely, 
$$
F_\beta(\boldsymbol z):=\beta^{-1}\log(\sum_{j=1}^q\exp(\beta z_j)),
$$
where $\beta>0$ is the smoothing parameter that controls the level of approximation. An elementary calculation shows that for all $z\in \mathbb R^q$, 
$$
0\leq F_\beta (\boldsymbol z)-\max_{1\leq j\leq q}z_j\leq \beta^{-1}\log q.
$$

We define,
\begin{equation*}
    \begin{aligned}
        W(x_1,\ldots,x_n)&=\beta^{-1}\log\left(\sum_{k=\lambda_n}^{n-\lambda_n}\exp\left\{ 2\beta  p n^{-1}\sum_{i<j} \breve{\upsilon}_{i,k}\breve{\upsilon}_{j,k}x_i^\top x_j/\sqrt{2\text{tr}(\R^2)} \right\}\right)\\
        :&=\beta^{-1}\log\left(\sum_{k=\lambda_n}^{n-\lambda_n}\exp\left\{ \beta  p n^{-1}\sum_{1\leq i<j\leq n} b_{i,j,k}x_i^\top x_j/\sqrt{2\text{tr}(\R^2)} \right\}\right),
    \end{aligned}
\end{equation*}
\begin{equation*}
    V(x_1,\ldots,x_n)=\beta^{-1} \log\left\{ \sum_{j=1}^p \sum_{k=\lambda_n}^{n-\lambda_n} \exp\left(\beta n^{-1/2} \sum_{t=1}^n \breve{\upsilon}_{t,k}x_{t j}\right)\right\}.
\end{equation*}

By setting $\beta=n^{1/8\wedge\omega_1}\log(np)$, Equation \eqref{eq:H0ind} is equivalent to 
\begin{equation}\label{eq:H0ind_WV}
\begin{aligned}
\pr\left(W(\U_1,\ldots,\U_p)\leq x,V(x_1,\ldots,x_n)\leq u_p\{\exp(-y)\}\right)\rightarrow F_V(x)\exp\{-\exp(-y)\}.
\end{aligned}
\end{equation}

Suppose $\{\boldsymbol Y_1,\boldsymbol Y_2,\ldots,\boldsymbol Y_n\}$ are sample from $N(0,\E( \boldsymbol U_1^\top \boldsymbol U_1))$ , and independent with $\boldsymbol U_1,\ldots,\boldsymbol U_n$.  The key idea is to show that: $(W(\boldsymbol U_1,\ldots,\boldsymbol U_n),V(\boldsymbol U_1,\ldots,\boldsymbol U_n))$ has the same limiting distribution as $(W(\boldsymbol Y_1,\ldots,\boldsymbol Y_n),V(\boldsymbol Y_1,\ldots,\boldsymbol Y_n))$.

Let $l^2_b(\mathbb R)$ denote the class of bounded functions with bounded and continuous derivatives up to order 3.It is known that a sequence of randon variables $\{Z_n\}_{n=1}^\infty$ converges weakly to a random variable $Z$ if and only if for every $f\in l^3_b(\mathbb R)$, $\E(f(Z_n))\rightarrow \E(f(Z))$. It suffices to show that:
$$
\E\{f(W(\boldsymbol U_1,\ldots,\boldsymbol U_n),V(\boldsymbol U_1,\ldots,\boldsymbol U_n))\}-\E\{f(W(\boldsymbol Y_1,\ldots,\boldsymbol Y_n),V(\boldsymbol Y_1,\ldots,\boldsymbol Y_n))\}\rightarrow 0,
$$
for every $f\in l_b^3(\mathbb R^2)$ as $(n,p)\rightarrow \infty$.

We introduce $\tilde { W}_d=W(\boldsymbol U_1,\ldots,\boldsymbol U_{d-1},\boldsymbol Y_d,\ldots,\boldsymbol Y_n)$  and $\tilde{ V}_d=V(\boldsymbol U_1,\ldots,\boldsymbol U_{d-1},\boldsymbol Y_d,\ldots,\boldsymbol Y_n)$ for $d=1,\ldots,n+1$, $\mathcal F_d=\sigma\{\boldsymbol U_1,\ldots,\boldsymbol U_{d-1},\boldsymbol Y_{d+1},\ldots,\boldsymbol Y_n\}$ for $d=1,\ldots,n$. If there is no danger of confusion, we simply write $\tilde { W}_d$ and $\tilde { V}_d$ as $ W_d$ and $ V_d$ for this part, respectively. Then,
\begin{equation*}
\begin{aligned}
   &\left\vert\E\left\{f(W(\boldsymbol U_1,\ldots,\boldsymbol U_n),V(\boldsymbol U_1,\ldots,\boldsymbol U_n))\right\}-\E\left\{f(W(\boldsymbol Y_1,\ldots,\boldsymbol Y_n),V(\boldsymbol Y_1,\ldots,\boldsymbol Y_n))\right\}\right\vert\\
\leq &\sum_{d=1}^n  \left\vert \E\{f(W_d,V_d)-\E\{f(W_{d+1},V_{d+1}) \}\right\vert. 
\end{aligned}
\end{equation*}

Let 
\begin{equation*}
\begin{aligned}
W_{d,0}&=\beta^{-1}\log\left(\sum_{k=\lambda_n}^{n-\lambda_n}\exp\left\{ \beta  p n^{-1}\left(\sum_{1\leq i<j\leq d-1} b_{i,j,k}\U_i^\top \U_j +\sum_{d+1\leq i<j\leq n} b_{i,j,k}\Y_i^\top \Y_j \right.\right.\right.\\
&\qquad\qquad\qquad\qquad\qquad\qquad\qquad\left.\left.\left.+\sum_{i=1}^{d-1}\sum_{j=d+1}^n b_{i,j,k}\U_i^\top \Y_j \right)/\sqrt{2\text{tr}(\R^2)}\right\}\right)\in\mathcal F_d,\\
V_{d,0}&= \beta^{-1} \log\left\{ \sum_{j=1}^p \sum_{k=\lambda_n}^{n-\lambda_n} \exp\left(\beta n^{-1/2}p^{1/2} \sum_{t=1}^{d-1} \breve{\upsilon}_{t,k}U_{t j}+\beta n^{-1/2} p^{1/2}\sum_{t=d+1}^{n} \breve{\upsilon}_{t,k}Y_{t j}\right)\right\} \in \mathcal F_d.\\
\end{aligned}
\end{equation*}

By Taylor expansion, we have
$$
\begin{aligned}
f\left(W_d, V_d\right)-f\left(W_{d, 0}, V_{d, 0}\right)= & f_1\left(W_{d, 0}, V_{d, 0}\right)\left(W_d-W_{d, 0}\right)+f_2\left(W_{d, 0}, V_{d, 0}\right)\left(V_d-V_{d, 0}\right) \\
& +\frac{1}{2} f_{11}\left(W_{d, 0}, V_{d, 0}\right)\left(W_d-W_{d, 0}\right)^2+\frac{1}{2} f_{22}\left(W_{d, 0}, V_{d, 0}\right)\left(V_d-V_{d, 0}\right)^2 \\
& +\frac{1}{2} f_{12}\left(W_{d, 0}, V_{d, 0}\right)\left(W_d-W_{d, 0}\right)\left(V_d-V_{d, 0}\right) \\
& +O\left(\left|V_d-V_{d, 0}\right|^3\right)+O\left(\left|W_d-W_{d, 0}\right|^3\right),
\end{aligned}
$$
and
$$
\begin{aligned}
f\left(W_{d+1}, V_{d+1}\right)-f\left(W_{d, 0}, V_{d, 0}\right)= & f_1\left(W_{d, 0}, V_{d, 0}\right)\left(W_{d+1}-W_{d, 0}\right)+f_2\left(W_{d, 0}, V_{d, 0}\right)\left(V_{d+1}-V_{d, 0}\right) \\
& +\frac{1}{2} f_{11}\left(W_{d, 0}, V_{d, 0}\right)\left(W_{d+1}-W_{d, 0}\right)^2\\
&+\frac{1}{2} f_{22}\left(W_{d, 0}, V_{d, 0}\right)\left(V_{d+1}-V_{d, 0}\right)^2 \\
& +\frac{1}{2} f_{12}\left(W_{d, 0}, V_{d, 0}\right)\left(W_{d+1}-W_{d, 0}\right)\left(V_{d+1}-V_{d, 0}\right) \\
& +O\left(\left|V_{d+1}-V_{d, 0}\right|^3\right)+O\left(\left|W_{d+1}-W_{d, 0}\right|^3\right),
\end{aligned}
$$
where for $f:=f(x, y), f_1(x, y)=\frac{\partial f}{\partial x}, f_2(x, y)=\frac{\partial f}{\partial y}, f_{11}(x, y)=\frac{\partial f^2}{\partial^2 x}, f_{22}(x, y)=\frac{\partial f^2}{\partial^2 y}$ and $f_{12}(x, y)=\frac{\partial f^2}{\partial x \partial y}$.

We first consider $V_d,V_{d+1},V_{d,0}$. For $l=k-\lambda_n+1+(j-1)(n-2\lambda_n+1)$, let $z^v_{d,0,l}=n^{-1/2}p^{1/2}\sum_{t=1}^{d-1}U_{t j}\breve{\upsilon}_{t,k}+n^{-1/2}p^{1/2}\sum_{t=d+1}^n Y_{t j}\breve{\upsilon}_{t,k}$, $z^v_{d,l}=z^v_{d,0,l}+n^{-1/2}p^{1/2}Y_{d j}\breve{\upsilon}_{d,k}$ and $z^v_{d+1,l}=z^v_{d,0,l}+n^{-1/2}p^{1/2}U_{d j}\breve{\upsilon}_{d,k}$. Define $z^v_{d,0}=(z^v_{d,0,1},\ldots,z^v_{d,0,n p})^\top$ and $z^v_d=(z^v_{d,1},\ldots,z^v_{d,n p})^\top$. By Taylor's expansion, we have
\begin{equation}\label{eq:expansion_V}
\begin{aligned}
&~~V_d-V_{d, 0}\\
= & \sum_{l=1}^{(n-2\lambda_n+1)p } \partial_l F_\beta\left(\boldsymbol{z}_{d, 0}\right)\left( z^v_{d, l}- z^v_{d, 0, l}\right)\\
&+\frac{1}{2} \sum_{l,k=1}^{(n-2\lambda_n+1)p } \partial_k \partial_l F_\beta\left(\boldsymbol{z}_{d, 0}\right)\left( z^v_{d, l}- z^v_{d, 0, l}\right)\left( z^v_{d, k}- z^v_{d, 0, k}\right) \\
 &+\frac{1}{6} \sum_{l,k,v=1}^{(n-2\lambda_n+1)p }\partial_v \partial_k \partial_l F_\beta\left(\boldsymbol{z}_{d, 0}+\tilde\vartheta\left(\boldsymbol{z}_d-\boldsymbol{z}_{d, 0}\right)\right)\left( z^v_{d, l}- z^v_{d, 0, l}\right)\left( z^v_{d, k}-\boldsymbol z^v_{d, 0, k}\right)\left(\boldsymbol z^v_{d, v}-\boldsymbol z^v_{d, 0, v}\right),
\end{aligned}
\end{equation}
for some $\tilde\vartheta \in(0,1)$. Again, due to $\E\left(\boldsymbol U_t\right)=\E\left(\boldsymbol Y_t\right)=0$ and $\E\left(\boldsymbol  U_t \boldsymbol U_t^{\top}\right)=\E\left(\boldsymbol Y_t \boldsymbol Y_t^{\top}\right)$, we can verify that $\E\left\{ z^v_{d, l}- z^v_{d, 0, l} \mid \mathcal{F}_d\right\}=\E\left\{ z^v_{d+1, l}- z^v_{d, 0, l} \mid \mathcal{F}_d\right\}$ and 
 $\E\left\{\left( z^v_{d, l}- z^v_{d, 0, l}\right)^2 \mid \mathcal{F}_d\right\}=\E\left\{\left( z^v_{d+1, l}- z^v_{d, 0, l}\right)^2 \mid \mathcal{F}_d\right\}$.

By Lemma A.2 in \cite{chernozhukov2013gaussian}, we have
$$
\left|\sum_{l,k,v=1}^{(n-2\lambda_n+1)p }\partial_v \partial_k \partial_l F_\beta\left(\boldsymbol{z}^v_{d, 0}+\tilde\vartheta\left(\boldsymbol{z}^v_d-\boldsymbol{z}^v_{d, 0}\right)\right)\right| \leq C \beta^2,
$$
for some positive constant $C$. 
By Lemma \ref{LemmaA4}, we have $\left\|\zeta_1^{-1}  U_{i, j}\right\|_{\psi_{\alpha_0}} \lesssim \bar{B},$ for all $i=1, \ldots, n$ and $j=1, \ldots, p$, which means $\pr(\vert \sqrt{p}\xi_{i,j}\vert\geq t)\leq 2\exp\{-(ct\sqrt{p}/\zeta_1)^{\alpha_0}\}\lesssim 2\exp\{-(ct)^{\alpha_0}\}$ and  $\pr\left(\max _{1 \leq i \leq n}\left|\sqrt{p}U_{i j}\right|>C \log n \right) \rightarrow 0$. Since $\sqrt{p}Y_{t j} \sim N(0,1)$ and $\pr(\max _{1 \leq i \leq n}\left|\sqrt{p}Y_{i j}\right|>C \log n ) \rightarrow 0$,
\begin{small}
\begin{equation*}
    \begin{aligned}
&\left|\frac{1}{6} \sum_{l,k,v=1}^{(n-2\lambda_n+1)p} \partial_v \partial_k \partial_l F_\beta\left(\boldsymbol{z}^v_{d, 0}+\tilde\vartheta\left(\boldsymbol{z}^v_d-\boldsymbol{z}^v_{d, 0}\right)\right)\left(\boldsymbol z^v_{d, l}-\boldsymbol z^v_{d, 0, l}\right)\left(\boldsymbol z^v_{d, k}-\boldsymbol z^v_{d, 0, k}\right)\left(\boldsymbol z^v_{d, v}-\boldsymbol z^v_{d, 0, v}\right)\right|\\
&\qquad\leq C\beta^2n^{-3/2}\log^3(np),\\
&\left|\frac{1}{6} \sum_{l,k,v=1}^{(n-2\lambda_n+1)p}\partial_v \partial_k \partial_l F_\beta\left(\boldsymbol{z}^v_{d+1, 0}+\tilde\vartheta\left(\boldsymbol{z}^v_{d+1}-\boldsymbol{z}^v_{d, 0}\right)\right)\left(\boldsymbol z^v_{d+1, l}-\boldsymbol z^v_{d, 0, l}\right)\left(\boldsymbol z^v_{d+1, k}-\boldsymbol z^v_{d, 0, k}\right)\left(\boldsymbol z^v_{d+1, v}-\boldsymbol z^v_{d, 0, v}\right)\right|\\
&\qquad\leq C\beta^2n^{-3/2}\log^3(np),
\end{aligned}
\end{equation*}
\end{small}
hold with probability approaching one.

Next we consider $W_d,W_{d+1},W_{d,0}$, Similarly, we define 
\begin{equation*}
    \begin{aligned}
       z^w_{d,0,k}=&p n^{-1}\sum_{1\leq i<j\leq d-1} b_{i,j,k}\U_i^\top \U_j/\sqrt{2\text{tr}(\R^2)} +p n^{-1}\sum_{d+1\leq i<j\leq n} b_{i,j,k}\Y_i^\top \Y_j/\sqrt{2\text{tr}(\R^2)} \\
       &\qquad+p n^{-1}\sum_{i=1}^{d-1}\sum_{j=d+1}^n b_{i,j,k}\U_i^\top \Y_j/\sqrt{2\text{tr}(\R^2)}, \\
       z^w_{d,k}=&z^w_{d,0,k}+p n^{-1}\sum_{i=1}^{d-1}b_{i,d,k}\U_i^\top \Y_d/\sqrt{2\text{tr}(\R^2)}+p n^{-1}\sum_{i=d+1}^n b_{i,d,k}\Y_d^\top \Y_i/\sqrt{2\text{tr}(\R^2)},\\
       z^w_{d+1,k}=&z^w_{d,0,k}+p n^{-1}\sum_{i=1}^{d-1}b_{i,d,k}\U_i^\top \U_d/\sqrt{2\text{tr}(\R^2)}+p n^{-1}\sum_{i=d+1}^n b_{i,d,k}\U_d^\top \Y_i/\sqrt{2\text{tr}(\R^2)},
    \end{aligned}
\end{equation*}
   and let  $z^w_{d,0}=(z^w_{d,0,1},\ldots,z^w_{d,0,n })^\top$ and $z^w_d=(z^w_{d,1},\ldots,z^w_{d,n})^\top$.

By Taylor's expansion, we have
\begin{equation}\label{eq:expansion_U}
\begin{aligned}
&~~W_d-W_{d, 0}= \\
& \sum_{l=\lambda_n}^{n-\lambda_n} \partial_l F_\beta\left(\boldsymbol{z}^w_{d, 0}\right)\left( z^w_{d, l}- z^w_{d, 0, l}\right)+\frac{1}{2} \sum_{l=\lambda_n}^{n-\lambda_n} \sum_{k=\lambda_n}^{n-\lambda_n} \partial_k \partial_l F_\beta\left(\boldsymbol{z}^w_{d, 0}\right)\left( z^w_{d, l}- z^w_{d, 0, l}\right)\left( z^w_{d, k}- z^w_{d, 0, k}\right) \\
& +\frac{1}{6} \sum_{l=\lambda_n}^{n-\lambda_n} \sum_{k=\lambda_n}^{n-\lambda_n} \sum_{v=\lambda_n}^{n-\lambda_n} \partial_v \partial_k \partial_l F_\beta\left(\boldsymbol{z}^w_{d, 0}+\tilde\vartheta\left(\boldsymbol{z}^w_d-\boldsymbol{z}^w_{d, 0}\right)\right)\left( z^w_{d, l}- z^w_{d, 0, l}\right)\left( z^w_{d, k}-\boldsymbol z^w_{d, 0, k}\right)\left(\boldsymbol z^w_{d, v}-\boldsymbol z^w_{d, 0, v}\right),
\end{aligned}
\end{equation}
for some $\tilde\vartheta \in(0,1)$. Again, due to $\E\left(\boldsymbol U_t\right)=\E\left(\boldsymbol Y_t\right)=0$ and $\E\left(\boldsymbol  U_t \boldsymbol U_t^{\top}\right)=\E\left(\boldsymbol Y_t \boldsymbol Y_t^{\top}\right)$, we can verify that $\E\left\{\left( z^w_{d, l}- z^w_{d, 0, l}\right) \mid \mathcal{F}_d\right\}=\E\left\{\left( z^w_{d+1, l}- z^w_{d, 0, l}\right) \mid \mathcal{F}_d\right\}$ and $\E\left\{\left( z^w_{d, l}- z^w_{d, 0, l}\right)^2 \mid \mathcal{F}_d\right\}=\E\left\{\left( z^w_{d+1, l}- z^w_{d, 0, l}\right)^2 \mid \mathcal{F}_d\right)$.

By Lemma A.2 in \cite{chernozhukov2013gaussian}, we have
$$
\left|\sum_{l=\lambda_n}^{n-\lambda_n} \sum_{k=\lambda_n}^{n-\lambda_n} \sum_{v=\lambda_n}^{n-\lambda_n} \partial_v \partial_k \partial_l F_\beta\left(\boldsymbol{z}^w_{d, 0}+\tilde\vartheta\left(\boldsymbol{z}^w_d-\boldsymbol{z}^w_{d, 0}\right)\right)\right| \leq C \beta^2,
$$
for some positive constant $C$. We next consider the term $\E\left(\max_{\lambda_n\leq k\leq n-\lambda_n}\vert z^w_{d,k}-z^w_{d,0,k}\vert\right)$ with $z^w_{d,k}-z^w_{d,0,k}=p n^{-1}\sum_{i=1}^{d-1}b_{i,d,k}\U_i^\top \Y_d/\sqrt{2\text{tr}(\R^2)}+p n^{-1}\sum_{i=d+1}^n b_{i,d,k}\Y_d^\top \Y_i/\sqrt{2\text{tr}(\R^2)}$. Taking expectation on $\{\U_1,\ldots,\U_{d-1},\Y_{d+1},\ldots,\Y_n\}$,
\begin{equation*}
    \begin{aligned}
    \phi_{z,d}^2:=&\max_{\lambda_n\leq k\leq n-\lambda_n} \E\left\{\sum_{i=1}^{d-1}(b_{i,d,k}\U_i^\top \Y_d)^2\right\}+\sum_{i=d+1}^n (b_{i,d,k}\Y_d^\top \Y_i)^2\\
        \leq & n \Y_d^\top \E (\U_1\U_1^\top) \Y_d.
    \end{aligned}
\end{equation*}
and
\begin{equation*}
    \begin{aligned}
        &\left\Vert \max_{\lambda_n\leq k\leq n-\lambda_n}\left( \max_{1\leq i\leq d-1}\U_i^\top \Y_d b_{i,d,k}+\max_{d+1\leq i\leq n} \Y_i^\top \Y_d b_{i,d,k} \right) \right\Vert_{\psi_{\alpha_0/2}}\\
        \leq &\sum_{j=1}^p \left(\left\Vert \max_{\lambda_n\leq k\leq n-\lambda_n} \max_{1\leq i\leq d-1}U_{ij} Y_{dj} b_{i,d,k} \right\Vert_{\psi_{\alpha_0/2}}+\left\Vert \max_{\lambda_n\leq k\leq n-\lambda_n} \max_{d+1\leq i\leq n}Y_{ij} Y_{dj} b_{i,d,k} \right\Vert_{\psi_{\alpha_0/2}}\right)\\
        \leq &\sum_{j=1}^p \left(\vert Y_{dj}\vert\left\Vert  \max_{1\leq i\leq d-1}U_{ij}   \right\Vert_{\psi_{\alpha_0/2}}+\vert Y_{dj}\vert\left\Vert  \max_{d+1\leq i\leq n}Y_{ij}  \right\Vert_{\psi_{\alpha_0/2}}\right)\\ 
        \leq & \zeta_1\sqrt{\log n}\sum_{j=1}^p \vert Y_{dj}\vert, 
    \end{aligned}
\end{equation*}
by the properties of $\psi_{\alpha_0}$ norm.
By Lemma \ref{LemmaE.1central} and Assumption \ref{ass:sum_R4}, we have 
\begin{equation*}
\begin{aligned}
    &\E\left(\max_{\lambda_n\leq k\leq n-\lambda_n}\vert z^w_{d,k}-z^w_{d,0,k}\vert\right)\\
    \lesssim& \E\left[\frac{p n^{-1}}{\sqrt{\tr(\R^2)}}\left\{ \sqrt{\Y_d^\top \E (\U_1\U_1^\top)\Y_d}\sqrt{n}\sqrt{\log n}+\zeta_1\sum_{j=1}^p \vert Y_{dj}\vert \log n\right\}\right]\\
   \leq & \frac{p n^{-1}}{\sqrt{\tr(\R^2)}}\left\{ \sqrt{\E \Y_d^\top \E (\U_1\U_1^\top)\Y_d}\sqrt{n}\sqrt{\log n}+\zeta_1\sum_{j=1}^p \E\vert Y_{dj}\vert \log n\right\}\\
   \lesssim &\frac{p n^{-1}}{\sqrt{\tr(\R^2)}}\left\{ \sqrt{p^{-2}\tr(\R^2)}\sqrt{n}\sqrt{\log n}+\zeta_1^2 p\log n\right\}\\
   \lesssim&n^{-(1/2\wedge\omega_1)}\log n.
   \end{aligned}
\end{equation*}

 Hence,
$$
\begin{aligned}
&\left|\frac{1}{6} \sum_{l,k,v=\lambda_n}^{n-\lambda_n} \partial_v \partial_k \partial_l F_\beta\left(\boldsymbol{z}^w_{d, 0}+\tilde\vartheta\left(\boldsymbol{z}^w_d-\boldsymbol{z}^w_{d, 0}\right)\right)\left(\boldsymbol z^w_{d, l}-\boldsymbol z^w_{d, 0, l}\right)\left(\boldsymbol z^w_{d, k}-\boldsymbol z^w_{d, 0, k}\right)\left(\boldsymbol z^w_{d, v}-\boldsymbol z^w_{d, 0, v}\right)\right|\\
&\leq C\beta^2n^{-(3/2\wedge 3\omega_1)}\log^{3}n,\\
&\left|\frac{1}{6} \sum_{l,k,v=\lambda_n}^{n-\lambda_n}  \partial_v \partial_k \partial_l F_\beta\left(\boldsymbol{z}^w_{d+1, 0}+\tilde\vartheta\left(\boldsymbol{z}^w_{d+1}-\boldsymbol{z}^w_{d, 0}\right)\right)\left(\boldsymbol z^w_{d+1, l}-\boldsymbol z^w_{d, 0, l}\right)\left(\boldsymbol z^w_{d+1, k}-\boldsymbol z^w_{d, 0, k}\right)\left(\boldsymbol z^w_{d+1, v}-\boldsymbol z^w_{d, 0, v}\right)\right|\\
&\leq C\beta^2n^{-(3/2\wedge 3\omega_1)}\log^{3}n,
\end{aligned}
$$
hold with probability approaching one. Consequently we have, with probability one,
$$
\left|\mathrm{E}\left\{f_1\left(W_{d, 0}, V_{d, 0}\right)\left(W_d-W_{d, 0}\right)\right\}-\mathrm{E}\left\{f_2\left(W_{d, 0}, V_{d, 0}\right)\left(W_{d+1}-W_{d, 0}\right)\right\}\right| \leq C \beta^2 n^{-(3 / 2\wedge 3\omega_1)} \log ^3n,
$$
$$
\left|\mathrm{E}\left\{f_2\left(W_{d, 0}, V_{d, 0}\right)\left(V_d-V_{d, 0}\right)\right\}-\mathrm{E}\left\{f_2\left(W_{d, 0}, V_{d, 0}\right)\left(V_{d+1}-V_{d, 0}\right)\right\}\right| \leq C \beta^2 n^{-3 / 2} \log ^3(np).
$$

Similarly, it can be verified that, 
$$
\left|\mathrm{E}\left\{f_{11}\left(W_{d, 0}, V_{d, 0}\right)\left(W_d-W_{d, 0}\right)^2\right\}-\mathrm{E}\left\{f_{22}\left(W_{d, 0}, V_{d, 0}\right)\left(W_{d+1}-W_{d, 0}\right)^2\right\}\right| \leq C \beta^2 n^{-(3 / 2\wedge 3\omega_1)} \log ^3 n,
$$
$$
\left|\mathrm{E}\left\{f_{22}\left(W_{d, 0}, V_{d, 0}\right)\left(V_d-V_{d, 0}\right)^2\right\}-\mathrm{E}\left\{f_{22}\left(W_{d, 0}, V_{d, 0}\right)\left(V_{d+1}-V_{d, 0}\right)^2\right\}\right| \leq C \beta^2 n^{-3 / 2} \log ^3(np),
$$
and
\begin{equation*}
    \begin{aligned}
& \left|\mathrm{E}\left\{f_{12}\left(W_{d, 0}, V_{d, 0}\right)\left(W_d-W_{d, 0}\right)\left(V_d-V_{d, 0}\right)\right\}-\mathrm{E}\left\{f_{12}\left(W_{d, 0}, V_{d, 0}\right)\left(W_{d+1}-W_{d, 0}\right)\left(V_{d+1}-V_{d, 0}\right)\right\}\right| \\
& \quad \leq C \beta^2 n^{-3 / 4-(3/4\wedge 3/2\omega_1)} \log ^3(np).
    \end{aligned}
\end{equation*}

By Equation \eqref{eq:expansion_V} and \eqref{eq:expansion_U}, $\mathrm{E}\left(\left|V_d-V_{d, 0}\right|^3\right)=O\left(n^{-3 / 2} \log ^3(n p)\right)$ and $\mathrm{E}\left(\left|W_d-W_{d, 0}\right|^3\right)=O\left(n^{-(3 / 2\wedge 3\omega_1)} \log ^3 n \right)$. Combining all facts together, we conclude that there exists constant $C$,
\begin{equation*}
\sum_{d=1}^n\left|\mathrm{E}\left\{f\left(W_d, V_d\right)\right\}-\mathrm{E}\left\{f\left(W_{d+1}, V_{d+1}\right)\right\}\right| \leq C \beta^2 \left(n^{-3 / 2} \log ^3 np + n^{-(3 / 2\wedge 3\omega_1)} \log ^3 n\right)\rightarrow 0,
\end{equation*}
as $(n, p) \rightarrow \infty$. The conclusion follows.

\subsection{Proof of Theorem \ref{thm:ind_H1}}

For (i), according to the proof of Theorem \ref{thm:Max-Sum}, under $H_{1,np}$, we have that,
\begin{equation}\label{eq:H1_signal}
    \begin{aligned}
        &~~S_{np}= \\
         &\max_{\lambda_n\leq k\leq n-\lambda_n}\sum_{1\leq i, j\leq n} \upsilon_{i,k}\upsilon_{j,k}s_{i}s_{j}R_i^{-1}R_j^{-1}(\frac{n-\tau}{n}\boldsymbol\delta)^\top \D^{-1/2}(\mathbf I_p-\U_i\U_i^\top)(\mathbf I_p-\U_j\U_j^\top)\D^{-1/2}(\frac{n-\tau}{n}\boldsymbol\delta)\\
        &+\max_{\lambda_n\leq k\leq n-\lambda_n}\sum_{l\in \mathcal{A}}\sum_{1\leq i\not=j\leq n}\upsilon_{i,k}\upsilon_{j,k} \U_{i,l} \U_{j,l}+\max_{\lambda_n\leq k\leq n-\lambda_n}\sum_{l\in \mathcal{A}^c}\sum_{1\leq i\not=j\leq n}\upsilon_{i,k}\upsilon_{j,k} \U_{i,l} \U_{j,l}+o_p(1).
    \end{aligned}
\end{equation}

 For the first part in Equation \eqref{eq:H1_signal}, denote $s_i=-1$, if $i\leq\tau$ and $s_i=1$, if $i>\tau$, $i=1,\ldots,n$. Taking the same procedure as in the proof of Lemma A.2 in \cite{feng2016multivariate}, we have
\begin{equation*}
    \begin{aligned}
        &\max_{\lambda_n\leq k\leq n-\lambda_n}\sum_{1\leq i, j\leq n} \upsilon_{i,k}\upsilon_{j,k}s_{i}s_{j}R_i^{-1}R_j^{-1}(\frac{n-\tau}{n}\boldsymbol\delta)^\top \D^{-1/2}(\mathbf I_p-\U_i\U_i^\top)(\mathbf I_p-\U_j\U_j^\top)\D^{-1/2}(\frac{n-\tau}{n}\boldsymbol\delta)\\
        =&\max_{\lambda_n\leq k\leq n-\lambda_n}\frac{k^2(n-k)^2p}{n^3\sqrt{2\tr(\R^2)}} (\frac{1}{k}\sum_{i=1}^k s_{i}R_i^{-1}-\frac{1}{n-k}\sum_{i=k+1}^n s_{i}R_i^{-1})^2\left\Vert\frac{n-\tau}{n}\D^{-1/2}\boldsymbol\delta\right\Vert^2(1+o_p(1))\\
        =&\max_{\lambda_n\leq k\leq n-\lambda_n}\frac{k^2(n-k)^2p}{n^3\sqrt{2\tr(\R^2)}} \left[\frac{1}{k}\sum_{i=1}^k s_{i}\{R_i^{-1}-\E 
 (R_i^{-1})\}-\frac{1}{n-k}\sum_{i=k+1}^n s_{i}\{R_i^{-1}-\E (R_i^{-1})\}\right.\\
 &\qquad\qquad\qquad\qquad\left.+\frac{1}{k}\sum_{i=1}^k s_{i}\E 
( R_i^{-1})-\frac{1}{n-k}\sum_{i=k+1}^n s_{i}\E (R_i^{-1})\right]^2\left\Vert\frac{n-\tau}{n}\D^{-1/2}\boldsymbol\delta\right\Vert^2(1+o_p(1)).\\
    \end{aligned}
\end{equation*}

We consider the term separately,
\begin{equation*}
    \begin{aligned}
        &\max_{\lambda_n\leq k\leq n-\lambda_n}k(n-k)\left\vert\frac{1}{k}\sum_{i=1}^k s_{i}\E 
( R_i^{-1})-\frac{1}{n-k}\sum_{i=k+1}^n s_{i}\E (R_i^{-1})\right\vert\\
 =&\max_{\lambda_n\leq k\leq n-\lambda_n}\left\vert k(n-k)\zeta_1+k(n-k)\zeta_1-2(n-k)(k-\tau+1)\zeta_1\right\vert\lesssim n^2\zeta_1,
    \end{aligned}
\end{equation*}
and
\begin{equation*}
    \begin{aligned}
                &\max_{\lambda_n\leq k\leq n-\lambda_n}\frac{k^2(n-k)^2p}{n^3\sqrt{2\tr(\R^2)}} \left[\frac{1}{k}\sum_{i=1}^k s_{i}\{R_i^{-1}-\E 
 (R_i^{-1})\}-\frac{1}{n-k}\sum_{i=k+1}^n s_{i}\{R_i^{-1}-\E (R_i^{-1})\}\right]^2\\
 =&\frac{p}{n^3\sqrt{2\tr(\R^2)}}\left[\max_{\lambda_n\leq k\leq n-\lambda_n} k(n-k)\left\vert\frac{1}{k}\sum_{i=1}^k s_{i}\{R_i^{-1}-\E 
 (R_i^{-1})\}-\frac{1}{n-k}\sum_{i=k+1}^n s_{i}\{R_i^{-1}-\E (R_i^{-1})\}\right\vert\right]^2.\\
    \end{aligned}
\end{equation*}

To bounding the first term in Equation \eqref{eq:H1_signal}, we define
\begin{equation*}
    \begin{aligned}
        \sigma^2_R:=&\max_{\lambda_n\leq k\leq n-\lambda_n}\left(\sum_{i=1}^k \E [(n-k)^2 s_i^2\{R_i^{-1}-\E (R_i^{-1})\}^2] + \sum_{i=k+1}^n \E [k^2 s_i^2\{R_i^{-1}-\E (R_i^{-1})\}^2]\right)\\
        \lesssim&n^3 p^{-1},
    \end{aligned}
\end{equation*}
and
\begin{equation*}
    \begin{aligned}
   M_R:= &\left\Vert \max_{\lambda_n\leq k\leq n-\lambda_n}\max[\max_{1\leq i\leq k}\vert (n-k)\{R_i-\E (R_i^{-1})\} \vert  ,\max_{k+1\leq i\leq n}\vert k\{R_i-\E (R_i^{-1})\}]\right\Vert_{\psi_{\alpha_0}}\\
        \leq &n\left\Vert \max_{1\leq k\leq n}\vert R_i-\E R_i^{-1} \vert  \right\Vert_{\psi_{\alpha_0}}\lesssim n\log n.
    \end{aligned}
\end{equation*}

By Lemma \ref{LemmaE.1central}, we have

\begin{equation*}
\begin{aligned}
  \E&\left[\max_{\lambda_n\leq k\leq n-\lambda_n}\frac{k(n-k)p^{1/2}}{n^{3/2}(2\tr(\R^2))^{1/4}} \left\vert\frac{1}{k}\sum_{i=1}^k s_{i}\{R_i^{-1}-\E 
 (R_i^{-1})\}-\frac{1}{n-k}\sum_{i=k+1}^n s_{i}\{R_i^{-1}-\E (R_i^{-1})\}\right\vert\right]\\
 \lesssim &\frac{p^{1/2}}{n^{3/2}\{2\tr(\R^2)\}^{1/4}}(\sigma_R\sqrt{\log n}+M_R\log n)\\
 \lesssim &\frac{p^{1/2}}{n^{3/2}\{2\tr(\R^2)\}^{1/4}}(n^{3/2}p^{-1/2}\sqrt{\log n}+n\log^2 n)\\
 \lesssim& p^{-1/4}\log^{1/2}n+n^{-2}\log^2 n.
\end{aligned}
\end{equation*}

Thus for Equation \eqref{eq:H1_signal}, we have 
\begin{equation*}
    \begin{aligned}
        S_{np} = &     \max_{\lambda_n\leq k\leq n-\lambda_n}\sum_{1\leq i\not=j\leq n}\upsilon_{i,k}\upsilon_{j,k} \U_i^\top \U_j+\tilde{\boldsymbol\Delta}_S+o_p(1)\\
        = & \max_{\lambda_n\leq k\leq n-\lambda_n}\sum_{l\in \mathcal{A}^c}\sum_{1\leq i\not=j\leq n}\upsilon_{i,k}\upsilon_{j,k} \U_{i,l} \U_{j,l}+\max_{1\leq k\leq n}\sum_{l\in \mathcal{A}}\sum_{1\leq i\not=j\leq n}\upsilon_{i,k}\upsilon_{j,k} \U_{i,l} \U_{j,l}+\tilde{\boldsymbol\Delta}_S+o_p(1).\\
    \end{aligned}
\end{equation*}
where 
\begin{equation*}
    \begin{aligned}
        \tilde{\boldsymbol\Delta}_S\lesssim \frac{n\Vert\D^{-1/2}\boldsymbol\delta\Vert^2}{\sqrt{2\tr(\R^2)}}\lesssim \frac{n\Vert\boldsymbol\delta\Vert^2}{\sqrt{2\tr(\R^2)}}=o(1),
    \end{aligned}
\end{equation*}
by Assumption \ref{ass:max3}. We next consider the second term,
\begin{equation*}
    \begin{aligned}
        &\max_{\lambda_n\leq k\leq n-\lambda_n}\sum_{l\in \mathcal{A}}\sum_{1\leq i\not=j\leq n}\upsilon_{i,k}\upsilon_{j,k} U_{i,l} U_{j,l}\\
        \leq&\vert\mathcal{A}\vert\max_{\lambda_n\leq k\leq n-\lambda_n}\max_{l\in \mathcal{A}}\left(\left\vert\sum_{i=1}^n\upsilon_{i,k} U_{i,l} \right\vert^2+\sum_{i=1}^n\upsilon_{i,k}^2U_{i,l}^2\right)\\
        \leq&\left(\vert\mathcal{A}\vert^{1/2}\max_{\lambda_n\leq k\leq n-\lambda_n}\max_{l\in \mathcal{A}}\left\vert\sum_{i=1}^n\upsilon_{i,k} U_{i,l} \right\vert\right)^2 + \vert\mathcal{A}\vert\max_{\lambda_n\leq k\leq n-\lambda_n}\max_{l\in \mathcal{A}}\sum_{i=1}^n\upsilon_{i,k}^2U_{i,l}^2.
    \end{aligned}
\end{equation*}

To bounding the above terms, we define
\begin{equation*}
    \begin{aligned}
        \sigma^2_{\upsilon^1}:=&\max_{\lambda_n\leq k\leq n-\lambda_n}\max_{l\in \mathcal{A}}\sum_{i=1}^n \upsilon_{ik}^2\E (U_{il}^2)\leq \frac{p}{\sqrt{\tr(\R^2)}}\left\{ \frac{1}{p}+O(p^{-1-\eta_0/2}) \right\}\lesssim\frac{1}{\sqrt{\tr(\R^2)}},\\
        \sigma^2_{\upsilon^2}:=&\max_{\lambda_n\leq k\leq n-\lambda_n}\max_{l\in \mathcal{A}}\sum_{i=1}^n \upsilon_{ik}^4\E (U_{il}^4)\lesssim \frac{p^2\zeta_1^4}{n\tr(\R^2)} =\frac{1}{n\tr(\R^2)} ,
    \end{aligned}
\end{equation*}
and
\begin{equation*}
    \begin{aligned}
        M_{\upsilon^1}:= &\left\Vert \max_{\lambda_n\leq k\leq n-\lambda_n}\max_{l \in \mathcal{A}}\max_{1\leq i\leq n}\vert \upsilon_{i,k}U_{i,l}\vert \right\Vert_{\psi_{\alpha_0}}\\
        \leq &\left\{ \frac{p}{n\sqrt{\tr(\R^2)}}
 \right\}^{1/2}\left\Vert \max_{l \in \mathcal{A}}\max_{1\leq i\leq n}\vert U_{i,l}\vert  \right\Vert_{\psi_{\alpha_0}}\lesssim\frac{\log(n\vert\mathcal{A}\vert)}{n^{1/2}\tr^{1/4}(\R^2)},\\
   M_{\upsilon^2}:= &\left\Vert \max_{\lambda_n\leq k\leq n-\lambda_n}\max_{l \in \mathcal{A}}\max_{1\leq i\leq n}\vert \upsilon_{i,k}U_{i,l}\vert^2  \right\Vert_{\psi_{\alpha_0/2}}\\
   \leq &\frac{p}{n\sqrt{\tr(\R^2)}}
 \left\Vert \max_{l \in \mathcal{A}}\max_{1\leq i\leq n}\vert U_{i,l}\vert  \right\Vert_{\psi_{\alpha_0}}^2\lesssim \frac{\log^2(n\vert\mathcal{A}\vert)}{n\sqrt{\tr(\R^2)}}.\\
    \end{aligned}
\end{equation*}

By Lemma \ref{LemmaE.1central}, we have
\begin{equation*}
    \begin{aligned}
       \vert\mathcal{A}\vert^{1/2} \E \max_{\lambda_n\leq k\leq n-\lambda_n}\max_{l\in \mathcal{A}}\left\vert\sum_{i=1}^n\upsilon_{i,k} U_{i,l} \right\vert \lesssim &\vert\mathcal{A}\vert^{1/2} \left\{ \frac{\log^{1/2}(n\vert\mathcal{A}\vert )}{\tr^{1/4}(\R^2)}+\frac{\log^2(n\vert\mathcal{A}\vert)}{n^{1/2}\tr^{1/4}(\R^2)}\right\}=o(1),\\
       \vert\mathcal{A}\vert\E\max_{\lambda_n\leq k\leq n-\lambda_n}\max_{l\in \mathcal{A}}\sum_{i=1}^n\upsilon_{i,k}^2U_{i,l}^2\lesssim & \vert\mathcal{A}\vert  \left\{ \frac{\log^{1/2}(n\vert\mathcal{A}\vert )}{n^{1/2}\sqrt{\tr(\R^2)}}+\frac{\log^3(n\vert\mathcal{A}\vert)}{n^{1/2}\tr^{1/4}(\R^2)}\right\}=o(1).
    \end{aligned}
\end{equation*}

By Markov inequality, we have, $\max_{1\leq k\leq n}\sum_{l\in \mathcal{A}}\sum_{1\leq i\not=j\leq n}\upsilon_{i,k}\upsilon_{j,k} U_{i,l} U_{j,l}=o_p(1)$. Thus, the Equation \eqref{eq:H1_signal} can be written as 
\begin{equation*}
    \begin{aligned}
        S_{np} = & \max_{\lambda_n\leq k\leq n-\lambda_n}\sum_{l\in \mathcal{A}^c}\sum_{1\leq i\not=j\leq n}\upsilon_{i,k}\upsilon_{j,k} U_{i,l} U_{j,l}+o_p(1).\\
    \end{aligned}
\end{equation*}

We rewrite $M_{np}$ as, 
\begin{equation*}
    M_{np}=\max_{\lambda_n\leq k\leq n-\lambda_n}\left(\max_{j\in \mathcal{A}} \vert C_{0,j}(k)\vert + \max_{j\in \mathcal{A}^c}\vert C_{0,j}(k)\vert\right).
\end{equation*}

From the $H_{1,np}$, the Bahadur representation for $\hat\bth_{1:k}$ and $\hat\bth_{k+1:n}$ still holds, by taking the same procedure of Lemma 1 in \cite{liu+feng+wang+2024} with minor modification. It is suffices to show the conclusion holds for $\{\U_i\}_{i=1}^n$ follows Gaussian data sequences.  According to Theorem \ref{thm:ind_H0}, we have known that $\max_{\lambda_n\leq k\leq n-\lambda_n}\sum_{l\in \mathcal{A}^c}\sum_{1\leq i\not=j\leq n}\upsilon_{i,k}\upsilon_{j,k} U_{i,l} U_{j,l}$ is asymptotically independent of $\max_{\lambda_n\leq k\leq n-\lambda_n}\max_{j\in \mathcal{A}^c}\vert C_{0,j}(k)\vert$. Hence it is suffices to show that, $\max_{\lambda_n\leq k\leq n-\lambda_n}\sum_{l\in \mathcal{A}^c}\sum_{1\leq i\not=j\leq n}\upsilon_{i,k}\upsilon_{j,k} U_{i,l} U_{j,l}$ is asymptotically independent of $U_{i,l},l\in \mathcal{A}$. 

Without loss of generality, we assume $\mathcal{A}=\{ j_1,j_2,\ldots,j_d \}$. For each $i=1,2,\ldots,n$, let $\U_{i,(1)}=(U_{i,j_1},U_{i,j_2},\ldots,U_{i,j_d})$ and $\U_{i,(2)}=(U_{i,j_{d+1}},U_{i,j_{d+2}},\ldots,U_{i,j_p})$ and $R_{kl}=\cov(\U_{i,(k)},\U_{i,(l)})$ for $k\in\{1,2\}$. By Lemma \ref{lemma_normaldecom}, $\U_{i,(2)}$ can be decomposed as $\U_{i,(2)}=$ $\V_i+\T_i$, where $\V_i:=\U_{i,(2)}-\R_{21} \R_{11}^{-1} \U_{i,(1)}$ and $\T_i:=\R_{21} \R_{11}^{-1} \U_{i,(1)}$ satisfying that $\V_i \sim N\left(0, \R_{22}-\right.$ $\left.\R_{21} \R_{11}^{-1} \R_{12}\right), \T_i \sim N\left(0, \R_{21} \R_{11}^{-1} \R_{12}\right)$ and
\begin{equation}\label{eq:VU_ind}
   \V_i\text{ and } \U_{i,(1)}\text{ are independent.}
\end{equation}

We have, 
\begin{equation*}
    \begin{aligned}
        &\left\vert\max_{\lambda_n\leq k\leq n-\lambda_n}\sum_{1\leq i\not=j\leq n}\upsilon_{i,k}\upsilon_{j,k} \U_{i,(2)}^\top \U_{j,(2)} - \max_{\lambda_n\leq k\leq n-\lambda_n}\sum_{1\leq i\not=j\leq n}\upsilon_{i,k}\upsilon_{j,k} \V_i^\top \V_j\right\vert\\
        \leq&2\left\vert\max_{\lambda_n\leq k\leq n-\lambda_n}\sum_{1\leq i\not=j\leq n}\upsilon_{i,k}\upsilon_{j,k} \V_i^\top\T_j\right\vert+\left\vert\max_{\lambda_n\leq k\leq n-\lambda_n}\sum_{1\leq i\not=j\leq n}\upsilon_{i,k}\upsilon_{j,k}\T_i^\top\T_j\right\vert.
    \end{aligned}
\end{equation*}

By using arguments similar to those in the proof of Lemma \ref{lemma_zetapd}, we have
\begin{equation*}
    \begin{aligned}
        \pr(\max_{\lambda_n\leq k\leq n-\lambda_n}\sum_{1\leq i\not=j\leq n}\upsilon_{i,k}\upsilon_{j,k} \V_i^\top\T_j\geq \varsigma)&\leq \log n \exp(-C \varsigma p^{1/2}/d^{1/2})\rightarrow 0,\\
        \pr(\max_{\lambda_n\leq k\leq n-\lambda_n}\sum_{1\leq i\not=j\leq n}\upsilon_{i,k}\upsilon_{j,k} \T_i^\top\T_j\geq \varsigma)&\leq \log n \exp(-C \varsigma p^{1/2}/d^{1/2})\rightarrow 0,
    \end{aligned}
\end{equation*}
since $d=\vert \mathcal{A}\vert=o(p/(\log\log p)^2)$ and $n\lesssim p^{1/(1-2\omega_1)}$. Consequently, we conclude that, 
\begin{equation*}
    \max_{\lambda_n\leq k\leq n-\lambda_n}\sum_{l\in \mathcal{A}^c}\sum_{1\leq i\not=j\leq n}\upsilon_{i,k}\upsilon_{j,k} U_{i,l} U_{j,l} = \max_{\lambda_n\leq k\leq n-\lambda_n}\sum_{1\leq i\not=j\leq n}\upsilon_{i,k}\upsilon_{j,k} \V_i^\top \V_j +o_p(1).
\end{equation*}

By Lemma \ref{lemmaS.10_feng} and Equation \eqref{eq:VU_ind}, we have $\max_{\lambda_n\leq k\leq n-\lambda_n}\sum_{l\in \mathcal{A}^c}\sum_{1\leq i\not=j\leq n}\upsilon_{i,k}\upsilon_{j,k} U_{i,l} U_{j,l} $ is asymptotically independent of $\U_{i,(1)}$. Hence Theorem \ref{thm:ind_H1}-(i) follows. The proof of \ref{thm:ind_H1}-(ii) is similar, and thus is omitted.

\subsection{Proof of Proposition \ref{prop:max-Linf}}

Observe that 
\begin{equation*}
    \begin{aligned}
        M_{n,p}=&\max_{\lambda_n\leq k\leq n-\lambda_n}\frac{k}{n}\left(1-\frac{k}{n}\right)\sqrt{n} \|\hat{\D}^{-1/2}(\hat{\bth}_{1:k}-\hat{\bth}_{(k+1):n})\|_\infty\\
        \geq & \frac{\tau}{n}\left(1-\frac{\tau}{n}\right)\sqrt{n} \|\hat{\D}^{-1/2}(\hat{\bth}_{1:\tau}-\hat{\bth}_{(\tau+1):n})\|_\infty\\
        =&\frac{\tau(n-\tau)}{n^{3/2}}\|\D^{-1/2}\boldsymbol\delta\|_\infty+O_p(\sqrt{\log p}),
    \end{aligned}
\end{equation*}
where the last equality follows from the assumptions that $\tau=[cn]$ for some $c\in (0,1)$ and Assumptions~\ref{ass:max1}--\ref{ass:cor}.
For a given significance level $\alpha$, the critical value of the test based on $M_{n,p}$ is
\begin{equation*}
    \begin{aligned}
        c_{M,\alpha}=p^{-1/2}\hat{\zeta_1}^{-1}\sqrt{[-\log\{-\log(1-\alpha)\}+\log(2p)]/2}\asymp \sqrt{\log p}\,.
    \end{aligned}
\end{equation*}
Therefore, under Assumption~\ref{ass:max3} and the condition $\|\boldsymbol{\delta}\|_{\infty} \geq C \sqrt{\log p / n}$ for some constant $C > 0$, it follows that with probability tending to one, $M_{n,p} \geq c_{M,\alpha}$. This establishes the consistency of the test based on the statistic $M_{n,p}$.

The proof of Proposition~\ref{prop:max-Linf} (ii) proceeds similarly and is thus omitted.

\subsection{Proof of Proposition \ref{prop:max-L2}}
 Suppose $Z_1,\ldots,Z_n$ are samples from $\mathrm{Bernoulli}(\kappa)$ with $\kappa=\tau/n$ and we have $\sum_{i=1}^n Z_i=\tau$. Suppose $\tilde{\X}_{i1}=\bth_0+\boldsymbol\epsilon_{i}$  and $\tilde{\X}_{i2}=\bth_0+\boldsymbol\delta+\boldsymbol\epsilon_{i}$ where $\boldsymbol\epsilon_{i}$ are i.i.d.~from the model \eqref{modelx}. Denote $\Y_i=Z_i\tilde\X_{i1}+(1-Z_i)\tilde\X_{i2}=\bth_0+\boldsymbol\epsilon_i+(1-Z_i)\boldsymbol\delta$, then $\E (\Y_i)=\bth_0+(1-\kappa)\boldsymbol\delta$ and $\var(\Y_i)= \var(\boldsymbol\epsilon_i)+\kappa(1-\kappa)\boldsymbol\delta\boldsymbol\delta^\top$. Thus, $\hat\bth_{1:n}$ is an estimator based on sample $\Y_1,\ldots,\Y_n$. 
 
 Given $\D$, $\hat\bth_{1:n}$ is an M-estimator and $L(\boldsymbol\beta)=\Vert \D^{-1/2}(X_i-\boldsymbol\beta)\Vert$ is strictly convex in $\boldsymbol\beta$. Let $\tilde{\D}=\diag\{\tilde{d}_1^2,\ldots,\tilde{d}_p^2\}$ and $\bth_\kappa$ satisfy $\E \{U(\tilde{\D}^{-1/2}(\Y_i-\bth_\kappa))\}=\mathbf{0}$ and $\diag\{\E\{U(\tilde{\D}^{-1/2}(\Y_i-\bth_\kappa))U(\tilde{\D}^{-1/2}(\Y_i-\bth_\kappa))^\top\}=p^{-1}\mathbf I_p$. We first consider the case of $\tau =n/2$. By symmetry, $\bth_\kappa=\bth_0+\boldsymbol\delta/2$ and $\tilde{d}_i^2/\tilde{d}_j^2\asymp (d_i^2+\delta_i^2)/(d_j^2+\delta_j^2)$. From the similar procedure as in the proof of Lemma A.3 in \cite{feng2016multivariate}, we have, $\Vert \tilde{\D}^{-1/2}(\hat{\bth}_{1:n}-\bth_\kappa)\Vert=O_p(p^{1/2}n^{-1/2})$, where the term is derived by dominated convergence theorem,
\begin{equation*}
    \begin{aligned}
        \E\left\{\frac{1}{\Vert\tilde{\D}^{-1/2}(\boldsymbol\epsilon_i+\boldsymbol\delta/2)\Vert}\right\}\geq& \E\left\{\frac{1}{\Vert\tilde{\D}^{-1/2}\D^{1/2}\Vert_F\Vert\D^{-1/2}\boldsymbol\epsilon_i\Vert+\Vert\tilde{\D}^{-1/2}\boldsymbol\delta/2\Vert}\right\}\\
        &\rightarrow\E\left\{\frac{1}{\Vert\tilde{\D}^{-1/2}\boldsymbol\delta/2\Vert}\right\}\gtrsim p^{-1/2}.
    \end{aligned}
\end{equation*}
 as $\Vert\boldsymbol\delta\Vert\rightarrow\infty$. For $i,j\in \{1,\ldots,\tau\}$, by $\Vert\boldsymbol\delta\Vert^{-1}\Vert\boldsymbol\delta\Vert_\infty=o(p^{-1/2}n^{1/2})$,
\begin{equation}\label{eq:power_L2_ijtau}
    \begin{aligned}
        1\geq\hat{\U}_i^\top \hat{\U}_j\geq&\left\{\frac{\Vert\boldsymbol\delta\Vert^{-1}\hat{\D}^{-1/2}(\X_i-\bth_0)+\Vert\boldsymbol\delta\Vert^{-1}\hat{\D}^{-1/2}\tilde{\D}^{1/2}\tilde{\D}^{-1/2}(\hat{\bth}_{1:n}-\bth_{\kappa})-\Vert\boldsymbol\delta\Vert^{-1}\hat{\D}^{-1/2}\boldsymbol\delta/2}{ \Vert\boldsymbol\delta\Vert^{-1}\Vert\hat{\D}^{-1/2}(\X_i-\bth_0)\Vert+\Vert\boldsymbol\delta\Vert^{-1}\Vert\hat{\D}^{-1/2}\tilde{\D}^{1/2}\tilde{\D}^{-1/2}(\hat{\bth}_{1:n}-\bth_\kappa)+\hat{\D}^{-1/2}\boldsymbol\delta/2\Vert}\right\}^\top\cdot\\
        &\quad\frac{\Vert\boldsymbol\delta\Vert^{-1}\hat{\D}^{-1/2}(\X_j-\bth_0)+\Vert\boldsymbol\delta\Vert^{-1}\hat{\D}^{-1/2}\tilde{\D}^{1/2}\tilde{\D}^{-1/2}(\hat{\bth}_{1:n}-\bth_{\kappa})-\Vert\boldsymbol\delta\Vert^{-1}\hat{\D}^{-1/2}\boldsymbol\delta/2}{ \Vert\boldsymbol\delta\Vert^{-1}\Vert\hat{\D}^{-1/2}(\X_j-\bth_0)\Vert+\Vert\boldsymbol\delta\Vert^{-1}\Vert\hat{\D}^{-1/2}\tilde{\D}^{1/2}\tilde{\D}^{-1/2}(\hat{\bth}_{1:n}-\bth_\kappa)+\hat{\D}^{-1/2}\boldsymbol\delta/2\Vert}\\
    \rightarrow &1
    \end{aligned}
\end{equation}
as $\Vert\boldsymbol\delta\Vert\rightarrow\infty$. Take the same procedure, we have, for all $i,j\in \{1,\ldots,n\}$, $\hat{\U}_i^\top \hat{\U}_j\rightarrow 1$ as $\Vert\boldsymbol\delta\Vert\rightarrow\infty$.

Thus, as $\Vert\boldsymbol\delta\Vert\rightarrow\infty$,
\begin{equation*}
    \begin{aligned}
        \frac{1}{\sqrt{2\tr(\R^2)}}S_{n,p}=&\max_{\lambda_n\le k\le n-\lambda_n}\left\{\tilde{\C}_{0}(k)^\top \tilde{\C}_{0}(k)-\frac{k(n-k)p}{n^2}\right\}\\
        \geq&\frac{1}{\sqrt{2\tr(\R^2)}}\left\{\tilde{\C}_{0}(\tau)^\top \tilde{\C}_{0}(\tau)-\frac{\tau(n-\tau)p}{n^2}\right\}\\
        \asymp& \frac{1}{\sqrt{2\tr(\R^2)}}\left\{\frac{4\tau^2(n-\tau)^2p}{n^3}-\frac{\tau(n-\tau)p}{n^2}\right\}\rightarrow\infty.
    \end{aligned}
\end{equation*}

{As $\tau\not= n/2$, W.L.O.G, $\tau<n/2$, we first show that, $\Vert \tilde{\D}^{-1/2}(\bth_\tau-\bth_0-\boldsymbol\delta)\Vert\rightarrow 0$  and $d_l^2/d_1^2\asymp (\delta_l^2+d_l^2)/(\delta_1^2+d_1^2)$ hold, for $i=1,\ldots,p$, as $\Vert\boldsymbol\delta\Vert\rightarrow\infty$.}
For $\bth_\kappa$, we consider the equation $\E \{U(\tilde{\D}^{-1/2}(\Y_i-\bth_\kappa))\}=\mathbf{0}$, i.e. ,
\begin{equation}\label{eq:bth_kappa}
    \kappa\E \frac{\tilde{\D}^{-1/2}(\boldsymbol\epsilon_i+\bth_0-\bth_\kappa)}{\Vert \tilde{\D}^{-1/2}(\boldsymbol\epsilon_i+\bth_0-\bth_\kappa)\Vert}+(1-\kappa)\E \frac{\tilde{\D}^{-1/2}(\boldsymbol\epsilon_i+\bth_0+\boldsymbol\delta-\bth_\kappa)}{\Vert \tilde{\D}^{-1/2}(\boldsymbol\epsilon_i+\bth_0-\bth_\kappa)\Vert}=\mathbf{0}.\end{equation}

Let $\bth_{\kappa,i}=\bth_{0,i}+c_i \delta_i$, $i=1,\ldots,p$ and $C=\diag\{c_1,\ldots,c_p\}$, $0\leq c_i\leq 1$. Then Equation \eqref{eq:bth_kappa} can be rewritten as 
\begin{equation*}
        \kappa\E \frac{\tilde{\D}^{-1/2}(\boldsymbol\epsilon_i+C\boldsymbol\delta)}{\Vert \tilde{\D}^{-1/2}(\boldsymbol\epsilon_i+\bth_0-\bth_\kappa)\Vert}+(1-\kappa)\E \frac{\tilde{\D}^{-1/2}(\boldsymbol\epsilon_i+(\mathbf I_p-C)\boldsymbol\delta)}{\Vert \tilde{\D}^{-1/2}(\boldsymbol\epsilon_i+\bth_0-\bth_\kappa)\Vert}=\mathbf{0},
\end{equation*}
if $\Vert \tilde{\D}^{-1/2}C\boldsymbol\delta\Vert\rightarrow\infty$ and $\Vert \tilde{\D}^{-1/2}(\mathbf I_p-C)\boldsymbol\delta\Vert\rightarrow\infty$ as $\Vert\boldsymbol\delta\Vert\rightarrow\infty$, the Equation  derived by Equation \eqref{eq:bth_kappa} holds,
\begin{equation*}
    \kappa \frac{\tilde{\D}^{-1/2}C\boldsymbol\delta}{\Vert \tilde{\D}^{-1/2}C\boldsymbol\delta\Vert}+(1-\kappa)\frac{\tilde{\D}^{-1/2}(\mathbf I_p-C)\boldsymbol\delta}{\Vert \tilde{\D}^{-1/2}(\mathbf I_p-C)\boldsymbol\delta\Vert}=0.
\end{equation*}
However, it does not holds for any $\boldsymbol\delta$ as $\Vert\boldsymbol\delta\Vert> 0$. It indicates that $\Vert \tilde{\D}^{-1/2}C\boldsymbol\delta\Vert<\infty$, $\Vert \tilde{\D}^{-1/2}(\mathbf I_p-C)\boldsymbol\delta\Vert\rightarrow\infty$ or $\Vert \tilde{\D}^{-1/2}C\boldsymbol\delta\Vert\rightarrow\infty$ and $\Vert \tilde{\D}^{-1/2}(\mathbf I_p-C)\boldsymbol\delta\Vert<\infty$ holds. If $\Vert \tilde{\D}^{-1/2}C\boldsymbol\delta\Vert<\infty$, $\Vert \tilde{\D}^{-1/2}(\mathbf I_p-C)\boldsymbol\delta\Vert\rightarrow\infty$ hold, we see that 
$$(1-\kappa)^2= \kappa^2\left\{\E \frac{\tilde{\D}^{-1/2}(\boldsymbol\epsilon_i+C\boldsymbol\delta)}{\Vert \tilde{\D}^{-1/2}(\boldsymbol\epsilon_i+\bth_0-\bth_\kappa)\Vert}\right\}^\top\left\{\E \frac{\tilde{\D}^{-1/2}(\boldsymbol\epsilon_i+C\boldsymbol\delta)}{\Vert \tilde{\D}^{-1/2}(\boldsymbol\epsilon_i+\bth_0-\bth_\kappa)\Vert}\right\}\leq \kappa^2,$$
contradicts to $\kappa<1/2$. Thus we have, $\Vert \tilde{\D}^{-1/2}C\boldsymbol\delta\Vert\rightarrow\infty$ and $\Vert \tilde{\D}^{-1/2}(\mathbf I_p-C)\boldsymbol\delta\Vert<\infty$, i.e. $\Vert\tilde{\D}^{-1/2}(\bth_\kappa-\bth_0)\Vert\rightarrow\infty$ and $\Vert\tilde{\D}^{-1/2}(\bth_\kappa-\bth_0-\boldsymbol\delta)\Vert<\infty$.

For $\tilde{d}_l^2$, we consider the equation $\diag\{\E \{U(\tilde{\D}^{-1/2}(\Y_i-\bth_\kappa))U(\tilde{\D}^{-1/2}(\Y_i-\bth_\kappa))^\top \}\}=p^{-1}\mathbf I_p$, i.e. ,
\begin{equation*}
    \begin{aligned}
        \kappa \E \frac{(\boldsymbol\epsilon_{il}+\bth_{0,l}-\bth_{\kappa,l})^2/\tilde{d}_l^2}{\Vert \tilde{\D}^{-1/2}(\boldsymbol\epsilon_i+\bth_0-\bth_\kappa)\Vert^2}+(1-\kappa) \E \frac{(\boldsymbol\epsilon_{il}+\bth_{0,l}+\delta_l-\bth_{\kappa,l})^2/\tilde{d}_l^2}{\Vert \tilde{\D}^{-1/2}(\boldsymbol\epsilon_i+\bth_0+\boldsymbol\delta-\bth_\kappa)\Vert^2}=\frac{1}{p}.
    \end{aligned}
\end{equation*}

 Taking same discussions, we have, $\tilde{d}_l^2/\tilde{d}_1^2\asymp (\delta_l^2+d_l^2)/(\delta_1^2+d_1^2)$ and $\Vert \tilde{\D}^{-1/2}(\hat{\bth}_{1:n}-\bth_\kappa)\Vert=O_p(p^{1/2}n^{-1/2})$, where the term is derived by dominated convergence theorem,
\begin{equation*}
    \begin{aligned}
        \E\left\{\frac{1}{\Vert\tilde{\D}^{-1/2}(\Y_i-\bth_{\kappa})\Vert}\right\}\geq &\kappa\E\left\{\frac{1}{\Vert \tilde{\D}^{-1/2}(\boldsymbol\epsilon_i+\bth_0-\bth_\kappa)\Vert}\right\}\\
        \geq& \E\left\{\frac{1}{\Vert\tilde{\D}^{-1/2}\D^{1/2}\Vert_F\Vert\D^{-1/2}\boldsymbol\epsilon_i\Vert+\Vert\tilde{\D}^{-1/2}\boldsymbol\delta\Vert+\Vert\tilde{\D}^{-1/2}(\bth_{\kappa}-\bth_0-\boldsymbol\delta)\Vert}\right\}\\
        &\rightarrow\E\left\{\frac{1}{\Vert\tilde{\D}^{-1/2}\boldsymbol\delta\Vert}\right\}\gtrsim p^{-1/2}.
    \end{aligned}
\end{equation*}
 as $\Vert\boldsymbol\delta\Vert\rightarrow\infty$.

We next consider the term $\hat{\U}_i^\top \hat{\U}_j$. For $i,j\in \{1,\ldots,\tau\}$, similar with Equation \eqref{eq:power_L2_ijtau}, by $\Vert\boldsymbol\delta\Vert^{-1}\Vert\boldsymbol\delta\Vert_{\infty}=o(p^{1/2}n^{-1/2})$, we have, $\hat{\U}_i^\top\hat{\U}_j\rightarrow 1$ as $\Vert\boldsymbol\delta\Vert\rightarrow\infty$. For $i,j\in\{\tau+1,\ldots,n\}$, by Taylor expansion and some calculations, we have, $\hat{\U}_i^\top \hat{\U}_j=\U_i^\top \U_j+O_p(\Vert\boldsymbol\delta\Vert_\infty n^{-1/2}+\Vert\boldsymbol\delta\Vert_\infty p^{-1/2})$. For $i\in\{1,\ldots,\tau\}$ and $j\in\{\tau+1,\ldots,n\}$, by Taylor expansion, $\Vert\boldsymbol\delta\Vert^{-1}\Vert\boldsymbol\delta\Vert_\infty=o(p^{-1/2}n^{1/2})$ and $\Vert\boldsymbol\delta\Vert_\infty=o((n\wedge p)^{1/2})$, 
\begin{equation*}
    \begin{aligned}
        \hat{\U}_i^\top\hat{\U}_j=&\left\{\frac{\hat{\D}^{-1/2}(\X_i-\bth_0)+\hat{\D}^{-1/2}(\hat{\bth}_{1:n}-\bth_\kappa)+\hat{\D}^{-1/2}(\bth_\kappa-\bth_0-\boldsymbol\delta)+\hat{\D}^{-1/2}\boldsymbol\delta}{\Vert\hat{\D}^{-1/2}(\X_i-\bth_0)+\hat{\D}^{-1/2}(\hat{\bth}_{1:n}-\bth_\kappa)+\hat{\D}^{-1/2}(\bth_\kappa-\bth_0-\boldsymbol\delta)+\hat{\D}^{-1/2}\boldsymbol\delta\Vert}\right\}^\top\cdot\\
        &\quad \frac{\hat{\D}^{-1/2}(\X_j-\bth_0-\boldsymbol\delta)+\hat{\D}^{-1/2}(\hat{\bth}_{1:n}-\bth_\kappa)+\hat{\D}^{-1/2}(\bth_\kappa-\bth_0-\boldsymbol\delta)}{\Vert\hat{\D}^{-1/2}(\X_i-\bth_0-\boldsymbol\delta)+\hat{\D}^{-1/2}(\hat{\bth}_{1:n}-\bth_\kappa)+\hat{\D}^{-1/2}(\bth_\kappa-\bth_0-\boldsymbol\delta)\Vert} \\
    =&\left\{\frac{\hat{\D}^{-1/2}(\X_i-\bth_0)+\hat{\D}^{-1/2}(\hat{\bth}_{1:n}-\bth_\kappa)+\hat{\D}^{-1/2}(\bth_\kappa-\bth_0-\boldsymbol\delta)+\hat{\D}^{-1/2}\boldsymbol\delta}{\Vert\hat{\D}^{-1/2}(\X_i-\bth_0)+\hat{\D}^{-1/2}(\hat{\bth}_{1:n}-\bth_\kappa)+\hat{\D}^{-1/2}(\bth_\kappa-\bth_0-\boldsymbol\delta)+\hat{\D}^{-1/2}\boldsymbol\delta\Vert}\right\}^{\top}\cdot\\
        & \left\{\U_j+R_j^{-1}\D^{-1/2}(\hat{\bth}_{1:n}-\bth_\kappa)+R_j^{-1}\D^{-1/2}(\bth_\kappa-\bth_0-\boldsymbol\delta)\right\}\left\{1+O_p(\Vert\boldsymbol\delta\Vert_\infty n^{-1/2}+\Vert\boldsymbol\delta\Vert_\infty p^{-1/2})\right\} \\ 
        &\rightarrow \frac{(\hat{\D}^{-1/2}(\bth_\kappa-\bth_0))^\top \U_j}{\Vert\hat{\D}^{-1/2}(\bth_\kappa-\bth_0)\Vert}\left\{1+O_p(\Vert\boldsymbol\delta\Vert_\infty n^{-1/2}+\Vert\boldsymbol\delta\Vert_\infty p^{-1/2})\right\},
        \end{aligned}
\end{equation*}
as $\Vert\boldsymbol\delta\Vert\rightarrow\infty$.

Thus, we have
\begin{equation*}
    \begin{aligned}
        \frac{1}{\sqrt{2\tr(\R^2)}}S_{n,p}=&\max_{\lambda_n\le k\le n-\lambda_n}\left\{\tilde{\C}_{0}(k)^\top \tilde{\C}_{0}(k)-\frac{k(n-k)p}{n^2}\right\}\\
        \geq&\frac{1}{\sqrt{2\tr(\R^2)}}\left\{\tilde{\C}_{0}(\tau)^\top \tilde{\C}_{0}(\tau)-\frac{\tau(n-\tau)p}{n^2}\right\}\\
        \asymp& \frac{1}{\sqrt{2\tr(\R^2)}}\left\{\frac{\tau^2(n-\tau)^2p}{n^3}-\frac{\tau(n-\tau)p}{n^2}\right\}\rightarrow\infty.
    \end{aligned}
\end{equation*}

By Theorem \ref{thm:Max-Sum}, the critical value $c_{S,\alpha}$ only depends on the significant level $\alpha$. Thus, $S_{np}>c_{S,\alpha}$ as $\Vert\boldsymbol\delta\Vert\rightarrow\infty$. Proposition \ref{prop:max-L2}-(ii) can be proved in the same way, thus the proof is omitted.

\subsection{Some useful lemmas}
\begin{lemma}\label{lemma_normaldecom}(Theorem 1.2.11 in \cite{muirhead2009aspects})
     Let $\X \sim N(\boldsymbol\mu, \mathbf\Sigma)$ with invertible $\mathbf\Sigma$, and partition $\X, \boldsymbol\mu$ and $\mathbf\Sigma$ as
$$
\X=\binom{\X_1}{\X_2}, \mu=\binom{\boldsymbol\mu_1}{\boldsymbol\mu_2} \text { and } \mathbf\Sigma=\left(\begin{array}{ll}
\mathbf\Sigma_{11} & \mathbf\Sigma_{12} \\
\mathbf\Sigma_{21} & \mathbf\Sigma_{22}
\end{array}\right).
$$
Then, $\X_2-\mathbf\Sigma_{21} \mathbf\Sigma_{11}^{-1} \X_1 \sim N\left(\boldsymbol\mu_2-\mathbf\Sigma_{21} \mathbf\Sigma_{11}^{-1} \boldsymbol\mu_1, \mathbf\Sigma_{22 \cdot 1}\right)$ and is independent of $\X_1$, where $\mathbf\Sigma_{22 \cdot 1}=$ $\mathbf\Sigma_{22}-\mathbf\Sigma_{21} \mathbf\Sigma_{11}^{-1} \mathbf\Sigma_{12}$.
\end{lemma}

\begin{lemma}\label{LemmaA1}
    (Lemma A1 in \cite{cheng2023statistical}) Suppose that Assumptions \ref{ass:max1}--\ref{ass:max3} hold. Then, for sufficient large $p$, there exists positive constant $c_1$ and $c_2$ such that,
    \begin{equation*}
        \pr\left\{ 
p-\epsilon p^{(1+\eta_0)/2}\leq \Vert \W_i\Vert^2\leq p+ \epsilon p^{(1+\eta_0)/2}\right\}\geq 1-c_1\exp\left\{  -c_2 p^{\eta_0\alpha_0/(4\alpha_0+4)}\right\},
    \end{equation*}
    and
    \begin{equation*}
        \pr\left\{ 
(1-\epsilon)\tr(\R)\leq \Vert \D^{-1/2}\Gamma\W_i\Vert^2\leq (1+\epsilon)\tr(\R)\right\}\geq 1-c_1\exp\left\{  -c_2 p^{\eta_0\alpha_0/(4\alpha_0+4)}\right\}.
    \end{equation*}
    for any fixed $0<\epsilon<1$.
\end{lemma}
\begin{lemma}\label{LemmaA2}
    (Lemma A2 in \cite{cheng2023statistical}) Suppose that Assumptions \ref{ass:max1}--\ref{ass:max3} hold. Then, for any $i=1,2,\ldots,n$,
    
    (i) $\E\left(\left\|U_i\right\|^4\right)=p \E\left(U_{i, j}^4\right)+p(p-1)$,

$$
\begin{aligned}
\E\left(\left\|\W_i\right\|^6\right)= & p \E\left(W_{i, j}^6\right)+3 p(p-1) \E\left(W_{i, j}^4\right)+p(p-1)(p-2), \\
\E\left(\left\|\W_i\right\|^8\right)= & p \E\left(W_{i, j}^8\right)+4 p(p-1) \E\left(W_{i, j_1}^6\right)+3 p(p-1)\left\{\E\left(W_{i, j_1}^4\right)\right\}^2 \\
& +3 p(p-1) \E\left(W_{i, j}^4\right)+p(p-1)(p-2)(p-3).
\end{aligned}
$$

In addition, $\E\left(\left\|\W_i\right\|^{2 k}\right)=p^k+O\left(p^{k-1}\right)$ and $\E\left(\|\W\|^k\right)=p^{k / 2}+O\left(p^{k / 2-1}\right)$ for any positive integer $k$.

(ii) $\E\left(\left\|\D^{-1/2}\Gamma \W_i\right\|^4\right)=p^2+O\left(p^{2-\eta_0}\right), \E\left(\left\|\D^{-1/2}\Gamma \W_i\right\|^6\right)=p^3+O\left(p^{3-\eta_0}\right)$. In addition, $\E\left(\left\|\D^{-1/2}\Gamma \W_i\right\|\right)=p^{1 / 2}+$ $O\left(p^{1 / 2-\eta_0}\right)$ and $\E\left(\left\|\D^{-1/2}\Gamma \W_i\right\|^3\right)=p^{3 / 2}+O\left(p^{3 / 2-\eta_0}\right)$.

(iii) $\E\left\{\left\|\D^{-1/2}\Gamma U\left(\W_i\right)\right\|^2\right\}=1+O\left(p^{-1 / 2}\right)$ and $\E\left\{\left\|\D^{-1/2}\Gamma U\left(\W_i\right)\right\|^4\right\}=1+O\left(p^{-1 / 3}\right)$.

(iv) $\E\left(\nu_i^{-k}\right) \lesssim \zeta_k p^{k / 2}$ for $k=1,2,3$.
\end{lemma}
\begin{lemma}\label{LemmaA4}
    (Lemma A4. in \cite{cheng2023statistical}) Suppose Assumptions \ref{ass:max1}--\ref{ass:max3} hold. Then,

(i) $\E\{(\zeta_1^{-1} U_{i,j})^4\}\lesssim\bar M^2$ and $\E\{(\zeta_1^{-1} U_{i,j})^2\}\gtrsim \underline{m}$ for all $i=1,2,\ldots,n$ and $j=1,2,\ldots,p$.

(ii) $\Vert \zeta_{1}^{-1} U_{i,j}\Vert_{\psi_{\alpha_0}}\lesssim \bar B$ for all $i=1,2,
\ldots,n$ and $j=1,2,\ldots,p$.

(iii )$\E( U_{i,j}^2)=p^{-1}+O(p^{-1-\eta_0/2})$ for $j=1,2,\ldots,p$ and $\E( U_{i,j} U_{i,l})=p^{-1}\sigma_{j,l} +O(p^{-1-\eta_0/2})$ for $1\leq j\not=l\leq p$.

(iv) if $\log p=o(n^{1/3})$,
$$
\left\vert n^{-1/2}\sum_{i=1}^n\zeta_1^{-1}\boldsymbol U_i\right\vert_\infty =O_p\{\log^{1/2}(np)\}\text{ and } \left\vert n^{-1} \sum_{i=1}^n(\zeta_1^{-1}\boldsymbol U_i)^2\right\vert_\infty=O_p(1).
$$
\end{lemma}

\begin{lemma}\label{lemma:for_Uis}
     Under Assumption \ref{ass:sum_R4}, we have\\
    (i) $\E (\U_1^\top \U_2)^4=O(1)\E^2(\U_1^\top \U_2)^2$; \\
    (ii) $\E(\U_1^\top\Sigma_w \U_2^2)=O(1)\{\E(\U_1^\top \Sigma_w \U_1)\}^2$;\\
    (iii) $\E(\U_1^\top \Sigma_w \U_2)^2=o(1)\{\E(\U_1^\top \Sigma_w \U_1)\}^2$;  furthermore, \\
    (iv) $\E(\U_1^\top \Sigma_w \U_2)^2=O(n^{-1+2\omega_1})\{\E(\U_1^\top \Sigma_w \U_1)\}^{2}$ for some $0<\omega_1<1/4$.
\end{lemma}
\begin{proof}
    See the proof of Lemma 1 in \cite{wang+peng+li-2015high} and replace some equations by Equation \eqref{eq:UW}.
\end{proof}

\begin{lemma}\label{lemma2_dj}
    (Lemma 2 in \cite{liu+feng+wang+2024}) Under Assumption \ref{ass:max1} and \ref{ass:max3} (iv), we have, $\max_{1\leq j\leq p}(\hat{d}_{a:b,j}-d_j)=O_p\{(b-a)^{-1/2}(\log p)^{1/2}\}$, as $b-a\rightarrow\infty$.
\end{lemma}
\begin{lemma}\label{LemmaE.1central}
    (Lemma E.1 in \cite{chernozhukov2017central}) Let $\X_1,\X_2,\ldots,\X_n$ be independent centered random vectors in $\mathbb R^p$ with $p\geq 2$. Define $Z:=\max_{1\leq j\leq p}\left\vert \sum_{i=1}^n X_{ij} \right\vert$, $M:=\max_{1\leq i\leq n}\max_{1\leq j\leq p}\left\vert X_{ij}\right\vert$ and $\sigma^2:=\max_{1\leq j\leq p}\sum_{i=1}^n \E (X_{ij}^2)$. Then,
    \begin{equation*}
        \E (Z)\leq K\left( \sigma\sqrt{\log p}+\sqrt{\E (M^2)}\log p
 \right),
    \end{equation*}
    where $K$ is a universal constant.
\end{lemma}
\begin{lemma}\label{lemmaS.10_feng}
    (Lemma S.10 in \cite{Feng2022AsymptoticIO}) Let $\{(U,U_p,\tilde{U}_p)\in \mathbb R^3;p\geq 1\}$ and $\{(V,V_p,\tilde{V}_p)\in \mathbb R^3;p\geq 1\}$ be two sequences of random variables with $U_p\rightarrow U$ and $V_p\rightarrow V$ in distributions as $p\rightarrow\infty$. Assume $U$ and $V$ are continuous random variables and 
    \begin{equation*}
        \tilde{U}_p=U_p+o_p(1)\text{ and }\tilde{V}_p=V_p+o_p(1).
    \end{equation*}
    If $U_p$ and $V_p$ are asymptotically independent, then $\tilde{U}_p$ and $\tilde{V}_p$ are also asymptotically independent.
\end{lemma}

\bibliographystyle{apa}
\bibliography{ref}
\end{document}